\documentclass[11pt]{article}
\usepackage[margin=1in]{geometry}

\usepackage[dvipsnames]{xcolor}
\usepackage{amsmath}
\usepackage{amsthm}	
\usepackage{amsfonts}
\usepackage{amssymb}
\usepackage{bbm}
\usepackage{mathtools}

\usepackage{booktabs} 
\usepackage{caption} 

\usepackage{subcaption} 
\usepackage{graphicx}
\usepackage{pgfplots}
\usepackage[all]{nowidow}
\usepackage[utf8]{inputenc}
\usepackage{tikz}

\usepackage{multicol}
\usepackage{algpseudocode,algorithm,algorithmicx}

\usepackage{hyperref}

\usepackage{scalerel}
\usepackage{natbib}

\usepackage[english]{babel}

\usepackage{yhmath}
\usepackage{cancel}

\usepackage{siunitx}
\usepackage{array}
\usepackage{multirow}
\usepackage{gensymb}
\usepackage{tabularx}
\usepackage{booktabs}

\usetikzlibrary{fadings}
\usetikzlibrary{patterns}
\usetikzlibrary{shadows.blur}

\newtheorem{theorem}{Theorem}
\newtheorem{lemma}[theorem]{Lemma}
\newtheorem{corollary}[theorem]{Corollary}

\newtheorem{assumption}{Assumption}
\newtheorem{proposition}[theorem]{Proposition}
\newtheorem{remark}{Remark}
\newtheorem{definition}{Definition}

\renewenvironment{proof}[1][\proofname]{{\noindent \bfseries\itshape #1.}}{\qed}

\newcommand{\onep}{\mathbbm{1}_{t_{k+1}}^+}
\newcommand{\onem}{\mathbbm{1}_{t_{k+1}}^-}
\newcommand{\onepj}{\mathbbm{1}_{t_{j+1}}^+}
\newcommand{\onemj}{\mathbbm{1}_{t_{j+1}}^-}

\newcommand{\onedj}{\mathbbm{1}_{t_{j+1}}^{\delta}}
\newcommand{\onen}{\mathbbm{1}_{t_{k+1}}}

\newcommand{\cp}{c_{t_{k+1}}^{+}}
\newcommand{\cm}{c_{t_{k+1}}^{-}}
\newcommand{\cpj}{c_{t_{j+1}}^{+}}
\newcommand{\cmj}{c_{t_{j+1}}^{-}}

\newcommand{\cdj}{c_{t_{j+1}}^{\delta}}

\newcommand{\Pp}{p_{t_{k+1}}^{+}}
\newcommand{\Pm}{p_{t_{k+1}}^{-}}
\newcommand{\Ppj}{p_{t_{j+1}}^{+}}
\newcommand{\Pmj}{p_{t_{j+1}}^{-}}

\newcommand{\Pdj}{p_{t_{j+1}}^{\delta}}
\newcommand{\lp}{L_{t_{k}}^{+}}
\newcommand{\lm}{L_{t_{k}}^{-}}
\newcommand{\lpj}{L_{t_{j}}^{+}}
\newcommand{\lmj}{L_{t_{j}}^{-}}

\newcommand{\ldj}{L_{t_{j}}^{\delta}}
\newcommand{\skk}{S_{t_{k+1}}}
\newcommand{\sk}{S_{t_{k}}}
\newcommand{\sjj}{S_{t_{j+1}}}
\newcommand{\sj}{S_{t_{j}}}
\newcommand{\tkk}{{t_{k+1}}}
\newcommand{\tk}{{t_{k}}}
\newcommand{\tjj}{{t_{j+1}}}
\newcommand{\tj}{{t_{j}}}
\newcommand{\tu}{{t_{u}}}

\newcommand{\F}{\mathcal{F}_{t_{k}}}
\newcommand{\Fj}{\mathcal{F}_{t_{j}}}
\newcommand{\Gt}{\mathcal{G}_{t_{k+1}}}
\newcommand{\Gtj}{\mathcal{G}_{t_{j+1}}}
\newcommand{\Hkkk}{\mathcal{H}_{t_{k+1}}^\varpi}
\newcommand{\Hkk}{\mathcal{H}_{t_k}^\varpi}
\newcommand{\Hjj}{\mathcal{H}_{t_j}^\varpi}

\newcommand{\bx}{\mathbf{\hat{x}}}
\newcommand{\by}{\mathbf{\hat{y}}}
\newcommand{\bz}{\mathbf{\hat{z}}}

\newcommand{\FKK}{\mathcal{F}_{t_{k+1}}}
\newcommand{\E}{\mathbb{E}}
\newcommand{\muo}{\mu_c}
\newcommand{\muoo}{\mu_{cp}}
\newcommand{\mut}{\mu_{c^2}}
\newcommand{\muto}{\mu_{c^2p}}
\newcommand{\mutt}{\mu_{c^2p^2}}
\newcommand{\Aone}{{}^{\scaleto{(1)}{5pt}}\!A_{\tk}}

\newcommand{\etk}{\pmb{e}_\tk}
\newcommand{\etkk}{\pmb{e}_\tkk}

\newcommand{\R}{\mathbb{R}}

\newcommand{\A}{\mathcal{A}}

\newcommand{\Ex}{\mathbb{E}}
\newcommand{\Px}{\mathbb{P}}
\newcommand{\sbt}{\, \begin{picture} (-1,1)(-1,-3)\circle*{3}\end{picture}\ }

\definecolor{blue}{HTML}{1F77B4}
\definecolor{orange}{HTML}{FF7F0E}
\definecolor{green}{HTML}{2CA02C}

\definecolor{Red}{rgb}{1,0,0}

\definecolor{DRed}{rgb}{0,0,0}
\newcommand{\DRed}{\color{DRed}}
\definecolor{Green}{rgb}{0.2,0.5,0.2} 

\definecolor{Blue}{rgb}{0,0,0}
\newcommand{\Blue}{\color{Blue}}
\definecolor{Orange}{rgb}{0.605,0.311,0.084} 
\definecolor{Black}{rgb}{0,0,0} \newcommand{\Black}{\color{Black}}

\definecolor{Grey}{rgb}{0.52, 0.52, 0.51} 
\definecolor{Purple}{rgb}{0.333,0.036, 0.631} \newcommand{\Purple}{\color{Purple}}

\definecolor{Aqua}{rgb}{0.018,0.592,0.39}

\usepackage{setspace}
\begin{document}

\title{Adaptive Optimal Market Making Strategies with Inventory Liquidation Cost\footnote{A preprint of this paper was distributed under the title of ``Market Making with Stochastic Liquidity Demand: Simultaneous Order Arrival and Price Change Forecasts". The present paper extends the results in the referred preprint,  which will remain as an unpublished manuscript.}}
%
\author{Jonathan Ch\'{a}vez-Casillas\thanks{Department of Mathematics and Applied Mathematical Sciences, University of Rhode Island, USA ({\tt jchavezc@uri.edu}).}, Jos\'e E. Figueroa-L\'opez\thanks{Department of Statistics and Data Science, Washington University in St. Louis, St. Louis, MO 63130, USA ({\tt figueroa-lopez@wustl.edu}). Research supported in part by the NSF Grants: DMS-2015323, DMS-1613016.}, Chuyi Yu\thanks{Department of Mathematics, Washington University in St. Louis, St. Louis, MO 63130, USA ({\tt chuyi@wustl.edu}).}, and  Yi Zhang\thanks{Department of Statistics, UIUC, IL, USA ({\tt yiz19@illinois.edu}).}}
%
%
%
\maketitle 
\begin{abstract}
A novel high-frequency market-making approach in discrete time is proposed that admits closed-form solutions. By taking advantage of demand functions that are linear in the quoted bid and ask spreads with random coefficients, we model the variability of the partial filling of limit orders posted in a limit order book (LOB). As a result, we uncover new patterns as to how the demand's randomness affects the optimal placement strategy. We also allow the price process to follow general dynamics without any Brownian or martingale assumption as is commonly adopted in the literature. The most important feature of our optimal placement strategy is that it can react or adapt to the behavior of market orders online. Using LOB data, we train our model and reproduce the anticipated final profit and loss of the optimal strategy on a given testing date using the actual flow of orders in the LOB. Our adaptive optimal strategies outperform the non-adaptive strategy and those that quote limit orders at a fixed distance from the midprice.
\end{abstract}

\section{Introduction}\label{ch1}

\subsection{Overview}
In a financial market, a market maker (MM) provides liquidity to the market by repeatedly placing bid and ask orders into the market and profiting from the bid-ask spread of her orders. {\Blue The literature of market making is extensive (see, e.g., the monographs of \cite{algo} and \cite{GueantBook} for references in the subject). In Subsection \ref{OtherRelWrk} below, we mention a few important works in addition to those more closely related to our model, which are reviewed through this part.} 
In the context of a Limit Order Book based market, we consider an intraday high-frequency MM who quotes both bid and ask limit orders {(LO)} at some prespecified discrete times and liquidates her inventory at the end of the trading period. {As is} often assumed in the literature, the terminal liquidation cost {\Blue or price impact}, originated {\DRed from} the use of a market order, is modeled as $I_T(S_T-\lambda I_T)$, with $S_T$ and $I_T$ {\DRed respectively} denoting the final fundamental stock price and {the} MM's inventory. Here, $\lambda$ is a constant `penalization' parameter. We aim to maximize the final Profit and Loss (PnL), $W_T+I_T(S_T-\lambda I_T)$, at the end of the trading period $T$, where $W_T$ is the MM's final wealth. Her wealth and inventory trajectory are determined by the prices of her quotes and the number of shares that are filled or lifted from her orders at these prices.

Modeling the number of lifted shares between consecutive actions is a key element of our framework. In continuous-time control problems, a common approach is to model the probability with which an incoming market order (MO) can lift one share of the MM's LO in the book (known as `lifting probability'). This approach, rooted in the seminal work of  \cite{Ho_Stoll_1981}, was popularized by the work of \cite{avellaneda2008high} and later on by other important works including \cite{gueant2013dealing},  \cite{cartea2014buy}, among others. 
For instance, in the seminal work of \cite{cartea2015risk}, it is assumed that MOs arrive according to a Poisson process and the lifting probability is modeled as the exponential of the negative distance of the MM's quote from the fundamental price times a constant. 

{An alternative approach
is to directly model the number of lifted} shares {between actions} via a liquidity demand function.  
For instance, in their work on price pressures, \cite{hendershott2014price} assume that the liquidity demand is normally distributed with a mean {parameter} that is linear in the bid and ask {\DRed quoted prices} {and constant variance}. 
{
\cite{adrian2016intraday}  propose a demand function that decreases linearly with the \emph{distance} of the quotes from a reference price, though their demand function is further restricted to be deterministic.} {\DRed We refer the reader to Remark \ref{CnctCntDscr} below for some further discussion about possible connections between the lifting probabilities- and the demand function-based approaches.}

Our work extends existing models of high-frequency market-making in several ways. As in  \cite{adrian2016intraday}, 
we assume the demand to be linear {\DRed in the spreads} when modeling the number of filled shares from the MM's limit orders. However, {in our case, the demand} is not {\it deterministic but {stochastic}}. This means that the actual number of shares bought or sold varies over time, even if the distances of quotes from the reference stock price stay the same. The resulting optimal placement strategy does not boil down to simply replacing the constant demand slope and reservation price in the optimal strategy obtained in \cite{adrian2016intraday} with their respective average values, but also depends on their `second-order' information and their mutual correlation. {The proposed randomization not only allows for greater flexibility and better fit to empirically observed order flows, but also uncovers novel properties of the resulting optimal placement strategies. For instance, it is known from \cite{adrian2016intraday} that, under a constant demand slope, the {required} inventory adjustment in the optimal placement at any given time decreases with the size of the slope. We show that the variance of the slope further reduces {\DRed the strength of this} adjustment. This implies that assets with more volatile demand profiles require less strict inventory adjustments. We also find that the optimal placement spreads (i.e., the distances between the optimal bid and ask prices and the fundamental price) increase with the correlation between the demand slope and investors' reservation price. 

Another distinguishing feature of our study is that we allow {a general reference stock price process without assuming that its dynamics follow that of a martingale or some kind of parametric specification (e.g., a Brownian It\^o semimartingale)}. We obtain a parsimonious formula that describes how the investor should adjust {online her LO placements based on her ongoing forecasts of future asset prices}. Intuitively, if the MM expects future price changes to be negative, {\Blue she will reduce the ask spread and increase the bid spread, proportionally to the expected price change}. The proportionality constant depends on the model parameters in a non-trivial way, which we characterize precisely. This feature could also enable the investor to take advantage of sophisticated time series or machine learning-based forecast procedures {for} asset prices {and incorporate them} into the intraday market-making process.

One of the key factors that affects the performance of any placement strategy is the arrival intensity of MOs on either side of the book. A {\DRed successful} placement strategy should incorporate, in an {online manner},  current information about the intensity level of MOs or about imbalances in the likelihood of {buying and selling MOs}. In other words, it is desirable that the placement strategy {adapts} to the local behavior of MOs. In {many} works, the intensity of MOs is fixed a priori as a deterministic function of time.
{In} reality, the intensity of MO is highly random and `rough' as indicated by Fig.~\ref{fig1111} below, where the {averages} of the indicators signaling {the} arrivals of buy MOs in a rolling window are plotted. However, as {\DRed shown} by the same figure, the intensity's level can be tracked or predicted quite well in an adaptive or online manner (see {\DRed Subsection \ref{EstPiSS0} for details as to how to perform this prediction}) {and it would be desirable that the optimal strategy incorporates this information {\Blue online}}. 
 
\begin{figure}[h]
	\centering
	\vspace{.5 cm}
	\includegraphics[width=0.8\textwidth]{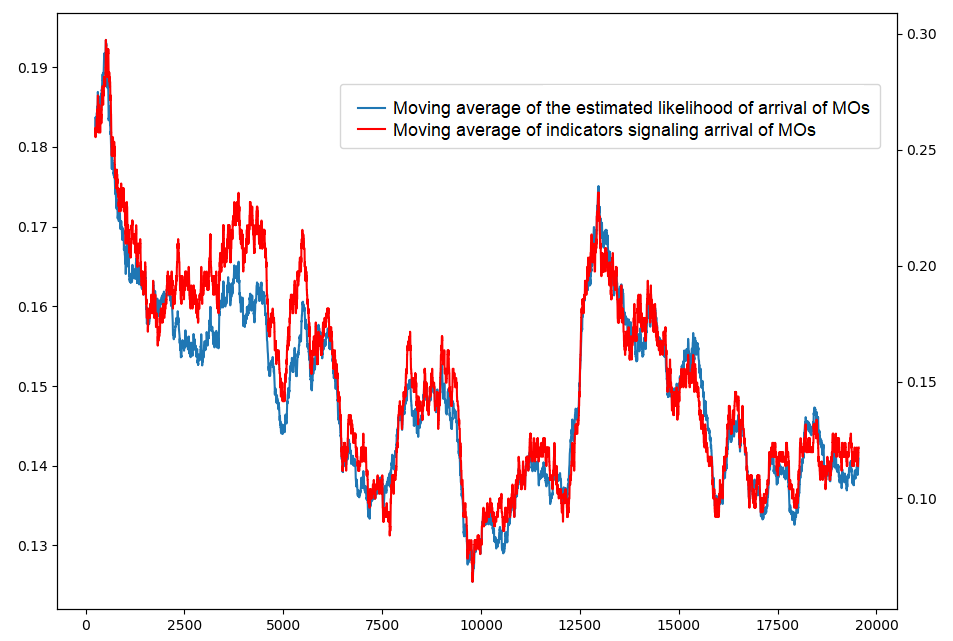}
	\vspace{.05 cm}
	\caption{Moving average of adaptive arrival probabilities of buy MO and the moving average of the indicators signaling arrival of MOs in consecutive intervals of 1 second based on LOB MSFT data on July 11th, {\Blue 2019}. {The window size of the rolling moving average is 500. {\Blue Note that the scale for the blue (red) curve is shown on the left (right) side of the figure.}}}\label{fig1111}
	\vspace{.5 cm}
\end{figure}

By construction, it would seem that adaptive placement strategies are not feasible since the dynamic programming problem to find them is solved in a backward manner in time, which contradicts the direction of a natural learning process that proceeds forward in time. However, in this work, we offer a natural approach to {resolve} this riddle. Essentially, we propose to create a `catalog' of optimal strategies depending not only on the current asset price and inventory level as well as future price forecasts but also on the recent history of MOs arrivals. In some figurative sense, we create parallel `universes', one for each possible combination or scenario of past MOs events, and solve the optimal placement in each of those universes. We do this by making the conditional probabilities of arrivals of MOs dependent on the recent history of MOs. When implementing the placements strategy, the MM observes the recent history or combinations of MOs to determine in which `universe' or scenario she is in, and places her LOs {accordingly} using the catalog of optimal strategies. 

To put our proposed approach to the test, we implement our optimal placement strategy using actual LOB data. Specifically, for a given testing day, we start calibrating the model parameters using LOB data of the past few days and then compute, backward in time, the optimal placement strategy for each possible scenario of consecutive MOs. Next, we roll forward our optimal placement strategy using the actual LOB events of the testing day to determine in an online {manner} the scenario we are at and choose the optimal placements accordingly. We then compute the MM's cash flows and inventory changes using the actual flow of MOs {and {the LOB} state}. {At the end of the trading period}, the MM submits a MO to liquidate its final inventory and determine the actual cost taking into account the state of the LOB. We repeat this procedure for each day {of a} 1-year time span. 
We find that our optimal placement yields, on average, larger revenue compared to those where the intensity of MO is assumed to be deterministic (time-dependent). Our empirical analysis also lends strong support to demand stochasticity: the slope coefficient has a standard deviation that is about 200\% larger than the average demand level, and a correlation of about 20\% with the investors' reservation price. Moreover, using real LOB data we estimate the optimal placement strategy based on a simple one-step ahead {price process} forecast {\DRed and compare it to} the one that presumes a martingale price evolution.

\subsection{Other relevant works}\label{OtherRelWrk}
{\Blue Optimal market making problems have a long history. In this part we mention a few important works that have not previously been discussed. Early contributions include those of \cite{Bradfield}, who analyze the increasing price variability induced by strategies that target a flat end-of-day inventory level, and \cite{OharaOldfield}, who consider a repeated optimal market making problem, in which each day consists of several trading periods, and the market maker maximizes utility over an infinite number of trading days while facing end-of-day inventory costs. More recently, \cite{Guilbaud} studied the performance of a MM submitting buy/sell LOs at the best two bid and ask levels, while \cite{Guilbaud2} also considered agents that can submit market orders. Both of these works assume a constant spread of one tick (the so-called large-tick stocks) and price dynamics that move 1 tick at a time. Other works in the same vein include \cite{Fodra} and \cite{Fodra2}. In all these works, there is no partial filling and the MM's LOs are filled in its totality when they are lifted. 
 
Stochastic demand functions of different type have been considered in other works. As mentioned above, both \cite{cartea2015risk} and \cite{cartea2014buy} modeled the demand during a given time interval via filling probabilities, which depend on the distance between the quotes and the fundamental price. While in the first of those two works, the features of these probabilities are assumed to be deterministic, in the second work, those are assumed to be stochastic. As explained in Remark \ref{CnctCntDscr} below, there is a possible connection between the linear demand assumption adopted in this work and the fill probabilities approach, but, in general, the two models are not equivalent. It is worth mentioning that liquidity models with stochastic features have also been considered in other problems of algorithmic trading. For instance, \cite{Lorig} and \cite{Berecher} both assume stochastic price impact of trades in optimal liquidation problems.

More recent works in the area have also incorporated other model features such as price impact (\cite{cartea2014buy}, \cite{Lorig}), model ambiguity (\cite{carteaDonnellyetal}, \cite{Nystrom}), and latency (\cite{CarteaLatency}, \cite{Gao}). More recently, \cite{Bergault} introduce a different modeling approach to incorporate different transaction sizes and the possibility for the MM to respond to requests with different sizes using marked point processes}.

\subsection{Outline of the paper}
{The rest of the paper is organized as follows. In Section 2, we present the model setup {and our assumptions together with the Bellman equation for our problem and its explicit solutions}. In Section 3, we assess the performance of our market-making strategy against real LOB data. {\Blue We finish with a conclusion section.} We {defer the proofs to an appendix section}.}

\section{A Finite-Horizon Optimal Control Problem for a Market Maker} \label{ch2}

{In this section, 
we introduce the model along with the relevant notation and its assumptions. Then, by using the Dynamic Programming Principle (DPP), we propose an adaptive trading strategy and by using the verification theorem, it is shown that, indeed, the  solution is optimal. Finally, its admissibility is {investigated}. 
All the proofs of this section will be deferred to Appendix \ref{appen:section:2}}.

\subsection{The Model and its Assumptions}

{We assume {that} a high-frequency (HF) Market Maker (MM)} places simultaneously buy and sell LOs at {some preset discrete} times $0=t_0<t_1<\ldots<t_N$, where $t_N<t_{N+1}:=T$ {and hereafter $T$ represents the terminal time of the trading}. {All} the random variables used in the model are defined on the same probability space $(\Omega,\mathcal{F},\mathbb{P})$ equipped with {an information} filtration $\{\mathcal{F}_t\}_{t\in\mathcal{T}}$, where $\mathcal{T} = \{t_0,t_{1},\ldots, t_{N+1}\}$. { For $k=0,\dots,N$, let {$\onep\in\mathcal{F}_{t_{k+1}}$} ({$\onem\in\mathcal{F}_{t_{k+1}}$}) {be Bernoulli random} variables} indicating whether at least one buy (sell) MO arrived during the time period $[t_k,t_{k+1})$, i.e.,
\begin{equation}\label{eq:MOind}
\begin{aligned}
\onep &= \mathbbm{1}_{\{\text{At least one buy MO arrives during }[t_k,t_{k+1})\}},\\
\onem &= \mathbbm{1}_{\{\text{At least one sell MO arrives during }[t_k,t_{k+1})\}}.
\end{aligned}
\end{equation}
Let $\varpi$ be a fixed positive integer and define the lag-$\varpi$ recent history of MOs at time {$t_k$ as
\begin{equation}\label{eq:defn:etk}
\begin{aligned}
\pmb{e}_{t_{k}}&=(\mathbbm{1}_{t_{k}}^+,\mathbbm{1}_{t_{k}}^-,
\dots,\mathbbm{1}_{t_{k-\varpi+1}}^+,\mathbbm{1}_{t_{k-\varpi+1}}^-) \in \{0,1\}^{2\varpi}.
\end{aligned}
\end{equation} 
We often use the shorthand notation $\pmb{e}_{t_{k}}=(\mathbbm{1}_{t_{k}}^{\pm},\dots, \mathbbm{1}_{t_{k-\varpi+1}}^{\pm})$. 
Above, we are assuming that $t_0=0$ is the beginning of the MM's trading {and that there} is a sufficiently large ``burn-out'' period prior to it. In particular, the indicators (\ref{eq:MOind}) are also defined for $k=0,-1,\dots$ by setting some times $0>t_{-1}>t_{-2}>\dots$ before the beginning of the trading at time $0$. Thus, for instance, $\mathbbm{1}_{t_{0}}^{\pm}$ is $1$ if at least one MO arrived during the time period $(t_{-1},0]$ and $\pmb{e}_{t_{0}}=(\mathbbm{1}_{t_{0}}^{\pm},\mathbbm{1}_{t_{-1}}^{\pm},\dots,\mathbbm{1}_{t_{-\varpi+1}}^{\pm})$ represents the indicators of MOs in the {\Blue $\varpi$} time periods previous to $0$}.
For future reference, let us introduce some notation. For  any adapted process $u=\{u_{t_k}\}_{k\geq{}0}$, {we set} 
\begin{equation}\label{Dfnalh0}
\begin{aligned}
\begin{array}{ll}
{u}_\tkk^0 :=\mathbb{E}\left[\left.{u}_\tkk \,\right|\,\F\right], \\
{{u}_\tkk^{i\pm} :=\mathbb{E}\left[\left.{u}_\tkk \,\right|\,\F,\mathbbm{1}^\pm_\tkk=i\right], \quad i\in\{0,1\}},\\
{{u}_\tkk^{i,j} :=\mathbb{E}\left[\left.{u}_\tkk \,\right|\,\F,\mathbbm{1}^+_\tkk=i,\mathbbm{1}^-_\tkk=j\right],\quad i,j\in\{0,1\}},
\end{array}
\end{aligned}
\end{equation}
{which can be considered some {sort} of one-step ahead forecasts of $u$}. Note that all these processes are adapted to the information process $\{\mathcal{F}_t\}_{t\in\mathcal{T}}$.

As mentioned before, at each time $t_k$, the {MM} will place simultaneously a buy and a sell LO. {\Blue Her sell} LO will be submitted with an execution price of $a_{t_k}${,} while {\Blue her buy} LO will have an execution price of $b_{t_k}$. The volume of these orders {is typically set to be} the average volume {of submitted LO in} the stock of interest. 
Both $a_{t_k}$ and $b_{t_k}$ are the MM's `controls'. It {will} be important to reparameterize the controls relative to a reference price $S_{t_k}$ associated with the stock such as the midprice or other related proxy. 
The MM's buy (sell) LOs submitted at time $t_k$ ($k=0,\dots,N$) will be matched against the sell (buy) MOs submitted during the period $[t_k,t_{k+1})$ as follows. Let us first consider the ask side. If $S_\tk$ denotes the reference {or fundamental} price of the stock at time $t_k$ and the {MM}'s ask LO is placed $L_\tk^+$ above {$S_{t_k}$} (i.e.{,} $a_{t_k}=S_{t_k}+L_\tk^+$), then the total number of shares {from the {MM} that are {sold} during $[t_k,t_{k+1})$ is denoted by $Q^+_{t_{k+1}}$ and is given {by}} 
\begin{equation}\label{pdemand}
\begin{aligned}
Q^+_{t_{k+1}}&:= \onep\cp\big[(\sk+\Pp)-(\sk+\lp)\big] =\onep\cp(\Pp-\lp),
\end{aligned}
\end{equation}
where {\Blue $\Pp,\cp \in \mathcal{F}_{\tkk}$ are nonnegative random variables. Broadly (but not literally) $\Pp$ is related to the maximum depth that buy MOs {walk} into the LOB {during $[t_k,t_{k+1})$} and $\cp$ is such that $\cp\Pp$ indicates the executed volume of a sell LO if this {were} placed at the {same level as the} reference price (that is, if $\lp=0$).} 
{We can also interpret $p^+_{t_{k+1}}$ as the buyers' reservation price in the market; i.e., the highest price that `buyers' in the market are willing to pay for the stock.} {\Blue These interpretations should not be taken literally as explained in points 2-4 of Remark \ref{InterpDemnd} below.}

We analogously define the {corresponding} quantities for the bid side of the book. That is, provided that at least one sell MO arrives during the time interval $[t_k,t_{k+1})$, $Q^-_\tkk$ will be the executed volume of the {MM's} buy LO placed at time $t_{k}$ at the price level $b_\tk$. {{Similarly to $a_\tk$}, we reparameterize $b_\tk$ in terms of the {\DRed distance} $L_\tk^-$ below the reference price $S_\tk$ {so that} $b_\tk=S_\tk-L_\tk^-$}. Similarly to \eqref{pdemand}, $Q^-_\tkk$ is {modeled as}
\begin{equation}\label{mdemand}
\begin{aligned}
Q^-_\tkk&:= \onem\cm\big[(\sk-\lm)-(\sk-\Pm)\big] =\onem\cm(\Pm-\lm),
\end{aligned}
\end{equation}
where $\Pm \in \mathcal{F}_{\tkk}$ and $\cm \in \mathcal{F}_{\tkk}$ has analogous interpretations as  $\Pp$ and $\cp$. Above, both $L^{+}_{t_k}$ and $L^{-}_{t_k}$ are the {MM `controls', while the reference price $S_{t_k}$ is exogenously determined by market conditions {independently from} the MM actions.} 
The {\Blue form of the function of} $Q^\pm_{\tkk}$ is illustrated in Figure \ref{fig:linearQ}. 

\begin{remark}\label{InterpDemnd}
	Some comments are in order to clarify our model assumptions \eqref{pdemand}-\eqref{mdemand} and contrast to earlier work:
	\begin{enumerate}
	\item  {\Blue As mentioned in the introduction, in a continuous-time setting, \cite{adrian2016intraday} considered demand functions similar to \eqref{pdemand}-\eqref{mdemand}, but with $c$ and $p$ being known deterministic constants. Then, the actual numbers of shares bought or sold over $[t_k,t_{k+1})$ depend only on the spreads $L_{t_{k}}^\pm$. Introducing randomness is more realistic since the actual demand during $[t_k,t_{k+1}]$, not only depend on the spreads of the quotes, but also on the initial state of the book, which is hard to summarize and incorporate given its high-dimensionality and variability, and on the flow of orders during the interval, which is extremely unpredictable at time $t_k$. One may argue that for the purposes of market making all what matters in the average demand during the given trading period $[0,T]$. That is, all what we need is to take the deterministic demand functions $\mathbbm{1}_{t_{k+1}}^{\pm}\mu_{c}^\pm(\mu_{p}^{\pm} -L_{t_k}^{\pm})$ or $\mathbbm{1}_{t_{k+1}}^{\pm}(\mu_{cp}^{\pm} -\mu_{c}^\pm L_{t_k}^{\pm})$, where $\mu_{c}^{\pm}$, $\mu_{p}^{\pm}$, $\mu_{cp}^\pm$ are the average values of $c$, $p$, and $cp$ over $[0,T]$, respectively. As explained in the items 1 and 4 of page 18 and 19 (see formulas \eqref{WNTSqrI} and \eqref{DfnTldL0bcNFFbAp} therein), the randomness of $c$ and the correlation between $c$ and $p$ play key roles in the optimal MM strategy. This is further verified in our numerical/empirical Section \ref{sec:Data} (see Tables \ref{tab:TerminalRewardaba}-\ref{tab:TerminalRewardbb} therein).}
	\item As mentioned above, the actual demand during $[t_k,t_{k+1})$ corresponding to placements $L_{t_k}^{\pm}$ would depend on the shape of the book at time $t_k$ as well as the volume and timing of MMs and the arrival of other LOs and cancellations during that interval. Figure \ref{DemandPlotDemo} shows the approximate demand during a prototypical time interval where at least one MM order arrived (see Subsection \ref{DmdFncEsD} for details about how to estimate such a demand). Then, in a nutshell, $c^\pm_{t_{k+1}}$ and $p^\pm_{t_{k+1}}$ are chosen so that the resulting linear model $c^\pm_{t_{k+1}}(p^\pm_{t_{k+1}}-L^\pm_{t_{k+1}})$ is close to the `actual' demand. 
In other words, \eqref{pdemand}-\eqref{mdemand} are viewed as the {`best'} linear fit for the actual demand during a given time period.

	\item One of the obvious drawbacks of \eqref{pdemand}-\eqref{mdemand} is that, in principle, we are assuming the possibility of negative demands (negative number of units sold or bough), which is obviously not realistic. But, this would happen only if the LO quotes $L^{\pm}_{t_k}$ were large compare to $S_{t_{k}}$. This assumption would then be an issue if the resulting optimal placements were, at times, far away from the reference price, since, in that case, the optimal strategy may be favoring or looking for negative demands, which in reality are not feasible. In the context of Figure \ref{DemandPlotDemo}, this would be the case if the optimal placements were more than about 6 ticks away from $S_{t_k}$ that is when the linear demand functions start to produce negative values. However, our empirical implementation in Section \ref{sec:Data} shows us that the resulting optimal placements are almost never far away from the reference price (almost always less than or equal to 4 ticks away) and, thus, the linear demand assumption is not an issue in practice. 
	
	\item Above it was mentioned that $\Pp$ is connected to the maximum depth that the buy MOs {walk} into the LOB during $[t_k,t_{k+1})$. This is because, under this interpretation, we obviously have that if  $L^+_{t_{k}}>p^+_{t_{k+1}}$, then the corresponding demand should be $0$. However, as mentioned above, it is more accurate to see $p^+_{t_{k+1}}$ as the value for which $c^\pm_{t_{k+1}}(p^\pm_{t_{k+1}}-L^\pm_{t_{k+1}})$ provides a good fit for the demand when $L^\pm_{t_{k+1}}$ is small.
	\end{enumerate}
\end{remark}

\tikzset{every picture/.style={line width=0.75pt}} 
\begin{figure}
\centering
\tikzset{every picture/.style={line width=0.75pt}} 

\begin{tikzpicture}[x=0.75pt,y=0.75pt,yscale=-1,xscale=1]

\draw  (50,142.02) -- (418.54,142.02)(75.83,19.72) -- (75.83,166.22) (411.54,137.02) -- (418.54,142.02) -- (411.54,147.02) (70.83,26.72) -- (75.83,19.72) -- (80.83,26.72)  ;
\draw [line width=1.5]    (232.5,67.8) -- (99.07,142.02) ;
\draw [line width=1.5]    (233.5,45.8) -- (398.92,142.02) ;
\draw   (173.28,142.02) -- (168.05,142.02)(170.67,139.4) -- (170.67,144.63) ;
\draw   (316.52,142.02) -- (311.28,142.02)(313.9,139.4) -- (313.9,144.63) ;
\draw   (170.67,158) .. controls (170.62,158.15) and (172.93,160.51) .. (177.6,160.56) -- (192.15,160.7) .. controls (198.82,160.77) and (202.13,163.14) .. (202.08,165) .. controls (202.13,163.14) and (205.48,160.84) .. (212.15,160.91)(209.15,160.88) -- (226.71,161.06) .. controls (231.38,161.11) and (233.73,158.8) .. (233.78,156) ;
\draw   (234.44,156) .. controls (234.47,158.81) and (236.82,161.12) .. (241.49,161.08) -- (264.23,160.89) .. controls (270.9,160.84) and (274.25,163.14) .. (274.28,165) .. controls (274.25,163.14) and (277.56,160.78) .. (284.22,160.73)(281.22,160.75) -- (306.96,160.54) .. controls (311.63,160.51) and (313.94,158.16) .. (313.9,158) ;
\draw [line width=1.5]  [dash pattern={on 5.63pt off 4.5pt}]  (233,29.5) -- (232.82,142.67) ;
\draw  [dash pattern={on 0.84pt off 2.51pt}]  (311.5,89.8) -- (313.79,141.57) ;
\draw  [dash pattern={on 0.84pt off 2.51pt}]  (170.5,101.8) -- (170.67,141.36) ;
\draw   (235.52,142.02) -- (230.28,142.02)(232.9,139.4) -- (232.9,144.63) ;
\draw   (101.52,142.02) -- (96.28,142.02)(98.9,139.4) -- (98.9,144.63) ;
\draw   (400.52,142.02) -- (395.28,142.02)(397.9,139.4) -- (397.9,144.63) ;
\draw  [dash pattern={on 0.84pt off 2.51pt}]  (75.5,67.8) -- (232.5,67.8) ;
\draw  [dash pattern={on 0.84pt off 2.51pt}]  (76.5,45.8) -- (233.5,45.8) ;

\draw (239,152) node  [inner sep=0.75pt]  [font=\scriptsize,rotate=-359,xslant=0]  {$S_{\tk}$};
\draw (393.63,144.31) node [anchor=north] [inner sep=0.75pt]  [font=\scriptsize]  {$S_{\tk} +p_{\tkk}^{+}$};
\draw (106,144.31) node [anchor=north] [inner sep=0.75pt]  [font=\scriptsize]  {$S_{\tk} -p_{\tkk}^{-}$};
\draw (170.34,152)  node  [font=\scriptsize]  {$b_{\tk}$};
\draw (314.23,152) node  [font=\scriptsize]  {$a_{\tk}$};
\draw (202.39,174.14) node  [font=\scriptsize]  {$L_{\tk}^{-}$};
\draw (275.19,173.28) node  [font=\scriptsize]  {$L_{\tk}^{+}$};
\draw (322.14,84.49) node  [font=\scriptsize]  {$Q_{\tkk}^{+}$};
\draw (160.94,93.09) node  [font=\scriptsize]  {$Q_{\tkk}^{-}$};
\draw (48.14,44.49) node  [font=\scriptsize]  {$c^{+}_{t_{k+1}} p^{+}_{t_{k+1}}$};
\draw (48.14,66.49) node  [font=\scriptsize]  {$c^{-}_{t_{k+1}} p^{-}_{t_{k+1}}$};
\draw (423.07,131.16) node [align=left] {{\scriptsize Price}};
\draw (83.81,4.44) node [align=left] {Number of filled shares};

\end{tikzpicture}
\caption{$S_{\tk}-p_{\tkk}^-$ is the lowest price that a sell market order can attain, and $S_{\tk}+p_{\tkk}^+$ is the highest price that a buy market order can attain during the time interval $[\tk,\tkk)$. The number of filled shares increase as the market maker places limit orders closer to the fundamental price $S_\tk$.}  \label{fig:linearQ}
\vspace{.5 cm}
\end{figure}
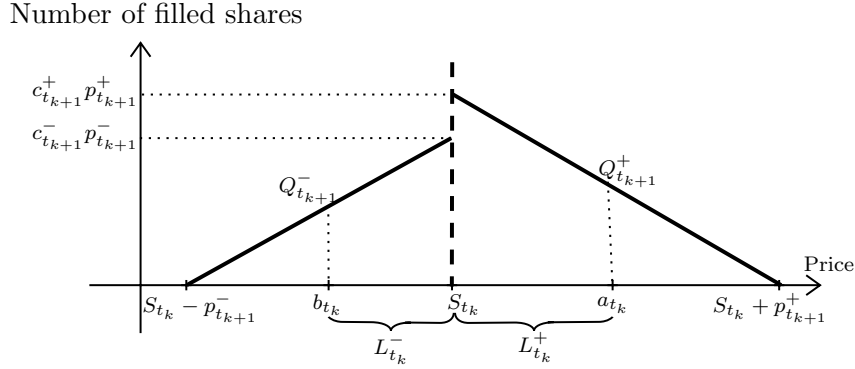

\begin{figure}
    \centering
  \includegraphics[width=1\textwidth]{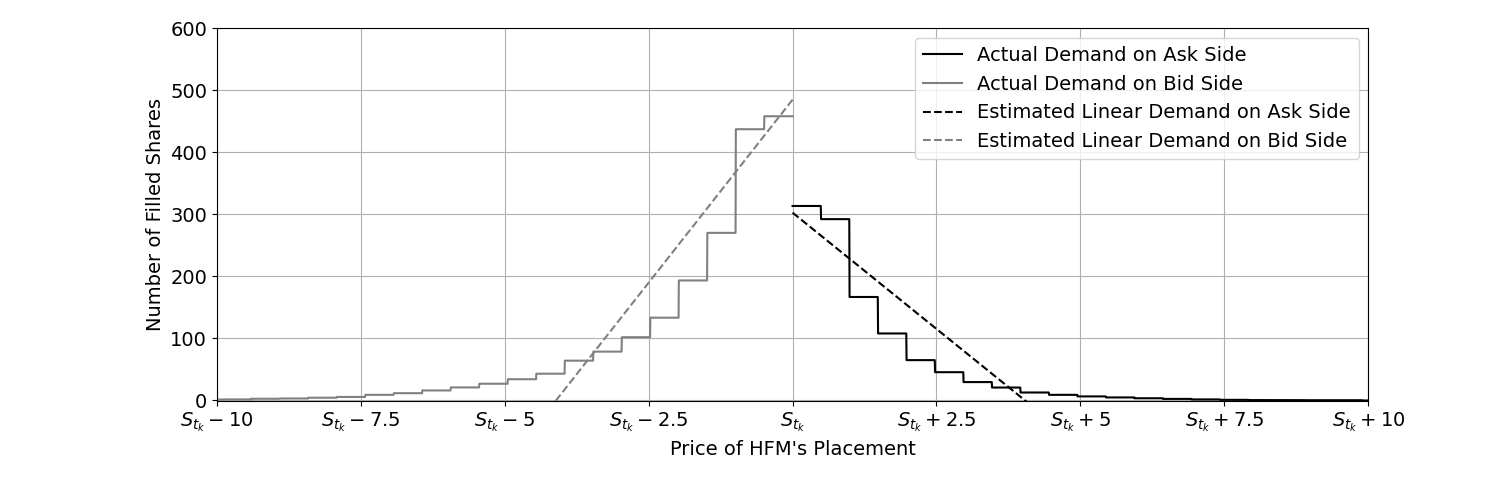}
  \caption{\textbf{{Prototypical Plot of the Actual Demand vs. Estimated Linear Demand over Time Interval $[\tk,\tkk)$.}}}
  \label{DemandPlotDemo}
  \vspace{.5 cm}
\end{figure}
Next we introduce the main assumptions on the distribution of the random variables $\onep,\onem,$ $c^+_{t_{k+1}},p^+_{t_{k+1}},c^-_{t_{k+1}}$, and $p^-_{t_{k+1}}$.

\begin{assumption} \label{assump:cp} For $k\in\{0,1,\ldots,N\}$, let {\DRed $\Gt:=\sigma(\F,\onep,\onem)$}. Then, 
\begin{enumerate}
		\item[{\rm (i)}] $(c^+_{t_{k+1}},p^+_{t_{k+1}})$ and $(c^-_{t_{k+1}},p^-_{t_{k+1}})$ are conditionally independent given $\Gt$.
		\item[{\rm (ii)}] The conditional distribution of $(c^+_{t_{k+1}},p^+_{t_{k+1}})$ given $\Gt$ is a measurable function depending solely on $\onep$, and it does not depend on $k$.
		\item[{\rm (iii)}] The conditional {distribution} of $(c^-_{t_{k+1}},p^-_{t_{k+1}})$ given $\Gt$ is a measurable function depending solely of $\onem$ that does not depend on {$k$.}
		\item[{\rm (iv)}] Let $\varpi$ be the fixed number defined in Eq.~\eqref{eq:defn:etk} and define $\Hkk=\sigma(\pmb{e}_{t_k})\subset\F$. Then, for some $d\in{\mathbb{N}}$, 
		{we assume the existence of} functions $g:\{0,1\}^{{2 \varpi}}\to\R^d$ and  $f,f^\pm:\R^d\to[0,1]$ such that for all $k\in\{0,1,2,\ldots,N\}$,
		\begin{equation}\label{eqn:gfunc:pi}
		\begin{aligned}
		\pi^\pm_{\tkk} &:= \mathbbm{P}(\mathbbm{1}_\tkk^\pm=1|\F)= \mathbbm{P}(\mathbbm{1}_\tkk^\pm=1|\Hkk) =f^\pm(g(\pmb{e}_\tk)),\\
		\pi_{\tkk}(1,1) &:= \mathbbm{P}(\onep=1,\onem=1|\F)=\mathbbm{P}(\onep=1,\onem=1|\Hkk)=f(g(\pmb{e}_\tk)).
		\end{aligned}
		\end{equation}
	\end{enumerate}
\end{assumption}

{By virtue of conditions (ii) and (iii) above, for each $m,n\in\mathbb{N}$, there exist functions $h^{+}_{m,n}:\{0,1\}\to\mathbb{R}$  and $h^{-}_{m,n}:\{0,1\}\to\mathbb{R}$ such that
\begin{equation}\label{eqn:mucpa0}
\Ex\Big[(c_{t_{k+1}}^\pm)^m(p_{t_{k+1}}^\pm)^n\Big|\mathcal{G}_{t_{k+1}}\Big]= \mathbb{E}\Big[(c_{t_{k+1}}^{\pm})^m(p_{t_{k+1}}^\pm)^n\Big|\mathbbm{1}_{t_{k+1}}^\pm\Big]=h^{\pm}_{m,n}(\mathbbm{1}_{t_{k+1}}^\pm),
\end{equation}
for any $k\in\{0,1,\ldots,N\}$. The optimal placement strategy will depend on the following nonrandom quantities:
\begin{equation}\label{eqn:mucp}
\mu_{c^mp^n}^{\pm}:=h^{\pm}_{m,n}(1)=\mathbb{E}\Big[(c_{t_{k+1}}^{\pm})^m(p_{t_{k+1}}^\pm)^n\Big|\mathbbm{1}_{t_{k+1}}^\pm=1\Big].
\end{equation}
When $m$ or $n$ are $0$, we simply write $\mu_{c^m}^{\pm}:=\mu_{c^mp^{0}}^{\pm}$ and $\mu_{p^n}^{\pm}:=\mu_{c^{0}p^n}^{\pm}$, and omit the exponents if $m$ and/or $n$ are $1$.} {Note also that 
\begin{align}\nonumber
	\pi_{\tkk}(1,0)&:=\Px\Big[\mathbbm{1}_{t_{k+1}}^+=1,\mathbbm{1}_{t_{k+1}}^-=0\,\Big|\,\F\Big]
	=\pi^+_{\tkk}-\pi_{\tkk}(1,1),\\
	\label{eq1234}	
	\pi_{\tkk}(0,1)&:=\Px\Big[\mathbbm{1}_{t_{k+1}}^+=0,\mathbbm{1}_{t_{k+1}}^-=1\,\Big|\,\F\Big]
	=\pi^-_{\tkk}-\pi_{\tkk}(1,1),\\
	\nonumber
	\pi_{\tkk}(0,0)&:=\Px\Big[\mathbbm{1}_{t_{k+1}}^+=0,\mathbbm{1}_{t_{k+1}}^-=0\,\Big|\,\F\Big]
	=1-\pi^+_{\tkk}-\pi^-_{\tkk}+\pi_{\tkk}(1,1),
\end{align}
and, thus, they all satisfy representations similar to \eqref{eqn:gfunc:pi} and, in particular, can be written as functions of $g(\pmb{e}_\tk)$.}

\begin{remark}\label{AdpHistMMD}
It is important to stress the {relevance} of the assumption given by Eq.~\eqref{eqn:gfunc:pi}. In \cite{adrian2016intraday} {\Blue and our earlier preprint \cite{zoe},} it is {assumed} that the probabilities {$\pi^\pm_{\tkk}$ and $\pi_{\tkk}(1,1)$} are {deterministic smooth functions of time, fixed} throughout the trading day. {In that case, for implementation purposes,} these functions {have to be} estimated at the beginning of the trading day {from}, for instance, historical data or other type of preliminary market analysis, {but once they are chosen, {they} cannot be changed through the trading day}. In the present work these probabilities are allowed to `react' to the `recent' history of buy/sell market orders $\pmb{e}_{t_k}=(\mathbbm{1}_{t_{k}}^\pm,\mathbbm{1}_{t_{k-1}}^\pm,\dots,\mathbbm{1}_{t_{k-\varpi+1}}^\pm)$, through a chosen function $g$. {The purpose of the function $g$ is two-fold: it summarizes the information contained in $\pmb{e}_{t_k}$ and it allows us to {alleviate} the computational burden by reducing the dimension of past information}. {This novelty enables the MM to adapt or adjust her trading strategy to the recently observed ``trades" in the market, which as shown empirically in {Section \ref{EstPisH}}, can provide a good forecast for the likelihood of a MO arriving on a given interval in either side of the book.} 
In our framework, the hyperparameter functions $f^{\pm}$ and $f$ in Eq,~\eqref{eqn:gfunc:pi} {will then} have to be calibrated at the beginning of the trading day based on historical data. We can think of each value of $g$ as a possible `scenario' of the recent MOs history. At the beginning of the trading day, we calibrate the probabilities $\pi^\pm_{\tkk}$ and $\pi_{\tkk}(1,1)$ for each possible scenario. This will allow us to choose the best possible placement strategy for each possible scenario. {We give further details in Subsection \ref{EstPisH}.}
\end{remark}

\begin{remark}
	{Under our conditions stated in Assumption \ref{assump:cp}, the average number of lifted shares from the MM will depend on the constants $\mu_{cp}^{\pm}$ and $\mu_{c}^{\pm}$ and the probabilities (\ref{eqn:gfunc:pi}). In particular, these will be adapted to the recent history of MOs, $\pmb{e}_{t_k}=(\mathbbm{1}_{t_{k}}^\pm,\mathbbm{1}_{t_{k-1}}^\pm,\dots,\mathbbm{1}_{t_{k-\varpi+1}}^\pm)$, at each time $t_k$. For instance, if the MM placed her sell LO at the the same level as the reference price at time $t_k$ ($L^{+}_{t_k}=0$), she would expect that $\pi_{t_{k+1}}^{+}\mu_{cp}^{+}$ shares of her order would be sold. In general, if $L_{t_k}^{\pm}=i$ ticks away from the reference price $S_{t_k}$, she would expect $\pi_{t_{k+1}}^{\pm}(\mu_{cp}^{\pm}-i\mu_{c}^{\pm})$ shares of her LO to be lifted during $(t_k,t_{k+1}]$. Indeed, from the formula (\ref{pdemand}) and notations (\ref{eqn:gfunc:pi})-(\ref{eqn:mucp}),} 
\begin{align*}
	\mathbb{E}\left[\left.Q_{t_{k+1}}^{\pm}\right|\mathcal{F}_{t_{k}}\right]&=\mathbb{E}\left[\left.\onep\mathbb{E}\left[\left. \cp(\Pp-i)\right|\mathcal{F}_{t_{k}},\mathbbm{1}_{t_{k+1}}^\pm\right]\right|\mathcal{F}_{t_{k}}\right]\\
	&=\mathbb{E}\left[\left.\onep [h_{cp}^{\pm}(\onep)-ih_{1,0}(\onep)]\right|\mathcal{F}_{t_{k}}\right]=\pi_{t_{k+1}}^{\pm}(\mu^{\pm}_{cp}-i\mu^{\pm}_{c}).
\end{align*}
\end{remark}

\begin{remark}
\label{CnctCntDscr}
{\DRed There is a possible connection between the approach based on exponential lifting probabilities (cf. \cite{cartea2015risk}) and that based on linear demand functions. Specifically, if $\lambda^{\pm}$ is the arrival intensity of MO's and the lifting probability is set to be $\exp(-{\kappa}^{\pm} L^{\pm})$, where $L^{\pm}$ is the distance between the LO quote and the {fundamental price}, then, during a time span of $\Delta$, we expect that ${\Delta \lambda^{\pm} \exp(-\kappa^{\pm} L^{\pm})}$ times a MO will lift a LO placed at distance $L^{\pm}$. Since in this stream of literature, it is typically assumed that only `one' share of the order is lifted at a time, when $L^{\pm}$ is small (as it is commonly the case), the expected number of shares filled during a time span {$\Delta$} is approximately equal to $\Delta \lambda^{\pm}-\Delta\lambda^{\pm} \kappa^{\pm} L^{\pm}$, which is precisely linear in $L^{\pm}$. Since we are allowing actions to take place only at discrete times, we believe that the modeling based on stochastic linear demand functions {provides} greater flexibility.}
\end{remark}


{As in \cite{algo}, \cite{adrian2016intraday}, and others, for} the performance criterion of our placement strategy, we use $W_T+S_T I_T -\lambda I_T^2$, where $W_t$ and $I_t$ {respectively} represent the MM's cash holding and stock inventory at time $t${,} and $\lambda$ is a constant penalization term. Note that at time $T=t_{N+1}$, the last two terms can be rewritten as, $S_T I_T -\lambda I_T^2=I_T(S_T-\lambda I_T)$, which may be interpreted as the MM's end-of-the-day cash flow incurred when liquidating her inventory $I_T$ using a {MO.} {\DRed Overall, the latter interpretation seems to be a good approximation of reality as shown by our empirical analysis of Section \ref{NumRsltSect} (compare Tables \ref{tab:TerminalRewardb} and \ref{tab:TerminalRewarda})\footnote{{\Blue Related to this, some recent works have proposed equilibrium models to deduce the price impact of a market order (see, e.g., \cite{Cetin}). See also \cite{Bhattacharya} for further insights about the relation between price impact and the depth of the book or the arrival frequencies of trades}.}.} {The optimal control problem {\DRed then} consists of finding the} {adapted placement positions} $L^\pm = (L_{t_0}^\pm,L_{t_1}^\pm,\dots,L_{t_N}^\pm)$ that maximize 
\begin{equation}\label{eq112}
\mathbb{E}[W_T+S_T I_T -\lambda I_T^2].
\end{equation}
{For future reference, note} that, {in light of Eqs.~(\ref{pdemand})} and (\ref{mdemand}), {we have, for} $k\in\{0,1,\dots,N\}$,
\begin{align}\label{eqi11}
I_{t_{k+1}}&=I_{t_{k}}-\onep\cp(\Pp-\lp) +\onem\cm(\Pm-\lm),\\
W_{t_{k+1}}&=W_{t_{k}}+(S_\tk+\lp)\onep\cp(\Pp-\lp) - (S_\tk-\lm)\onem\cm(\Pm-\lm).\label{eqw11}
\end{align}

\subsection{Optimal Placement Strategy for a Martingale Midprice Process} \label{sec:optimal:martingale}

For ease of exposition and to establish the main ideas, in this subsection we {first} present the solution of the optimal placement problem under the assumption that the {reference price process $\{S_{t_{k}}\}_{k\geq{}0}$} is a martingale. The case of a general price process 
is presented in the following subsection. The results herein {will enable} us to give a more tractable presentation of the general case. All the proofs in this subsection are deferred to Appendix \ref{appdx:martingale:price}. 

In order to proceed, we need to make an {additional} assumption on the distribution of the increments {of} the price process.
 
\begin{assumption}\label{assump:price}
{\rm (i)} For any $k \in\{0,1,\ldots,N\}$, the price increments $S_{t_{k+1}}-S_{t_{k}}$  and the random vector $(\onep,\onem,c^+_{t_{k+1}},p^+_{t_{k+1}},c^-_{t_{k+1}},p^-_{t_{k+}})$ are conditionally independent given $\F$, {and {\rm (ii)} $\{S_{t_{k+1}}-S_{t_{k}}\}_{k=0,\dots,N}$ is a martingale, i.e.,   $\E[S_{t_{k+1}}-S_{t_k}|\F]=0$, for any $k=0,\dots,N$.}
\end{assumption}

We now specify when a strategy will be admissible. {We specify two types of admissibility}.

\begin{definition} \label{defn:admissible}
	For any $k\in\{0,1,\ldots,N\}$, a strategy $(L^\pm_{t_k},\dots,L^\pm_{t_N})$ running from time $t_k$ to time $t_N$ is {said to be} admissible if, for every $j\geq k$, $L^\pm_{t_j} \in\mathcal{F}_{t_j}$. 
	{If, in addition, we have $L^+_{t_j}+L^-_{t_j}>0$, for all $j\geq k$, we say that the strategy is strictly admissible. The set of all (strictly) admissible strategies running from time $t_k$ to time $t_N$ is denoted by ($\bar{\A}_{{t_k,t_N}}$) $\A_{{t_k,t_N}}$.}
\end{definition}

Note that the {strict} admissibility condition $L^+_{t_j}+L^-_{t_j}>0$  is equivalent to {\DRed $a_{t_j}>{}b_{t_j}$}, i.e., the {selling} price $a_{t_{j}}=S_{t_{j}}+L^{+}_{t_{j}}$ of the {MM} is higher than her {buying} price $b_{t_j}=S_{t_{j}}-L^{-}_{t_{j}}$ at all future times. We don't require that $L^+$ and $L^-$ are nonnegative because $S_{t_k}$ is not {necessarily} seen as the midprice $S^{mid}_{t_{k}}$, but rather as the `fundamental' price of the stock. In practice, the placement {\DRed will always be} set at the tick right above (below) the midprice if $a_{t_k}$ ($b_{t_k}$) is found to be {below (above)} $S^{mid}_{t_{k}}$.

{In accordance to performance criterion \eqref{eq112}}, we can then write the value function at time $t_k$ as
\begin{equation}\label{OOCP0}
V_{t_k}:=\sup_{(L^\pm_{t_k},\dots,L^\pm_{t_N})\in \mathcal{A}_{{t_k,t_N}}}\mathbb{E}[\left.W_T+S_T I_T -\lambda I_T^2\right|\F], \qquad  k=0,1,\dots,N.
\end{equation}
To solve the optimal control problem, {we first} assume that $V_{t_k}$ follows the ansatz
\begin{equation}\label{eqn:ansatz}	
V_{t_k}=W_\tk+\sk I_\tk+\alpha_{\tk}I^2_\tk+h_{\tk}I_\tk + g_{t_k},
\end{equation}
where $\alpha_{\tk},h_{\tk},g_{\tk}$ are some 
$\F$-adapted real-valued random variables to be determined from the dynamical programming principle (see {Theorem} \ref{prop:measurability} below). The ansatz is motivated by the specific form of the performance criterion in \eqref{eq112} and the dynamic principle given in Eq.~\eqref{eqn:optimalprob} below.

{The Dynamic Programming Principle {\Blue (see, e.g., \cite{hernandez} and \cite{Rieder})} associated with \eqref{OOCP0} can then be written as
\begin{equation}\label{eqn:optimalprob}
V_{t_k} = \sup_{L^\pm_{t_k} \in \mathcal{A}_{t_k}} \mathbb{E}\Big[V_{t_{k+1}}\,\Big|\,\F\Big],\quad k=0,\dots,N,
\end{equation}
where} we set $V_{T}:=V_{t_{N+1}}:=W_T+S_T I_T -\lambda I_T^2$ {and $\mathcal{A}_{t_k}$ consists of all $L^\pm_{t_k} \in\mathcal{F}_{t_k}$}. Using the ansatz \eqref{eqn:ansatz}, we can rewritten (\ref{eqn:optimalprob}) as
\begin{equation}\label{eq113}
\begin{aligned}
&W_\tk+\alpha_{\tk}I_\tk^2+\sk I_\tk+h_{\tk}I_\tk +g_{\tk}\\
&\quad =\sup_{L^\pm_{t_k} \in \mathcal{A}_{t_k}} \mathbb{E}\Big[W_\tkk+\alpha_{\tkk}I_\tkk^2+\skk I_\tkk +{\Blue h_{\tkk}I_\tkk} +g_{\tkk}\,\Big|\,\F\Big].
\end{aligned}
\end{equation}
By plugging the recursions (\ref{eqi11})-(\ref{eqw11}) in (\ref{eq113}), we will be able to find a candidate for the optimal placement strategy ({Theorem} \ref{prop:measurability} and Corollary \ref{cor:candidates} below). 
{It is until Theorem \ref{VeriThrm1} when we shall verify that our candidate is indeed the solutions to our original optimal control problem \eqref{OOCP0}. In Section \ref{sec:admissibility}, we study the strict admissibility of the optimal strategy.}

{To write explicit formulas for the optimal placement strategy, we introduce the following terminology:
\begin{align}
\label{eqn:gamma:eta}
\begin{split}
		\rho_{t_k}^\pm &:={\pi^+_{\tkk}}{\pi^-_{\tkk}}\muo^\pm(\alpha^{1\mp}_{\tkk}\mut^\mp-\muo^\mp),\\
		\psi_{t_k}^\pm&:={\pi^\mp_\tkk}{\pi_\tkk(1,1)}\alpha^{1,1}_\tkk\muo^\pm(\muo^\mp)^2,\\
	\gamma_{t_{k}}&:=\left({\pi_\tkk(1,1)}\alpha^{1,1}_\tkk\mu^{-}_{c}\mu^{+}_{c}\right)^2-\frac{\rho_{t_{k}}^{+}\rho_{t_{k}}^{-}}{{\pi^+_{\tkk}}{\pi^-_{\tkk}}\mu^{+}_{c}\mu^{-}_{c}},
\end{split}
	\end{align}
where above $\alpha^{1\pm}_{\tkk}$ and $\alpha^{1,1}_{\tkk}$ are defined using the notation \eqref{Dfnalh0}. 	
We will prove below that the optimal spreads for the ask and bid side can be written as $L_{\tk}^{+,*}:={}^{\scaleto{(1)}{5pt}}\!A^+_{\tk}I_\tk+{}^{\scaleto{(2)}{5pt}}\!A^+_{\tk}+{}^{\scaleto{(3)}{5pt}}\!A^+_{\tk}$ and $L_{\tk}^{-,*}=-\,{}^{\scaleto{(1)}{5pt}}\!A^-_{\tk}I_\tk-{}^{\scaleto{(2)}{5pt}}\!A^-_{\tk}+{}^{\scaleto{(3)}{5pt}}\!A^-_{\tk}$ with the coefficients:
\begin{align} \label{eq:A1}
	\begin{split}
	{}^{\scaleto{(1)}{5pt}}\!A^\pm_{\tk}&:=\frac{\alpha^{1\pm}_{\tkk}\rho^{\pm}_{t_{k}}-\alpha^{1\mp}_{\tkk}\psi_{t_{k}}^{\pm}}{\gamma_{t_{k}}},
	\\
	{}^{\scaleto{(2)}{5pt}}\!A^\pm_{\tk}&:=\frac{h^{1\pm}_{\tkk}\rho^{\pm}_{t_{k}}-h^{1\mp}_{\tkk}\psi_{t_{k}}^{\pm}}{2\gamma_{t_{k}}},
	\\ 
	{}^{\scaleto{(3)}{5pt}}\!A^\pm_{\tk}&:=\frac{\rho_{t_{k}}^{\pm}}{2\pi^\pm_\tkk\mu_{c}^{\pm}\gamma_{t_{k}}}\Bigg[\pi^\pm_{\tkk} \Big(\muoo^\pm-2\alpha^{1\pm}_{\tkk}\muto^\pm\Big) +2\frac{\psi^{\pm}_{t_{k}}\mu_{cp}^{\mp}}{\pi_{t_{k+1}}^{\mp}(\mu_{c}^{\mp})^2}\Bigg]\\
	&\qquad +\frac{\psi_{t_{k}}^{\pm}}{2\pi^\mp_\tkk\mu_{c}^{\mp}\gamma_{t_{k}}}\Bigg[\pi^\mp_\tkk \Big(\muoo^\mp-2\alpha^{1\mp}_{\tkk}\muto^\mp\Big)+2\frac{\psi^{\mp}_{t_{k}}\mu_{cp}^{\pm}}{\pi_{t_{k+1}}^{\pm}(\mu_{c}^{\pm})^2}\Bigg],
\end{split}	
\end{align}
where we again used \eqref{Dfnalh0} to define $h^{1\pm}_{\tkk}$. We first show that the maximization problem in \eqref{eq113} is indeed well-posed.}

\begin{theorem}\label{prop:measurability}
{Under the Assumptions \ref{assump:cp}  and \ref{assump:price},} the following statements hold:
\begin{itemize}
	\item[{\rm (i)}] {There exist coefficients  $\alpha_{t_{k}}$, $h_{t_k}$, and $g_{\tk}$ that solve \eqref{eq113} for $k=0,\dots,N$ with dynamics \eqref{eqi11}-\eqref{eqw11} and terminal conditions $\alpha_{t_{N+1}}=-\lambda$ and $h_{t_{N+1}}=g_{t_{N+1}}=0$.}
	\item[{\rm (ii)}] For $k=0,\dots,N+1$, the coefficients $\alpha_{t_k}$, $h_{t_k}$, and $g_{t_k}$ in {\rm (i)} are $\Hkk$-measurable random variables, {where recall} from Assumption \ref{assump:cp} that $\Hkk=\sigma(\pmb{e}_{t_k})=\sigma(\mathbbm{1}_{t_{k}}^\pm,\mathbbm{1}_{t_{k-1}}^\pm,\dots,\mathbbm{1}_{t_{k-\varpi+1}}^\pm)\subset\F$.
	\item[{\rm (iii)}] 
	For $k=0,\dots,N$, the {coefficients of {\rm(i)}} can be computed recursively by the equations:
	\begin{align}\label{eqn:alphak}
	\begin{split}
	\alpha_{\tk}&=\alpha^0_{\tkk}+\sum_{\delta=\pm} \pi^\delta_\tkk \Bigg[\Big(\alpha^{1\delta}_{\tkk}\mut^\delta-\muo^\delta\Big)\Big({}^{\scaleto{(1)}{5pt}}\!A^\delta_{\tk}\Big)^2+2\alpha^{1\delta}_{\tkk}\muo^\delta \; {}^{\scaleto{(1)}{5pt}}\!A^\delta_{\tk}\Bigg]\\
	&\qquad  \qquad+2\alpha^{1,1}_{\tkk}{\Black \pi_\tkk(1,1)}\muo^+\muo^- \; {}^{\scaleto{(1)}{5pt}}\!A^+_{\tk}\; {}^{\scaleto{(1)}{5pt}}\!A^-_{\tk},
	\end{split}
	\end{align} 	
	\begin{align}\nonumber
	h_{\tk}&= h^0_{\tkk}+\sum_{\delta=\pm} \pi^\delta_\tkk \Bigg\{2\Big(\alpha^{1\delta}_{\tkk}\mut^\delta-\muo^\delta\Big)\;{}^{\scaleto{(1)}{5pt}}\!A^\delta_{\tk}\left({}^{\scaleto{(2)}{5pt}}\!A^\delta_{\tk} + (\delta\; {}^{\scaleto{(3)}{5pt}}\!A^\delta_{\tk})\right) +2\alpha^{1\delta}_{\tkk}\muo^\delta\Big({(\delta\,{}^{\scaleto{(3)}{5pt}}\!A^\delta_{\tk})+ {}^{\scaleto{(2)}{5pt}}\!A^\delta_{\tk}}\Big)\\ \label{eqn:htk}
	&\qquad\qquad\qquad\qquad\quad -2(\delta\alpha^{1\delta}_{\tkk})\muoo^\delta
	+(\delta\; {}^{\scaleto{(1)}{5pt}}\!A^\delta_{\tk})\Big(\muoo^\delta+(\delta h^{1\delta}_{\tkk})\muo^\delta-2\alpha^{1\delta}_{\tkk}\muto^\delta\Big)\Bigg\}\\
	\nonumber
	&\quad -2\alpha^{1,1}_{\tkk}{\Black \pi_\tkk(1,1)}\muo^+\muo^-\Bigg[{}^{\scaleto{(1)}{5pt}}\!A^+_{\tk}\Big( {}^{\scaleto{(3)}{5pt}}\!A^-_{\tk}-{}^{\scaleto{(2)}{5pt}}\!A^-_{\tk}\Big)-{}^{\scaleto{(1)}{5pt}}\!A^-_{\tk} \Big({}^{\scaleto{(2)}{5pt}}\!A^+_{\tk}+{}^{\scaleto{(3)}{5pt}}\!A^+_{\tk}\Big)+ {}^{\scaleto{(1)}{5pt}}\!A^-_{\tk}\;\frac{\muoo^+}{\muo^+} -{}^{\scaleto{(1)}{5pt}}\!A^+_{\tk}\;\frac{\muoo^-}{\muo^-}\Bigg],
	\end{align}  
{and }	
\begin{align}       
 \nonumber
        g_{\tk}&=g_{\tkk}^{0}+\sum_{\delta=\pm}{\pi^\delta_\tkk}\Big[(\alpha^{1\delta}_{\tkk}\mut^\delta-\muo^\delta)({ {}^{\scaleto{(3)}{5pt}}\!A^\delta_{\tk}+(\delta\, {}^{\scaleto{(2)}{5pt}}\!A^\delta_{\tk}}))^2+\alpha_{\tkk}^{1\delta}\mutt^\delta-{(\delta h^{1\delta}_{\tkk})}\muoo^\delta\\
                \label{eq:g}
        &\quad\qquad\qquad\qquad\qquad\quad+(\muoo^\delta+{(\delta h^{1\delta}_{\tkk})}\muo^\delta-2\alpha^{1\delta}_{\tkk}\muto^\delta)({ {}^{\scaleto{(3)}{5pt}}\!A^\delta_{\tk}+{(\delta\,} {}^{\scaleto{(2)}{5pt}}\!A^\delta_{\tk}}))\Big]\\
        &\qquad\qquad-2\alpha^{1,1}_{\tkk}{\pi_\tkk(1,1)}\muo^+\muo^-\Big[({}^{\scaleto{(2)}{5pt}}\!A^+_{\tk}+{}^{\scaleto{(3)}{5pt}}\!A^+_{\tk})({ {}^{\scaleto{(3)}{5pt}}\!A^-_{\tk}-{}^{\scaleto{(2)}{5pt}}\!A^-_{\tk}})
              \nonumber
        \\
        &\quad\qquad\qquad\qquad\qquad\qquad\qquad\qquad-\frac{\muoo^+}{\muo^+}({{}^{\scaleto{(3)}{5pt}}\!A^-_{\tk}-{}^{\scaleto{(2)}{5pt}}\!A^-_{\tk}})-\frac{\muoo^-}{\muo^-}({}^{\scaleto{(2)}{5pt}}\!A^+_{\tk}+{}^{\scaleto{(3)}{5pt}}\!A^+_{\tk})+\frac{\muoo^+\muoo^-}{\muo^-\muo^+}\Big]{,}
        \nonumber
\end{align}	
where we used the notation \eqref{Dfnalh0}, \eqref{eqn:gamma:eta}, and \eqref{eq:A1}.
	\end{itemize}
\end{theorem}

{Before finding the optimal controls of \eqref{eqn:optimalprob}, we state an important preliminary result that will also be needed to show {the verification theorem and the strict} admissibility of the optimal controls. 
 This result is deceivable simple, though its proof is rather intricate.}
\begin{lemma}\label{lemma:alpha}
	The random variables $\{\alpha_{t_k}\}_{k\geq0}$ defined in Eqs.~(\ref{eqn:alphak}) are such that
	\[
	\alpha_\tkk^0<\alpha_\tk<0,\quad \text{for any }\;k\in\{0,1,\ldots,N\},
	\]
where above $\{\alpha_{t_k}^0\}_{k\geq{}0}$ is computed from $\{\alpha_{t_k}\}_{k\geq{}0}$ using \eqref{Dfnalh0}.
\end{lemma}

We are now ready to find the optimal controls of \eqref{eqn:optimalprob} under the ansatz \eqref{eqn:ansatz}. Most of its proof is embedded into the proof of {Theorem} \ref{prop:measurability} but due to its importance it is stated separately.

\begin{corollary}\label{cor:candidates}
	The optimal placements that maximize the right-hand side of the Eq.~(\ref{eqn:optimalprob}) under the ansatz \eqref{eqn:ansatz} are given by
	\begin{align}\label{eqn:OptimalL}
	\begin{split}
	L_{\tk}^{+,*}&={}^{\scaleto{(1)}{5pt}}\!A^+_{\tk}I_\tk+{}^{\scaleto{(2)}{5pt}}\!A^+_{\tk}+{}^{\scaleto{(3)}{5pt}}\!A^+_{\tk},\\
	L_{\tk}^{-,*}&=-\,{}^{\scaleto{(1)}{5pt}}\!A^-_{\tk}I_\tk-{}^{\scaleto{(2)}{5pt}}\!A^-_{\tk}+{}^{\scaleto{(3)}{5pt}}\!A^-_{\tk},
	\end{split}
	\end{align}
	\normalsize
	where the coefficients above are given as in \eqref{eq:A1}.
\end{corollary}

	{In light of the previous result, the optimal placement strategy takes the form:
		\begin{align}\label{DfnTldL0bcNFF}
\begin{split}	
		{a}_{t_k}^{*}&=S_{t_{k}}+{}^{\scaleto{(1)}{5pt}}\!A^+_{\tk}I_\tk+{}^{\scaleto{(2)}{5pt}}\!A^+_{\tk}+{}^{\scaleto{(3)}{5pt}}\!A^+_{\tk},\\
	{b}_{t_k}^{*}&=S_{t_{k}}+{}^{\scaleto{(1)}{5pt}}\!A^-_{\tk}I_\tk+{}^{\scaleto{(2)}{5pt}}\!A^-_{\tk}-{}^{\scaleto{(3)}{5pt}}\!A^-_{\tk}.
\end{split}
\end{align}
In the preprint \cite{zoe}, an extensive analysis of the properties of the optimal strategy was carried out in the case that the arrival intensity of MOs is deterministic rather than adaptive as in our setting (see Remark \ref{AdpHistMMD}). One of the main conclusions therein is that the randomness of $c$ and $p$ are not just a mathematical `artifact' for the sake of generalization, but play an important role in the behavior of the optimal placement strategy. Many of the conclusions therein transfer to our setting, but, for the sake of space, we just highlight some of the most important here:}
\begin{enumerate}
\item  {The second term in \eqref{DfnTldL0bcNFF} is fundamental as it can be interpreted as the inventory  adjustment to the optimal strategy. In the case of $\pi(1,1)=0$ (which is met to a good degree when trading frequency is high enough), the coefficient ${}^{\scaleto{(1)}{5pt}}\!A^-_{\tk}$ simplifies as follows: 
\begin{equation}\label{WNTSqrI}
	{}^{\scaleto{(1)}{5pt}}\!A^\pm_{\tk}=\frac{\muo^\pm\alpha^{1\pm}_{\tkk}}{\muo^\pm-\alpha^{1\pm}_{\tkk}\mut^\pm}=
	\frac{\alpha^{1\pm}_{\tkk}}{1-\alpha^{1\pm}_{\tkk}\muo^\pm-\alpha^{1\pm}_{\tkk}{\rm Var}(c^{\pm}_{t_{k+1}}|\mathcal{F}_{t_{k}})/\muo^\pm}.
\end{equation}
Due to Lemma \ref{lemma:alpha}, the coefficient above is negative, which means that when the inventory is positive (negative), the ask and bid levels decrease (increase) to stimulate selling (buying) of stock and, hence, bring inventory closer to $0$. The larger the level of the slope $c$, the smaller the effect of inventory in the optimal placement strategy. However, with the same average value of $c$, stocks with more variable $c$ require smaller inventory adjustment.} 

\item {In the case of $\pi(1,1)\neq{}0$, we still have that ${}^{\scaleto{(1)}{5pt}}\!A^\pm_{\tk}<0$.
Indeed, recalling \eqref{eqn:gamma:eta}-\eqref{eq:A1} and $\alpha^{1\pm}_{\tkk}\leq 0$, and applying \eqref{eqn:comparison:alpha}, the numerator of ${}^{\scaleto{(1)}{5pt}}\!A^\pm_{\tk}<0$ satisfies:
\begin{align*}
	\alpha^{1\pm}_{\tkk}\rho^{\pm}_{t_{k}}-\alpha^{1\mp}_{\tkk}\psi_{t_{k}}^{\pm}&=\alpha^{1\pm}_{\tkk}{\pi^+_{\tkk}}{\pi^-_{\tkk}}\muo^\pm(\alpha^{1\mp}_{\tkk}\mut^\mp-\muo^\mp)-\alpha^{1\mp}_{\tkk}{\pi^\mp_\tkk}{\pi_\tkk(1,1)}\alpha^{1,1}_\tkk\muo^\pm(\muo^\mp)^2\\
	&\geq
	\alpha^{1\pm}_{\tkk}{\pi^+_{\tkk}}{\pi^-_{\tkk}}\muo^\pm(\alpha^{1\mp}_{\tkk}\mut^\mp-\muo^\mp)-\alpha^{1\mp}_{\tkk}{\pi^\mp_\tkk}\alpha^{1\pm}_{\tkk}{\pi^{\pm}_{\tkk}}\muo^\pm(\muo^\mp)^2\\
	&=\alpha^{1+}_{\tkk}\alpha^{1-}_{\tkk}{\pi^+_{\tkk}}{\pi^-_{\tkk}}\muo^\pm(\mut^\mp-(\muo^\mp)^2)-\alpha^{1\pm}_{\tkk}{\pi^+_{\tkk}}{\pi^-_{\tkk}}\muo^+\muo^->0.
\end{align*}
Since the denominator $\gamma_{t_{k}}<0$ (cf.~(\ref{eq:Dneg})), we conclude that ${}^{\scaleto{(1)}{5pt}}\!A^\pm_{\tk}<0$.}

\item {Again, assuming that $\pi(1,1)\equiv 0$, we can further write:
\begin{align}\label{DfnTldL0bcNFFb}
\begin{split}	
		{a}_{t_k}^{*}&=S_{t_{k}}+\frac{\muo^+\alpha^{1+}_{\tkk}}{\muo^+-\alpha^{1+}_{\tkk}\mut^+}I_\tk+\frac{1}{2}\frac{\muo^+h^{1+}_{\tkk}}{\muo^+-\alpha^{1+}_{\tkk}\mut^+}
		+\frac{1}{2} \frac{\muoo^+-2\alpha^{1+}_{\tkk}\muto^+}{\muo^+-\alpha^{1+}_{\tkk}\mut^+},\\
	{b}_{t_k}^{*}&=S_{t_{k}}+\frac{\muo^-\alpha^{1-}_{\tkk}}{\muo^--\alpha^{1-}_{\tkk}\mut^-}I_\tk+\frac{1}{2}\frac{\muo^-h^{1-}_{\tkk}}{\muo^--\alpha^{1-}_{\tkk}\mut^-}-\frac{1}{2} \frac{\muoo^--2\alpha^{1+}_{\tkk}\muto^-}{\muo^--\alpha^{1-}_{\tkk}\mut^-}.
\end{split}
\end{align}
Computationally, it can be shown that $\alpha^{1\pm}_{\tkk}$ and $h^{1\pm}_{\tkk}$ are close to $0$ for most of the time interval $[0,T]$ and it is only for $t_{k+1}$ close to $T$, that their values are significantly different from $0$ (especially, $\alpha^{1\pm}_{\tkk}$). Then, we have the approximations:
\begin{align}\label{DfnTldL0bcNFFbAp}
\begin{split}	
		{a}_{t_k}^{*}&\approx S_{t_{k}}
		+\frac{1}{2} \frac{\muoo^+}{\muo^+}=
		S_{t_{k}}+\dfrac{\mu_{p}^{+}}{2}+\dfrac{{\rm Cov}(c^{+}_\tkk,p^{+}_{\tkk}|\F)}{2\mu_{c}^{+}},\\
		{b}_{t_k}^{*}&\approx S_{t_{k}}
		-\frac{1}{2} \frac{\muoo^-}{\muo^-}=S_{t_{k}}-\dfrac{\mu_{p}^{-}}{2}-\dfrac{{\rm Cov}(c^{-}_\tkk,p^{-}_{\tkk}|\F)}{2\mu_{c}^{-}} .
\end{split}
\end{align}
The correlation between $c$ and $p$ now plays a key role in the optimal placements. When $c$ and $p$ are uncorrelated (such as when $c$ or $p$ are deterministic), the optimal placements are near the midpoint between $S_\tk$ and the average reservation price $\sk\pm{\mu_{p}^\pm}$ for most of the time. However, when the correlation between $c$ and $p$ is positive, 
instead of placing LOs around $\sk\pm \mu_{p}^\pm/2$, the HFM will tend to go deeper into the book. Roughly, a larger realization of $c$ also implies a large value of $p$, resulting in a larger demand function and, hence, greater opportunity for the MM to obtain better prices for her filled LOs.}

\item {Under the condition $\pi(1,1)\equiv0$ and certain market symmetry and independence conditions, we can strengthen the conclusions of the previous item. Specifically, if we assume that $\pi^{+}=\pi^{-}$, $\mu_{cp}^{\pm}=\mu_{c}^{\pm}\mu_{p}^{\pm}$, $\mu_{c^2p}^{\pm}=\mu_{c^2}^{\pm}\mu_{p}^{\pm}$, $\mu_{c^m}^{+}=\mu_{c^m}^{-}=:\mu_{c^m}$ and $\mu_{p^m}^{+}=\mu_{p^m}^{-}=:\mu_{p^m}$, when $m=1,2$, and $\alpha^{1+}=\alpha^{1-}$ (see Lemma \ref{lemma:gfunc} and the proof of Corollary \ref{cor:cond2:alt} for sufficient conditions for the latter to hold), then we have that $h^{1\pm}_{t_{k+1}}\equiv0$ and \eqref{DfnTldL0bcNFFb} uncovers the existence of a critical inventory level that dictates the relation of the optimal placements relative to the nominal values ${a}_{t_k}^{0}:=S_{t_{k}}
		+ \frac{\mu_p}{2}$ and ${b}_{t_k}^{0}:=S_{t_{k}}
		-\frac{\mu_p}{2}$. Specifically, let ${I}^0:=\dfrac{\mut\mu_p}{2\muo}$.
Then, we have:}
\begin{itemize}
\item {When $I_\tk={I}^0$ ($I_\tk=-{I}^0$), the optimal ask (bid) quote is at the level 
${S_{t_{k}}+\mu_{p}/2}$ ({${S_{t_{k}}-\mu_{p}/2}$});}
\item {When the inventory level $I_\tk\in(0,{I}^0)$ ($I_\tk\in(-{I}^0,0)$), the optimal ask (bid) quote is deeper in the LOB {relative} to the levels 
${S_{t_{k}}+\mu_{p}/2}$ ({${S_{t_{k}}-\mu_{p}/2}$});} 
\item {When the inventory level $I_\tk>{I}^0$ ($I_\tk<-{I}^0$), the optimal strategy is to place the ask (bid) quote closer to {$S_{t_{k}}$} than {to} { ${S_{t_{k}}+\mu_{p}/2}$} ({${S_{t_{k}}-\mu_{p}/2}$}), and the bid (ask) quote {farther} from {$S_{t_{k}}$} than {from}  {${S_{t_{k}}-\mu_{p}/2}$} ({${S_{t_{k}}+\mu_{p}/2}$}) into the LOB.}
\end{itemize}
\end{enumerate}

%
%
%
We next prove a verification theorem for the optimal placements given in Eq.~(\ref{DfnTldL0bcNFF}). {Its proof is given in Appendix \ref{ProofVerifyH}.}
\begin{theorem}\label{VeriThrm1}
	{The optimal value function $V_\tk$ {of} the control problem (\ref{OOCP0}) is given by 
		\begin{equation*}
			V_\tk=v(t_k,\sk,W_\tk,I_\tk),
		\end{equation*} 
		where, for $\tk\in \mathcal{T}$,  
		\[
		v(\tk,s,{\mathsf{w}},i)= {\mathsf{w}}+\alpha_{\tk}i^{2}+s i+h_{\tk}i+g_{\tk},
		\]
		{with $\alpha_\tk$, $h_\tk$ and $g_\tk$ given {as} in {Theorem} \ref{prop:measurability}.
			Furthermore,} the optimal controls are given by $L_.^{\pm,*}$ as defined in {(\ref{eqn:OptimalL})}.}
\end{theorem}

\subsection{Optimal Placement Strategy for a General Midprice Process} \label{Sec:General:process}

The objective of this subsection is to extend our previous results to the case when the midprice process is a general stochastic process {without relying on a martingale assumption. As we will see below, in} that case, the optimal placement strategy will also depend on the forecasts of {future price} {changes:
\begin{equation}\label{FrcstFtrPri}
\Delta_{t_j}^\tk:=\E(S_{t_{j+1}}-S_{t_j}|\F),\quad j\geq k.
\end{equation}
We can see $\Delta_{t_j}^\tk$ as the {MM}'s forecast of the price change during the time interval $[t_j,t_{j+1}]$, $j\geq{}k$, as seen at time $\tk$.} We first need to {modify} our Assumption \ref{assump:price} as follows:
\begin{assumption}\label{assump:not:mart}
For any $k=1,2,\dots,N$, $\{{S_{t_{j+1}}-S_{t_j}}\}_{{j=k,
\dots,N}}$ and $(\onep,\onem,c^+_{t_{k+1}},p^+_{t_{k+1}},c^-_{t_{k+1}},p^-_{t_{k+1}})$ are conditionally independent given $\F$.
\end{assumption}

To solve the optimization problem \eqref{eqn:optimalprob}, we use an ansatz for the value function similar to that in the previous subsection:
\begin{equation}\label{eq:V_NMGb}
V_\tk:= W_\tk+\sk I_\tk+\alpha_{\tk}I_\tk^2+\widetilde{h}_{\tk}I_\tk+\widetilde{g}_{\tk},
\end{equation} 
where $\alpha_{\tk}$, $\widetilde{h}_{\tk}$, and $\widetilde{g}_{\tk}$ are $\F$-adapted real-valued random variables to be determined from the dynamical programming principle \eqref{eqn:optimalprob}. As one may suspect from the notation above, $\alpha_\tk$ will turn out to be the same as before: an $\Hkk$-measurable random variable determined by the recursive relation \eqref{eqn:alphak}. However, $\widetilde{h}_{\tk}$ will be different (in fact, not necessarily $\Hkk$-measurable).

The following theorem summarizes the analogous results of Theorem \ref{prop:measurability} and Corollary \ref{cor:candidates} under {a} general price dynamics. Its proof is provided in Appendix \ref{appdx:general:price}.

\begin{theorem}\label{thm:General:Price:controls}
	Under Assumptions \ref{assump:cp} and \ref{assump:not:mart}, 
	the optimal strategy that solves the Bellman equation \eqref{eqn:optimalprob} with the ansatz \eqref{eq:V_NMGb} and terminal conditions $\alpha_{t_{N+1}}=-\lambda$ and $\widetilde{h}_{t_{N+1}}=0$ is given, for $k=0,\dots,N$, by 	\begin{align}
	\begin{split}
	\widetilde{L}_{\tk}^{+,*}&={L}_{\tk}^{+,*}+\frac{{\rho_{t_k}^+-\psi_{t_k}^+}}{2\gamma_{t_{k}}}\Delta_{t_k}^{t_k}+\frac{1}{2\gamma_{t_{k}}}\sum_{i=k+1}^{N}\left\{{\rho_{t_k}^+}\Delta_{t_i}^{t_k}\mathbb{E}\Big[\mathbb{E}\Big(\prod_{l=k+2}^{i+1}\xi_l\Big|\FKK\Big)\Big|\F,\onep=1 \Big]\right.\\
	&\qquad\qquad\qquad\qquad\qquad\qquad\qquad\qquad      - \left.{\psi_{t_k}^+}\Delta_{t_i}^{t_k}\mathbb{E}\Big[\mathbb{E}\Big(\prod_{l=k+2}^{i+1}\xi_l\Big|\FKK\Big)\Big|\F,\onem=1 \Big]\right\},\\
	\widetilde{L}_{\tk}^{-,*}&={L}_{\tk}^{-,*}-\frac{\rho_{t_k}^--\psi_{{t_k}}^-}{2\gamma_{t_{k}}}\Delta_{t_k}^{t_k}-\frac{1}{2\gamma_{t_{k}}}\sum_{i=k+1}^{N}\left\{\rho_{{t_k}}^-\Delta_{t_i}^{t_k}\mathbb{E}\Big[\mathbb{E}\Big(\prod_{l=k+2}^{i+1}\xi_l\Big|\FKK\Big)\Big|\F,\onep=1 \Big]\right.\\
	&\qquad\qquad\qquad\qquad\qquad\qquad\qquad\qquad      \left.- \psi_{{t_k}}^-\Delta_{t_i}^{t_k}\mathbb{E}\Big[\mathbb{E}\Big(\prod_{l=k+2}^{i+1}\xi_l\Big|\FKK\Big)\Big|\F,\onem=1 \Big]\right\},\label{eq:LtildeNMGbb}
	\end{split}
	\end{align}
	where ${L}_{\tk}^{\pm,*}$ is defined as in Corollary \ref{cor:candidates}, {$\gamma_{t_k}$, $\rho_{t_k}^\pm$, and  $\psi_{t_k}^\pm$ are defined as in {\eqref{eqn:gamma:eta}}, and the quantity $\xi_{k+1} \in \FKK$ is given} as:
		\begin{align}
		\label{eqn:rho:k}
		\begin{split}
		\xi_{k+1}&= 1 +  \frac{\onep}{\pi^+_\tkk\gamma_\tk}(\pi^+_\tkk\muo^+\alpha^{1+}_\tkk\rho_{{t_k}}^+-  \pi^-_\tkk\muo^-\alpha^{1-}_\tkk\psi^{-}_{t_k}) \\ 
		&\qquad+  \frac{\onem}{\pi^-_\tkk\gamma_\tk}(\pi^-_\tkk\muo^-\alpha^{1-}_\tkk\rho^{-}_{t_k}-  \pi^+_\tkk\muo^+\alpha^{1+}_\tkk\psi^+_{t_k}).
		\end{split}
		\end{align}
\end{theorem}

\begin{remark} \label{rem:example0} 
{Formula (\ref{eq:LtildeNMGbb}) {gives us} some interesting insights, even in the simplest case where only one-step-ahead forecast $\Delta_{t_{k}}^{t_{k}}$ is implemented, while assuming that $\Delta_{t_i}^{t_k}=0$ afterward ($i\geq k+1$). In that case, the third summands of $\widetilde{L}_{\tk}^{+,*}$ and $\widetilde{L}_{\tk}^{-,*}$ vanish, yielding the following {parsimonious formulas for the optimal placement strategy}:
	\begin{align}\label{DfnTldL0}
\begin{split}	
		\widetilde{a}_{t_k}^{*}&:=S_{t_{k}}+\widetilde{L}_{\tk}^{+,*}=S_{t_{k}}+{L}_{\tk}^{+,*}+\frac{\rho_{{t_k}}^+-\psi_{{t_k}}^+}{2\gamma_{t_{k}}}\Delta_{t_k}^{t_k},\\
	\widetilde{b}_{t_k}^{*}&:=S_{t_{k}}-\widetilde{L}_{\tk}^{-,*}=S_{t_{k}}-{L}_{\tk}^{-,*}+\frac{\rho_{{t_k}}^--\psi_{{t_k}}^-}{2\gamma_{t_{k}}}\Delta_{t_k}^{t_k}.
\end{split}
	\end{align}
The last term above allows the MM to adjust `online' her strategy depending on her views or forecast of the next price change at each time instant $t_k$. This feature  in turn provides  a more data driven placement strategy in addition to the {method} described in Remark \ref{AdpHistMMD} above. Using the facts that $\alpha^{1,1}_{t_{k+1}}\leq{}0$, $(\mu_c^\pm)^2\leq{}\mu_{c^2}^\pm$, and (\ref{eqn:alpha1p})-(\ref{eqn:alpha1m}), it is possible to show that
\begin{align*}
\rho_{{t_k}}^{+}-\psi_{{t_k}}^{+}
&\leq{\pi^+_{\tkk}}{\pi^-_{\tkk}}\muo^+\left(\alpha^{1,1}_\tkk\mut^- \pi_\tkk(1,1)\left[\frac{1}{\pi^{-}_{t_{k+1}}}-\frac{1}{\pi^{+}_{t_{k+1}}}\right]-\muo^-+
\alpha^{0,1}_\tkk\mut^-\left[1-\frac{\pi_\tkk(1,1)}{\pi^{-}_{t_{k+1}}}\right]\right),\\
\rho_{{t_k}}^{-}-\psi_{{t_k}}^{-}
&\leq{\pi^+_{\tkk}}{\pi^-_{\tkk}}\muo^-\left(\alpha^{1,1}_\tkk\mut^+ \pi_\tkk(1,1)\left[\frac{1}{\pi^{+}_{t_{k+1}}}-\frac{1}{\pi^{-}_{t_{k+1}}}\right]-\muo^++
\alpha^{1,0}_\tkk\mut^+\left[1-\frac{\pi_\tkk(1,1)}{\pi^{+}_{t_{k+1}}}\right]\right).
\end{align*}
Since $\alpha^{0,1}_\tkk,\alpha^{1,0}_\tkk\leq{}0$, we conclude that $\rho_{{t_k}}^{\pm}-\psi_{{t_k}}^{\pm}\leq{}0$ if 
\begin{align}\label{CndMksD}
	\pi_{t_{k+1}}(1,1)=0,\quad \text{ or }\quad 
	\pi_{t_{k+1}}^+=\pi_{t_{k+1}}^{-}.
\end{align}
In those cases, since $\gamma_{t_{k}}<0$ (cf.~(\ref{eq:Dneg})), the coefficients of $\Delta_{t_k}^{t_{k}}$ in (\ref{DfnTldL0}) are positive. This sign makes sense since,  if, for example, $\Delta_{t_k}^{t_k}>0$, the MM will try to post her LOs at higher price levels (on both sides of the book) to anticipate the expected higher price $S_{t_{k+1}}$ in the subsequent interval. {The first condition in \eqref{CndMksD} is satisfied to a good extend when the trading frequency is high enough, while the second condition therein is supported empirically by our analysis in Section \ref{EstPisH}.}}

\end{remark}

\subsection{Admissibility of the Optimal Strategy}\label{sec:admissibility}

In this section, we will {give} sufficient conditions to guarantee that the optimal strategy of Theorem  \ref{thm:General:Price:controls} is {strictly} admissible. 
{All the proof of this subsection are deferred to Appendix \ref{App:proof:admissible}.}

Recall from Definition \ref{defn:admissible} that a strategy $(L^\pm_{t_0},\dots,L^\pm_{t_N})$ is {strictly} admissible if for all $k\in\{0,1,\ldots,N\}$,  $L^\pm_{t_k} \in\F$, and $L^+_{t_k}+L^-_{t_k}>0$, implying that the execution price $a_{t_{k}}=S_{\tk}+L^+_{t_k}$ of the {ask} LO  {is always larger} than the  execution price $b_{t_{k}}=S_{\tk}-L^+_{t_k}$ of the {bid} LO. Proposition \ref{prop:admissibility} below provides sufficient conditions under which the optimal strategy introduced in Theorem \ref{thm:General:Price:controls} {enjoys this property}.

\begin{proposition}  \label{prop:admissibility}
	Under {Assumptions \ref{assump:cp} and \ref{assump:not:mart}} and regardless of the dynamics of the midprice process, the optimal strategy of Theorem \ref{thm:General:Price:controls} yields positive spreads at all times $($i.e., $a_{t_{k}}>b_{t_{k}}$, for all $k\in\{0,\dots,N\})$, provided that the following four conditions hold:
	\begin{enumerate}
		\item[(1)] The first and second conditional moments of $c^\pm$, as defined in Equation \eqref{eqn:mucp}, satisfy
		\begin{equation}\label{Cnd1PosSpr}
		\muo := \muo^+=\muo^-,\quad \mut := \mut^+=\mut^-.
		\end{equation}
		\item[(2)] At every time $k\in\{0,\dots,N\}$, 
		\begin{equation}
		\begin{aligned}\label{Cnd2PosSpr}
		\pi^+_{\tkk}= \pi^-_{\tkk}=:\pi_{\tkk}.
		\end{aligned}
		\end{equation}
		\item[(3)] For every $k\in\{0,1,2,\ldots,N\}$, the conditional expectations of $c_{t_k}^\pm p_{t_k}^\pm$ and $(c_{t_k}^\pm)^2 p_{t_k}^\pm$, as defined in Eq.~\eqref{eqn:mucp}, satisfy
		\begin{equation}\label{Cnd3PosSpr}
		\mu_{cp}^\pm=\mu_{c}^\pm\mu_{p}^\pm,\quad \mu_{c^2 p}^\pm=\mu_{c^2}^\pm\mu_p^\pm.
		\end{equation}
		\item[(4)] {For every $k\in\{1,2,\ldots,N+1\}$, the $\Hkk$-measurable random variables $\alpha_{t_k}$ and $h_{t_k}$, defined by Eqs.~\eqref{eqn:alphak} and \eqref{eqn:htk}, depend on $\mathbbm{1}_{t_{k}}^+$ and $\mathbbm{1}_{t_{k}}^-$ only through $\mathbbm{1}_{t_{k}}^++\mathbbm{1}_{t_{k}}^-$}.
%
	\end{enumerate}
\end{proposition}

{Conditions \eqref{Cnd1PosSpr} and \eqref{Cnd2PosSpr} are some type of symmetry conditions between the bid and ask sides of the market.} 
Condition
\eqref{Cnd3PosSpr} postulates that the demand {and supply slopes $c_{t_{k+1}}^\pm$ and the} corresponding reservation prices $p_{t_{k+1}}^\pm$ are uncorrelated. These assumptions {are} empirically supported {by our empirical analysis in Section \ref{sec:Data} (see Figure \ref{fig:OptimalPaths} and Table \ref{tab1})}.

{The following result shows that the conditions (2) {and (3)} of Proposition \ref{prop:admissibility} can be relaxed in the case that there is no possibility of simultaneous arrivals of sell and buy market orders in the same subinterval. The latter condition is expected to be met {reasonably well} when the frequency of trades is high enough (i.e., $\max_{k}\{t_{k}-t_{k-1}\}\approx 0$).}

\begin{corollary} \label{cor:cond2:alt}
{Suppose that, for every $k\in\{0,1,\ldots,N+1\}$, the conditional probability $\pi_{\tkk}(1,1) = \mathbbm{P}(\onep=1,\onem=1|\F)=\mathbbm{P}(\onep=1,\onem=1|\Hkk)$ is $0$. Then, regardless of the dynamics of the midprice process, the optimal strategy is admissible if Conditions {(1) and (4)} of Proposition \ref{prop:admissibility} are satisfied.} 
\end{corollary}

{The most technical condition {in the results above} is (4), which can also be interpreted as another symmetry assumption}.
This condition could be difficult to verify due to the intrinsic complexity of the recursive formulas \eqref{eqn:alphak} and \eqref{eqn:htk}. 
 In the {Lemma \ref{lemma:gfunc} below, we show that it suffices to pick the function $g:\R^{2\varpi}\to\R^d$ introduced in Eq.~\eqref{eqn:gfunc:pi} such that it depends on $\pmb{e}_{t_k}^+:=(\mathbbm{1}_{t_{k}}^+,\ldots,\mathbbm{1}_{t_{k-\varpi+1}}^+)$ and $\pmb{e}_{t_k}^-:=(\mathbbm{1}_{t_{k}}^-,\ldots,\mathbbm{1}_{t_{k-\varpi+1}}^-)$ through $\pmb{e}_{t_k}^++\pmb{e}_{t_k}^-$.}
 
 \begin{lemma} \label{lemma:gfunc} {If the function $g$ in Eq.~\eqref{eqn:gfunc:pi} is of the form $g(\pmb{e}_\tk)=\varphi(\pmb{e}_{t_k}^++\pmb{e}_{t_k}^-)$ for some $\varphi: \{0,1,2\}^{\varpi}\to\mathbb{R}$, then $\alpha_{t_k}$ and $h_{t_k}$ will depend on $\pmb{e}_{t_k}^+$ and $\pmb{e}_{t_k}^-$ only thorugh $\pmb{e}_{t_k}^++\pmb{e}_{t_k}^-$. In particular, they depend on $\mathbbm{1}_{t_{k}}^+$ and $\mathbbm{1}_{t_{k}}^-$ only through $\mathbbm{1}_{t_{k}}^++\mathbbm{1}_{t_{k}}^-$ and the condition (4) in Proposition \ref{prop:admissibility} is satisfied.}
\end{lemma}

\subsection{Inventory Analysis of the optimal strategy} \label{sec:invent:analysis}
{In this section we analyze the behavior of the inventory under the optimal strategy found in Section \ref{sec:optimal:martingale}, when $\pi_{t_{k+1}}(1,1)\equiv0$ and some symmetry conditions are satisfied. 
 As explained above, the first condition is reasonable 
 when the trading frequency is high enough. 
The following result shows that when the initial inventory is $0$, then the expected inventory stays at $0$ at all future times. The  proofs of this part are deferred to the Appendix \ref{App:proof:inventory}.
\begin{proposition}\label{ExpFutInv}
{\Blue Suppose that the Assumptions of Theorem \ref{prop:measurability}, the conditions (1)-(2) of Proposition \ref{prop:admissibility}, and the condition of Lemma \ref{lemma:gfunc} are satisfied. We also assume that $\pi_{t_{k+1}}(1,1)\equiv0$ and} $\mu_{cp}^+=\mu_{cp}^-$ and $\mu_{c^2 p}^+=\mu_{c^2 p}^-$.
Then, {\Blue under the optimal placement strategy of Corollary \ref{cor:candidates}}, for any $k$, we have 
\begin{align}\label{Krel1a}
	{\rm (i)}\;\mathbb{E}\left[I_{t_{k+1}}|\mathcal{F}_{t_{k}}\right]=\left(1+2\pi_{t_{k+1}}\frac{\muo^2\alpha^{1}_{\tkk}}{\muo-\alpha^{1}_{\tkk}\mut}\right)I_{t_k},\qquad 
{\rm (ii)}\; \mathbb{E}[I_{t_k}]=0,
\end{align}
where, per the proof of Corollary \ref{cor:cond2:alt},  $\alpha^{1}_{\tkk}:=\alpha^{1+}_{\tkk}=\alpha^{1-}_{\tkk}$.
\end{proposition}

{\Blue Note that \eqref{Krel1a}-(ii) does not follow directly from \eqref{Krel1a}-(i) as in the nonadaptive case of \cite{zoe} because $\alpha^{1}_{\tkk}$ is random here and correlated to $I_{t_k}$. The derivation of \eqref{Krel1a}-(ii) requires an explicit representation of $I_{t_k}$ and heavily depends on the symmetry condition of Lemma \ref{lemma:gfunc}.}
Recalling from Lemma \ref{lemma:alpha} that $\alpha^{1}_{\tkk}\leq0$ and {\Blue since $\pi_{t_{k+1}}\mu_{c}^2\leq \mu_{c}^2\leq{}\mu_{c^2}$}, 
we have $-1\leq$ {$1+2\pi_{t_{k+1}}\frac{\muo^2\alpha^{1}_{\tkk}}{\muo-\alpha^{1}_{\tkk}\mut}\leq 1$.} Thus, by Equation  \eqref{Krel1a}-(i), it follows that
\begin{align}\label{InqHNo}
\Big|\mathbb{E}\left[I_{t_{k+1}}|\mathcal{F}_{t_{k}}\right]\Big|\leq |I_{t_k}|.
\end{align}
In particular, if $I_{t_{k}}>0$ ($I_{t_{k}}<0$), then $\mathbb{E}\left[I_{t_{k+1}}|\mathcal{F}_{t_{k}}\right]\leq I_{t_k}$ ($\mathbb{E}\left[I_{t_{k+1}}|\mathcal{F}_{t_{k}}\right]\geq I_{t_k}$), hence the inventory is mean-reverting toward $0$.

\subsection{Optimal MM strategy under running inventory penalization} \label{sec:invent:runningpen}
{\Blue In continuous-time settings, \cite{Guilbaud}, \cite{cartea2014buy}, and others have advocated for a running inventory penalty to further control the inventory risk. Under this control, the performance criterion is $\mathbb{E}\left[\left.W_T-S_TI_T-\lambda I_T^2-\phi\int_{t}^{T}I_s^2ds\right|\mathcal{F}_{t}\right]$. In discrete-time, the analogous value function at time $t_k$ naturally takes the form:
\begin{equation}\label{eqn:runn:penalty}
V_{t_k}:=\sup_{(L^\pm_{t_k},\dots,L^\pm_{t_N})\in  \mathcal{A}_{{t_k,t_N}}}\mathbb{E}\left[\left.W_T+S_T I_T -\lambda I_T^2-\phi\sum\limits_{j=k+1}^{N+1}I_{t_j}^2\right|\F\right], \quad  k=0,1,\dots,N+1,
\end{equation}
where as usual $\sum_{j=N+2}^{N+1}=0$. The dynamic programming principle corresponding to \eqref{eqn:runn:penalty} can then be written, for $k=0,\dots,N$,  as 
\begin{equation}\label{eqn:new:DPP}
    V_{t_k} = \sup_{L^\pm_{t_k} \in \mathcal{A}_{t_k}} \mathbb{E}\Big[V_{t_{k+1}}-\phi I_{t_{k+1}}^2\,\Big|\,\F\Big],
    \end{equation}
starting with $V_{t_{N+1}} =W_T+S_T I_T -\lambda I_T^2$. The heuristics behind \eqref{eqn:new:DPP} is classical: 
\begin{align*}
V_{t_k}
	&=\sup_{(L^\pm_{t_k},\dots,L^\pm_{t_N})\in  \mathcal{A}_{{t_k,t_N}}}\mathbb{E}\Big[\mathbb{E}\Big[W_T+S_T I_T -\lambda I_T^2-\phi\sum\limits_{j=k+2}^{N+1}I_{t_j}^2\Big|\mathcal{F}_{t_{k+1}}\Big]-\phi I_{t_{k+1}}^2\Big|\F\Big]\\
	&=\sup_{L^\pm_{t_k} \in \mathcal{A}_{t_k}}\mathbb{E}\Big[\sup_{(L^\pm_{t_{k+1}},\dots,L^\pm_{t_N})\in  \mathcal{A}_{{t_{k+1},t_N}}}\mathbb{E}\Big[W_T+S_T I_T -\lambda I_T^2-\phi\sum\limits_{j=k+2}^{N+1}I_{t_j}^2\Big|\mathcal{F}_{t_{k+1}}\Big]-\phi I_{t_{k+1}}^2\Big|\F\Big]\\
	&=\sup_{L^\pm_{t_k} \in \mathcal{A}_{t_k}} \mathbb{E}\Big[V_{t_{k+1}}-\phi I_{t_{k+1}}^2\,\Big|\,\F\Big].
\end{align*}
The following result, whose proof is deferred to Appendix \ref{App:run:inv}, formalizes the heuristics above. Specifically, this shows that, under \eqref{eqn:runn:penalty}, most of the results considered in this paper follow with minor modifications.}
\begin{theorem} \label{lemma:runn:penalty}
{\Blue Suppose that the setting and assumptions of Theorem \ref{prop:measurability} and Corollary \ref{cor:candidates} hold true. Then, the following statements hold:
\begin{enumerate}
	\item The conclusions of Theorem \ref{prop:measurability} and Corollary \ref{cor:candidates} follow with \eqref{eqn:new:DPP} replacing \eqref{eqn:optimalprob} and an ansatz of the form
\begin{equation}\label{eqn:ansatzPhi}	
	V_{t_k}=W_\tk+\sk I_\tk+\,{}^{\scaleto{(\phi)}{5pt}}\!\alpha_{\tk}I^2_\tk+\,{}^{\scaleto{(\phi)}{5pt}}\!h_{\tk}I_\tk + \,{}^{\scaleto{(\phi)}{5pt}}\!g_{t_k}.
\end{equation}
The optimal controls ${}^{\scaleto{(\phi)}{5pt}}\!L_{\tk}^{\pm,*}$, $k=0,\dots,N$ are then given by
	\begin{align}\label{eqn:new:OptimalL}
	\begin{split}
	{}^{\scaleto{(\phi)}{5pt}}\!L_{\tk}^{+,*}&={}^{\scaleto{(1)}{5pt}}\!A^{+,\phi}_{\tk}I_\tk+{}^{\scaleto{(2)}{5pt}}\!A^{+,\phi}_{\tk}+{}^{\scaleto{(3)}{5pt}}\!A^{+,\phi}_{\tk},\\
	{}^{\scaleto{(\phi)}{5pt}}\!L_{\tk}^{-,*}&=-\,{}^{\scaleto{(1)}{5pt}}\!A^{-,\phi}_{\tk}I_\tk-{}^{\scaleto{(2)}{5pt}}\!A^{-,\phi}_{\tk}+{}^{\scaleto{(3)}{5pt}}\!A^{-,\phi}_{\tk},
	\end{split}
	\end{align}
with the coefficients ${}^{\scaleto{(1)}{5pt}}\!A^{\pm,\phi}_{\tk}$, ${}^{\scaleto{(2)}{5pt}}\!A^{\pm,\phi}_{\tk}$, and ${}^{\scaleto{(3)}{5pt}}\!A^{\pm,\phi}_{\tk}$ taking the same form as \eqref{eq:A1}, but with $\alpha^{1\pm}_{\tkk}$ and $\alpha^{1,1}_\tkk$ replaced with ${}^{\scaleto{(\phi)}{5pt}}\!\alpha^{1\pm}_{\tkk}-\phi$ and ${}^{\scaleto{(\phi)}{5pt}}\!\alpha^{1,1}_\tkk-\phi$. Here, ${}^{\scaleto{(\phi)}{5pt}}\!\alpha_{\tk}$, ${}^{\scaleto{(\phi)}{5pt}}\!h_{\tk}$, and ${}^{\scaleto{(\phi)}{5pt}}\!g_{\tk}$ follow the same formulas as \eqref{eqn:alphak}-\eqref{eq:g}, but with $\alpha^0_{\tkk}$, $\alpha^{1\pm}_{\tkk}$, $\alpha^{1,1}_\tkk$, $h^0_{\tkk}$, $g^0_{\tkk}$, and  ${}^{\scaleto{(\ell)}{5pt}}\!A^{\pm}_{\tk}$ replaced with ${}^{\scaleto{(\phi)}{5pt}}\!\alpha^{0}_{\tkk}-\phi$, ${}^{\scaleto{(\phi)}{5pt}}\!\alpha^{1\pm}_{\tkk}-\phi$, ${}^{\scaleto{(\phi)}{5pt}}\!\alpha^{1,1}_\tkk-\phi$,  ${}^{\scaleto{(\phi)}{5pt}}\!h^0_{\tkk}$, ${}^{\scaleto{(\phi)}{5pt}}\!g^0_{\tkk}$, and ${}^{\scaleto{(\ell)}{5pt}}\!A^{\pm,\phi}_{\tk}$, respectively, on the right-hand side of all the formulas. 
\item The conclusion of Lemma \ref{lemma:alpha} also follows under the new running inventory penalty; i.e., 
\[
	{}^{\scaleto{(\phi)}{5pt}}\!\alpha_\tkk^{0}<{}^{\scaleto{(\phi)}{5pt}}\!\alpha_\tk<0,\quad \text{for any }\;k\in\{0,1,\ldots,N\}.
\]	
\item The verification Theorem \ref{VeriThrm1} holds true with the value function (\ref{OOCP0}) replaced with \eqref{eqn:runn:penalty}.
\end{enumerate}}
\end{theorem}

\subsection{Implementing the Optimal Strategy}\label{sec222}

{As indicated in Remark \ref{AdpHistMMD}, the function $g$ introduced in \eqref{eqn:gfunc:pi} has two main purposes. First, it summarizes the information contained in the recent history of MOs, $\pmb{e}_{t_{k}}=(\mathbbm{1}_{t_{k}}^{\pm},\dots, \mathbbm{1}_{t_{k-\varpi+1}}^{\pm})$, and, more importantly, it alleviates the computational burden by reducing the dimension of the different scenarios one needs to consider. However, in order for $g$ to work as intended, we must impose one additional condition:
\begin{assumption}\label{LACndOnal}
For $k=0,\dots,N$, $g(\etkk)$ is a function of $(\onep,\onem,g(\etk))$; i.e., if we denote the image of $g$ as $\mathcal{I}_g$, then there exists a function $\Gamma:\{0,1\}^2\times \mathcal{I}_g\to\mathcal{I}_g$ such that 
\begin{align}\label{AnthrCndBL}
	g(\etkk)=\Gamma(\onep,\onem,g(\etk)).
\end{align}
\end{assumption}
Obviously, if we picked $g$ to be the identity function ($g(\etk)=\etk$) so that there is no dimension reduction, then \eqref{AnthrCndBL} is trivially satisfied by taking $\Gamma$ to be the mapping that drops the last two coordinates of the vector $(\onep,\onem,g(\etk))$. A more interesting example is when $g(\etkk)=\pmb{e}_{\tkk}^++\pmb{e}_{\tkk}^-$ (which satisfies the conditions for admissibility stated in Lemma \ref{lemma:gfunc}). In that case, denoting the mapping that drops the last two entries in a vector as $\Gamma_{-}$, we have:
\[
	g(\etkk)=(\onep+\onem,\Gamma_{-}(\pmb{e}_{\tk}^++\pmb{e}_{\tk}^-))=(\onep+\onem,\Gamma_{-}(g(\pmb{e}_{\tk}))),
\]
and the condition of Assumption \ref{LACndOnal} is satisfied.

Working backward by induction, it is not hard to check that, under Assumption \ref{LACndOnal}, the coefficients ${}^{\scaleto{(1)}{5pt}}\!A^+_{\tk}$, ${}^{\scaleto{(2)}{5pt}}\!A^+_{\tk}$, and ${}^{\scaleto{(3)}{5pt}}\!A^+_{\tk}$ for the optimal placement strategy \eqref{eqn:OptimalL} depend only on $g(\pmb{e}_{\tk})$. 
Indeed, the probabilities $\mathbbm{P}(\onep=i,\onem=j|\F) $ have this property since these can be expressed in terms of $\pi^\pm_{\tkk}$ and $\pi_{\tkk}(1,1)$ (see \eqref{eq1234}), which enjoy the stated property per Assumption \ref{assump:cp}-(iv). So, it only remains to show that ${\alpha}_\tkk^0$, ${\alpha}_\tkk^{1\pm}$, and ${\alpha}_\tkk^{i,j}$ (and the corresponding quantities for $h_{t_{k+1}}$) have this property. Indeed, by the induction step, we can assume that ${\alpha}_\tkk=\Lambda(g(\etkk))$, for some function $\Lambda: \mathcal{I}_{g}\to\mathbb{R}$, and, by Assumption \ref{LACndOnal}, ${\alpha}_\tkk=\Lambda(\Gamma(\onep,\onem,g(\etk))$. Then, similar to \eqref{CndCmpAlph}, 
\begin{equation}\label{CndCmpAlphcc}
\begin{aligned}
\alpha_\tkk^{1+} 
&=\sum_{\ell\in\{0,1\}}\Lambda(\Gamma(1,
\ell,g(\etk))\frac{\mathbb{P}\big[\onep=1,\onem=\ell\big|\F\big]}{\mathbb{P}\big[\onep=1\big|\F\big]},
\end{aligned}
\end{equation}
which is clearly a function of $g(\pmb{e}_{\tk})$. We can similarly deal with ${\alpha}_\tkk^0$, ${\alpha}_\tkk^{i,j}$, $h_\tkk^0$, $h_\tkk^{1\pm}$, and $h_\tkk^{i,j}$.

To carry out the optimal placement strategy, we think of each value $\pmb{\iota}\in\mathcal{I}_g$, the image of $g$, as a possible scenario of the immediate history of MOs. Before the beginning of trading, the MM first computes, backward in time using Eqs.~\eqref{eqn:gamma:eta}-\eqref{eqn:alphak}, the coefficients ${}^{\scaleto{(1)}{5pt}}\!A^\pm_{\tk}$, ${}^{\scaleto{(2)}{5pt}}\!A^\pm_{\tk}$, and ${}^{\scaleto{(3)}{5pt}}\!A^\pm_{\tk}$ of the optimal placement strategy \eqref{eqn:OptimalL} for each time $t_k$ and each possible scenario $\pmb{\iota}\in\mathcal{I}_g$. This results in a type of ``dictionary" or ``catalog". Once the catalog is computed, she can then start her trading moving forward in time. At each time $t_k$, she observes the recent history of MOs  $\pmb{e}_{t_{k}}=(\mathbbm{1}_{t_{k}}^{\pm},\dots, \mathbbm{1}_{t_{k-\varpi+1}}^{\pm})$. Based on the $\pmb{e}_{t_{k}}$, the inventory value $I_{t_k}$, the asset price $S_{t_k}$, and her forecasts $\Delta_{t_j}^\tk$ of future price changes, she places her bid and ask LOs using the catalog. For instance, if she assumes $\Delta_{t_i}^{t_k}=0$, for all $i\geq k+1$, she will place her orders at
	\begin{align}\label{DfnTldL0bc}
\begin{split}	
		\widetilde{a}_{t_k}^{*}&=S_{t_{k}}+{}^{\scaleto{(1)}{5pt}}\!A^+_{\tk}I_\tk+{}^{\scaleto{(2)}{5pt}}\!A^+_{\tk}+{{}^{\scaleto{(3)}{5pt}}\!A^+_{\tk}}+\frac{\rho_k^+-\psi_k^+}{2\gamma_{t_{k}}}\Delta_{t_k}^{t_k},\\
	\widetilde{b}_{t_k}^{*}&:=S_{t_{k}}+{}^{\scaleto{(1)}{5pt}}\!A^-_{\tk}I_\tk+{}^{\scaleto{(2)}{5pt}}\!A^-_{\tk}-{}^{\scaleto{(3)}{5pt}}\!A^-_{\tk}+\frac{\rho_k^--\psi_k^-}{2\gamma_{t_{k}}}\Delta_{t_k}^{t_k},
\end{split}
	\end{align}
where above we are also assuming that the coefficients of $\Delta_{t_k}^{t_k}$ are computed at the beginning of the trading for each time $t_k$ and each scenario $\pmb{\iota}\in\mathcal{I}_g$.}
	

\section{Calibration and Testing the Optimal Strategy on LOB Data} \label{sec:Data}
{In this section, we give further details about the implementation of the optimal placement strategy of Section \ref{Sec:General:process}, including model calibration. We then illustrate our approach with real LOB data from Microsoft Corporation (stock symbol MSFT) during the year of 2019\footnote{{\Blue We have tried a few other stocks but the results are not shown here for the sake of space.}}. {Our data set is obtained from Nasdaq TotalView-ITCH 5.0, which is a direct data feed product offered by The Nasdaq Stock Market, LLC\footnote{\tt{http://www.nasdaqtrader.com/Trader.aspx?id=Totalview2}}. TotalView-ITCH uses a series of event messages to track any change in the LOB. For each message, we observe the timestamp, type, direction, volume, and price. We reconstruct the dynamics of the top 20 levels of the LOB directly from the event message data. We treat each day as an independent sample.}

In the first subsection 
we detail the model training or parameter estimation procedure. The training will be based on the historical LOB data of the 20 days prior to each testing day\footnote{{\Blue We tried different windows. The performance is good provided that the window size is not too small (say, 2 or 5 days).}}.
 In the second subsection, we present the performance of the optimal placement policy using the real flow of orders for MSFT in a given testing day and compare it with ``fix-placement" strategies that place LOs at fixed  ask (bid) price levels. More specifically, in each test day, we assume the MM places her LOs every second from 10:00 am to 3:30 pm, placing a total of 19800 LOs at each side of the book.

\subsection{Parameter Estimation}\label{EstPisH}

\subsubsection{Estimation of \texorpdfstring{$\pi_\tkk^\pm, \pi_\tkk(1,1)$}.} \label{EstPiSS0}

We first need to specify the function $g$ in Assumption \ref{assump:cp}-(iv). We consider the following three functions:
\begin{align}\label{eqn:choice:g}
g_1: \{0,1\}^{6} \to \{0,1\}^{6},\quad&   g_1 (\pmb{e}_\tk)=\pmb{e}_\tk=(\mathbbm{1}_{t_{k}}^+,\mathbbm{1}_{t_{k}}^-, \mathbbm{1}_{t_{k-1}}^+ , \mathbbm{1}_{t_{k-1}}^-, \mathbbm{1}_{t_{k-2}}^+ , \mathbbm{1}_{t_{k-2}}^-),%
\nonumber \\
g_2: \{0,1\}^{6} \to \{0,1,2\}^{3},&\quad   g_2 (\pmb{e}_\tk)=(\mathbbm{1}_{t_{k}}^++\mathbbm{1}_{t_{k}}^-, \mathbbm{1}_{t_{k-1}}^+ + \mathbbm{1}_{t_{k-1}}^-, \mathbbm{1}_{t_{k-2}}^+ + \mathbbm{1}_{t_{k-2}}^-),\\
g_3: \{0,1\}^{8} \to \{0,1,2\}^{4},&\quad   g_3 (\pmb{e}_\tk)=(\mathbbm{1}_{t_{k}}^++\mathbbm{1}_{t_{k}}^-, \mathbbm{1}_{t_{k-1}}^+ + \mathbbm{1}_{t_{k-1}}^-, \mathbbm{1}_{t_{k-2}}^+ + \mathbbm{1}_{t_{k-2}}^-,\mathbbm{1}_{t_{k-3}}^+ + \mathbbm{1}_{t_{k-3}}^-).\nonumber
\end{align}
The first function does not satisfy the conditions for admissibility of Lemma \ref{lemma:gfunc}. However, in practice, for each of the functions above, it is very rare that $a_{t_k}<b_{t_k}$ or that $a_{t_k}$ ($b_{t_k}$) is below (above) the midprice $S^{mid}_{t_{k}}$. When any of these events happen, we simply set $a_{t_k}$ and/or $b_{t_k}$ at the tick right above and below the midprice depending on what is appropriate.

Once we have selected the function $g$, we use the historical LOB data of the 20 days prior to each testing day to estimate the functions $f$, $f^+$, and $f^-$ in Assumption \ref{assump:cp}-(iv). To this end, we simply leverage the interpretations of those functions as conditional probabilities.
 Specifically, setting $\mathcal{I}_{g}:=\text{Im}(g)\subset\R^d$, for every $\pmb{\iota}\in\mathcal{I}_g$, we estimate $f^\pm(\pmb{\iota})$ and $f(\pmb{\iota})$ as
\begin{align}\label{eqn:estimation:f1}
\widehat{f}^\pm(\pmb{\iota})&=\dfrac{\#\{\mathbbm{1}_{t_\ell}^\pm=1,g(\pmb{e}_{t_\ell})=\pmb{\iota}\}}{\#\{g(\pmb{e}_{t_\ell})=\pmb{\iota}\}},\\ \label{eqn:estimation:f2}
\widehat{f}(\pmb{\iota})&=\dfrac{\#\{\mathbbm{1}_{t_\ell}^+=1,\mathbbm{1}_{t_\ell}^-=1,g(\pmb{e}_{t_\ell})=\pmb{\iota}\}}{\#\{g(\pmb{e}_{t_\ell})=\pmb{\iota}\}},
\end{align}
where $\#A$ indicates the cardinality of a set $A$ and the times $t_\ell$'s {range} over all the seconds from 10:00 am to 3:30 pm in the 20 days prior to the testing day. 
The results of the estimation procedure for $g_1,g_2$ and $g_3$ can be found in Figure \ref{fig:OptimalPaths} below. The overall behavior is what one will expect: when $\pmb\iota$ takes a value corresponding to more ones in $\pmb{e}_{t_\ell}$, the probabilities $\pi^{\pm}$ take larger values. {These figures also indicate that our symmetry assumption $\pi^+=\pi^-$ required for admissibility (see Proposition \ref{prop:admissibility}) is reasonable.}

\bigskip
\begin{figure}[H]
	\centering
	\begin{subfigure}[b]{1\textwidth}
		\centering
		\includegraphics[width=.6\textwidth]{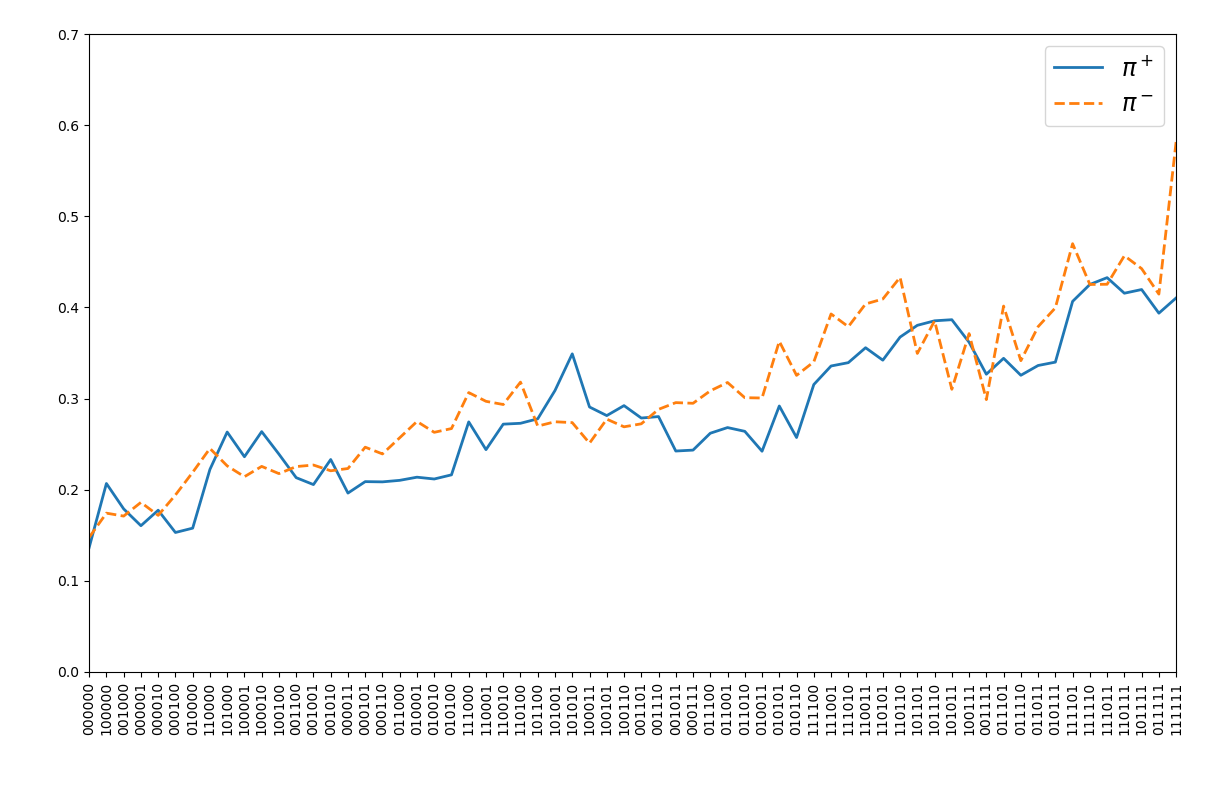}
		\caption{$\pi^\pm$ as functions of $g_1$.}
		\label{fig:PricePath} 
	\end{subfigure}
	
	\begin{subfigure}[b]{1\textwidth}
		\centering
		\includegraphics[width=.6\textwidth]{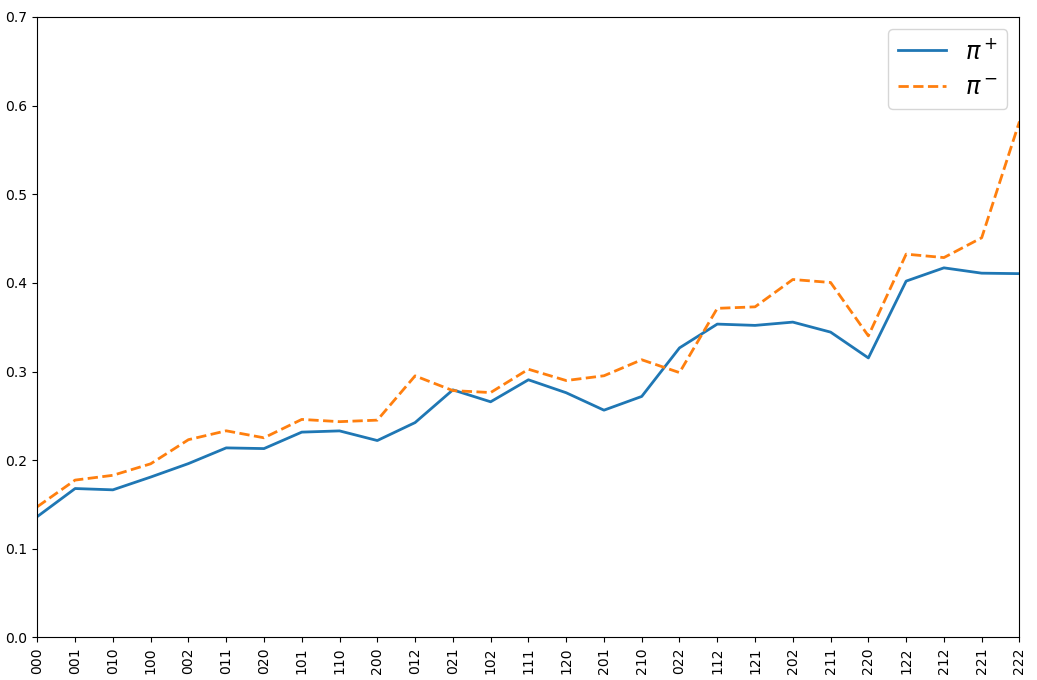}
		\caption{$\pi^\pm$ as functions of $g_2$.}
		\label{fig:InvPath}
	\end{subfigure}
	
	\begin{subfigure}[b]{1\textwidth}
		\centering
		\includegraphics[width=.8\textwidth]{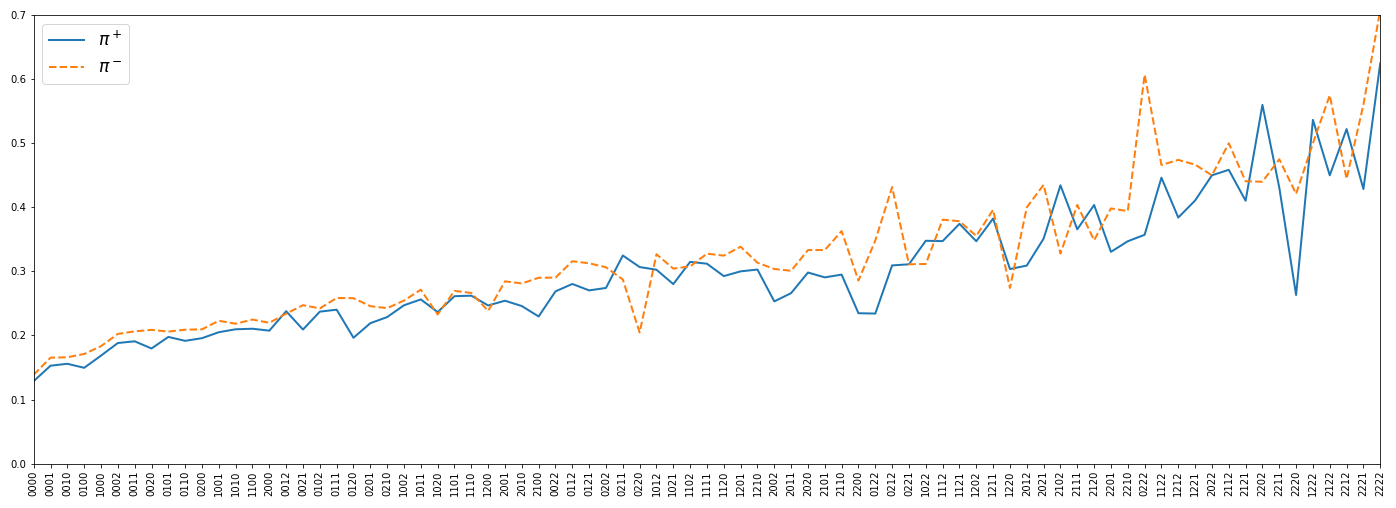}
		\caption{$\pi^\pm$ as functions of $g_3$.}
		\label{fig:InvPath0}
	\end{subfigure}
	\caption{Estimation of $\pi^{\pm}$ under the different choices of the function $g$ given in \eqref{eqn:choice:g}. The $x$-axis displays the equivalence classes where the function $g$ assumes different values, whereas the $y$-axis represents the value of $\pi^{\pm}$.}
	\label{fig:OptimalPaths}
\end{figure}
%
%
%
%

\begin{remark}
One of the key principles behind our approach is the presumption that we could forecast {to a good degree} the intensity of MOs through the day using the estimated functions \eqref{eqn:estimation:f1} and \eqref{eqn:estimation:f2} and the history of previous MOs, $\pmb{e}_\tk$. To assess the validity of this principle, we compare the average of the sets
\begin{align*}
\mathcal{M}_{k}^+&=\{\mathbbm{1}_{t_{k-250}}^+,\mathbbm{1}_{t_{k-249}}^+,\ldots,\mathbbm{1}_{t_k}^+,\ldots,\mathbbm{1}_{t_{k+249}}^+,\mathbbm{1}_{t_{k+250}}^+\}{,}\\
\mathcal{P}_{k}^+&=\{\hat{\pi}^+_{t_{k-250}},\hat{\pi}^+_{t_{k-249}},\ldots,\pi^+_{t_k},\ldots,\pi^+_{t_{k+249}},\pi^+_{t_{k+250}}\},
\end{align*}
where $\hat\pi_{t_{\ell}}^+:=\widehat{f}^+(g_1(\mathbbm{1}_{t_{\ell}}^+,\mathbbm{1}_{t_{\ell}}^-, \mathbbm{1}_{t_{k-1}}^+ , \mathbbm{1}_{t_{\ell-1}}^-, \mathbbm{1}_{t_{\ell-2}}^+ , \mathbbm{1}_{t_{\ell-2}}^-))$. 
Figure \ref{fig:MovAvg0807} below shows the result for all the seconds $t_k$ in a prototypical day. This shows that our approach is surprisingly accurate {in tracking the} intensity of MO's throughout the day using historical data from previous days and past information of MOs.
\end{remark}

\begin{figure}[h]
	\centering
	\vspace{.5 cm}
	\includegraphics[width=.8\textwidth]{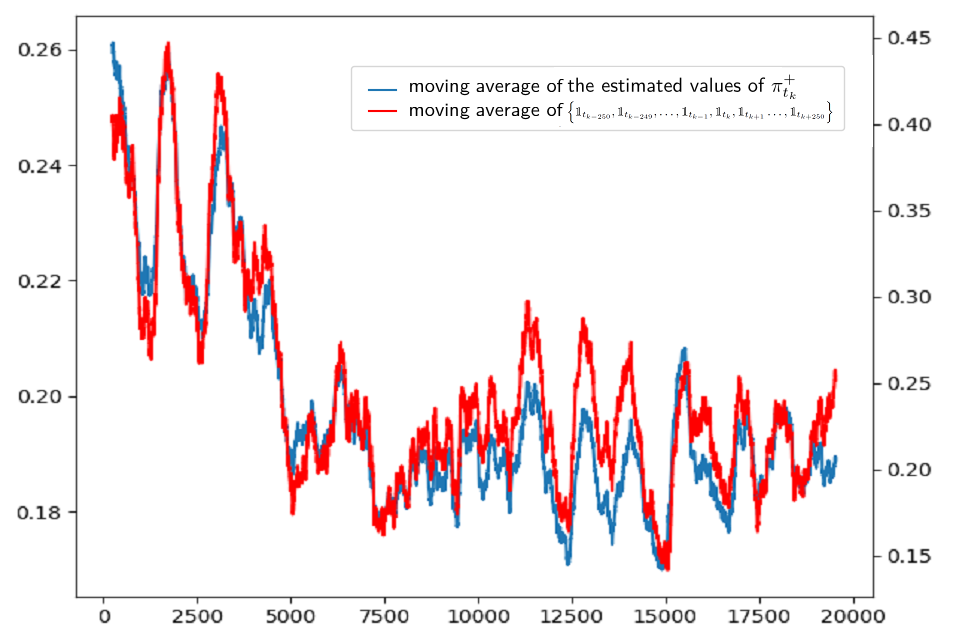}
	\caption{Moving average of the estimated values of $\mathcal{P}^+$ (using the function $g=g_1$) and the moving average of $\mathcal{M}_{k}^+$ on August 7th, {\Blue 2019}.}\label{fig:MovAvg0807}
	\vspace{.5 cm}
\end{figure}


\subsubsection{Estimation of the Demand Functions.}\label{DmdFncEsD} To estimate the constants $\mu_{c^mp^n}^{\pm}$ of Eq.~\eqref{eqn:mucp},  upon which our strategy depends on, we use the sample averages within the previous 20 trading days. For instance, the estimate of $\mu^+_{c^m p^n}$ is 
\begin{equation}\label{EstCntAH}
	\hat{\mu}^+_{c^m p^n}:=\frac{1}{N}\sum_{\ell=0}^{N-1}(\hat{c}^+_{t_{\ell+1}})^m(\hat{p}^+_{t_{\ell+1}})^n,
\end{equation}
where $t_{\ell+1}$ ranges over all the seconds of the previous 20 days from 10:00 am to 3:30 pm and $N$ is the total number of those. To estimate $(\hat{c}^+_{t_{\ell+1}},\hat{p}^+_{t_{\ell+1}})$ for one of those previous 1-second time {intervals $[t_\ell,t_{\ell+1})$}, we apply the following procedure. Assume that the MM places an ask LO at price level $P_\ell$ at time $t_\ell$ {\Blue with volume $V_{LO}$}, and that the volume of existing {\Blue ask LOs} with prices lower than $P_\ell$ is $V_L$. If a buy MO with volume $V_M$ arrives during $[t_\ell,t_{\ell+1})$, then the number of shares of the MM's LO to be filled equals to ${\Blue ((V_M-V_L)\vee0)\wedge V_{LO}}$\footnote{Here, for computational simplicity, we are assuming the MM's  LO at level $P_\ell$ is ahead of the queue (hence, her shares are the first to be filled at that level), which is a common simplification in the literature.}. We then compute this quantity for all buy MOs arriving during the interval $[t_{\ell},t_{\ell+1})$ so that $Q_\ell:=\sum_{\text{all MO}}\Big({\Blue ((V_M-V_L)\vee0)\wedge V_{LO}}\Big)$ will quantify the actual demand at price level $P_\ell$ during that interval. Once those demands have {been} computed for all price level $P_\ell$ above the midprice, we {performed} a weighted linear regression to estimate $(\hat{c}^+_{t_{\ell+1}},\hat{p}^+_{t_{\ell+1}})$, with the actual demand $Q_\ell$ being the response variable and the price level $P_\ell$ being the predictor, placing higher {weights} on price levels closer to the midprice $S_\tk$. Following Eq.~(\ref{pdemand}), we can estimate $\hat{c}^+_{t_\ell}$ and $\hat{p}^+_{t_\ell}$ as the slope and as the quotient intercept/slope of the regression line, respectively. In the preprint \cite[p. 28]{zoe}, it is shown that time series $(\hat{c}^+_{t_\ell},\hat{p}^+_{t_\ell})$ are reasonably stationary, implying that our method to estimate the constants $\mu_{c^mp^n}^{\pm}$ as \eqref{EstCntAH} is justifiable.

In Figure \ref{DemandPlot0} below, we plot the average demand curve {and the regression line whose slope and $p$ value is set to the averages of the $(\hat{c}^+_{t_{\ell+1}},\hat{p}^+_{t_{\ell+1}})$'s of all 1-second time intervals $[t_\ell,t_{\ell+1})$ during that date. This graph shows} that the linear model in Eqs. (\ref{pdemand})-(\ref{mdemand}) is a reasonably good approximation of the actual volume of shares executed, especially as they are closer to the midprice where most MOs are executed.

\begin{figure}[h]
	\centering
	\includegraphics[width=.6\textwidth]{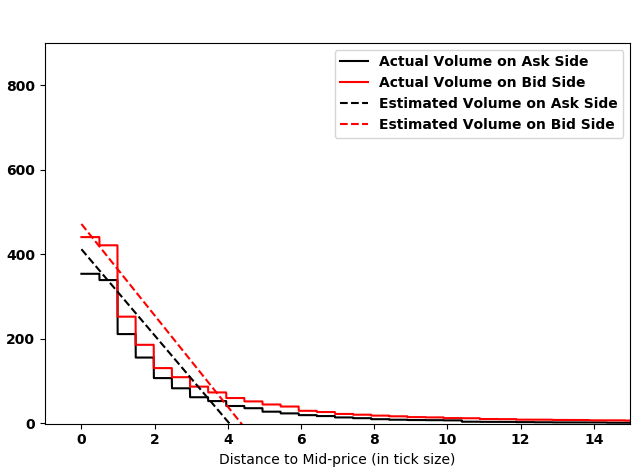}
	\caption{Plot of the actual demand on October 3rd vs. the estimated linear demand function over a 1-second trading interval}%
	\label{DemandPlot0}
	\vspace{.5 cm}
\end{figure}

To give an idea of the values of $\mu^\pm_{c^m p^n}$, we estimate those constants  for each day of the 252 days of our sample (using an estimator like that in \eqref{EstCntAH} but with the $t_{\ell}$'s ranging over all the seconds of each day) and then we take the averages of the resulting 252 estimates $\hat{\mu}^\pm_{c^m p^n}$. Table \ref{tab1} shows the results. The table also shows some {other related quantities} to assess the validity of Eqs.~(\ref{Cnd1PosSpr}) and (\ref{Cnd3PosSpr}) in Proposition \ref{prop:admissibility}. As shown therein, these assumptions are reasonably met in our sample data.

\medskip

\renewcommand{\arraystretch}{1.5}

\begin{table}[ht]
	\centering
	\setlength{\tabcolsep}{10pt}

	\begin{tabular}{|rcl|c|rcl|}
		\cline{1-3}\cline{5-7}
		\textcolor{RedOrange}{$\hat{\mu}_c^+$} &=& \textcolor{RedOrange}{$125.512$} &\qquad& \textcolor{RedOrange}{$\hat{\mu}_c^- $}&=& \textcolor{RedOrange}{$130.622$} \\ \cline{1-3}\cline{5-7}
		\textcolor{ForestGreen}{$\hat{\mu}_p^+$} &=& \textcolor{ForestGreen}{$3.287$} &\qquad& \textcolor{ForestGreen}{$\hat{\mu}_p^- $}&=& \textcolor{ForestGreen}{$3.292$} \\ \cline{1-3}\cline{5-7}
		\textcolor{WildStrawberry}{$\hat{\mu}_c^+ \hat{\mu}_p^+$} &=&\textcolor{WildStrawberry}{$412.558$} &\qquad& \textcolor{TealBlue}{$\hat{\mu}_c^-\hat{\mu}_p^-$} &=& \textcolor{TealBlue}{$430.008$}\\ \cline{1-3}\cline{5-7}
		\textcolor{WildStrawberry}{$\hat{\mu}_{cp}^+$} &=& \textcolor{WildStrawberry}{$451.263$} &\qquad& \textcolor{TealBlue}{$\hat{\mu}_{cp}^-$} &=& \textcolor{TealBlue}{$471.685$} \\ \cline{1-3}\cline{5-7}
		\textcolor{Sepia}{$\hat{\mu}_{c^2}^+$} &=& \textcolor{Sepia}{$8.47\times 10^4$} &\qquad& \textcolor{Sepia}{$\hat{\mu}_{c^2}^-$} &=& \textcolor{Sepia}{$5.45\times 10^4$} \\ \cline{1-3}\cline{5-7}
		\textcolor{Fuchsia}{$\hat{\mu}_{c^2p}^+$} &=& \textcolor{Fuchsia}{$3.53\times 10^5$} &\qquad& $\hat{\mu}_{c^2p}^-$ &=& $2.25\times 10^5$ \\ \cline{1-3}\cline{5-7}
		\textcolor{Fuchsia}{$\hat{\mu}_{c^2}^+\hat{\mu}_p^+$} &=& \textcolor{Fuchsia}{$2.78\times 10^5$} &\qquad& $\hat{\mu}_{c^2}^-\hat{\mu}_p^-$ &=& $1.79\times 10^5$ \\ \cline{1-3}\cline{5-7}
	\end{tabular}

	\vspace{.2cm}
	\caption{Average values of $c^\pm$ and $p^\pm$ over 252 trading days in 2019. Entries with the same color should be of a similar magnitude to have our model assumptions validated. \bigskip}
	\label{tab1}
\end{table}

\subsubsection{Estimation of the Drift for the Midprice Process.} 
{For simplicity}, we set the fundamental price to be the midprice process. For our implementation, we assume that 
\begin{equation*}
\Delta_{t_j}^{t_k}=\E(S_{t_{j+1}}-S_{t_j}|\mathcal{F}_{t_k})=0, \quad j\geq k+1,
\end{equation*}
because in practice, one could expect $\Delta_{t_{j}}^{t_k}=\E(S_{t_{j+1}}-S_{t_{j}}|\mathcal{F}_{t_k})$ to quickly decrease to $0$ as $j$ is farther away from $k$ (otherwise, statistical arbitrage opportunities are likely appear) and also because the estimation error of the forecasts $\Delta_{t_j}^{t_k}$ increases quickly as $t_{j}$ is farther away from $t_k$. To estimate the one-step ahead forecast $\Delta_\tk = \E(S_\tkk-S_\tk|\F)$,  we {simply} take the average over the last 5 increments:
\begin{equation}\label{OSFN}
    \hat{\Delta}_{t_k}= \frac{1}{5}\sum_{i = 1}^5 (S_{t_{k-i+1}}-S_{t_{k-i}})=\dfrac{S_{t_{k}}-S_{t_{k-5}}}{5}.
\end{equation}

\subsection{Numerical Results} \label{NumRsltSect}

In this subsection we {\Blue illustrate the performance of} the optimal trading strategy using as $g$ function each of the functions in \eqref{eqn:choice:g}. For each of these choices we compute the terminal value of the performance criterion $W_T+S_T I_T -\lambda I_T^2$ under the martingale assumption. {\Blue We also compute the terminal value when applying the placement strategy \eqref{DfnTldL0} with $\Delta_{t_k}^{t_k}$ estimated as \eqref{OSFN} and using the function $g_3$}. As mentioned above, our optimal placement strategy was tested based on the LOB data of MSFT observed in the year 2019. As suggested in the preprint \cite[Section 5.2]{zoe}, we set $\lambda=0.0005$, which gives a good estimate of the average liquidity cost for our sample {data. 
The} cash holding $W_\tk$ and the stock inventory $I_\tk$ were computed as
\begin{align*}
W_{t_{k+1}} - W_\tk = a_\tk \widetilde{Q}_\tkk^+-b_\tk \widetilde{Q}_\tkk^-,\qquad
I_{t_{k+1}}-I_\tk = -\widetilde{Q}_\tkk^++\widetilde{Q}_\tkk^-,
\end{align*}
where $a_\tk$ and $b_\tk$ were the price of selling (buying) LOs placed at time $\tk$ and $\widetilde{Q}_\tkk^+$ ($\widetilde{Q}_\tkk^-$) was the executed volume of selling (buying) LOs calculated by using the actual flow of {MOs} in the market in each one-second interval of each testing day\footnote{{Here, we are again assuming for simplicity that the MM's  LO is ahead of the queue.}}. Since all the parameters of the model, namely the filling probabilities $\pi_\tkk^\pm, \pi_\tkk(1,1)$ and the constants $\mu_{c^mp^n}^{\pm}$, are calibrated using the past 20 days to each testing day, the first testing day is set to be January $30^{\text{th}}$, which was the $21^{\text{st}}$ trading day of 2019.    

\subsubsection{Performance Criterion Distribution.} 
The {sample means and standard deviations of the  end of the day performance criterion $W_T+S_TI_T-\lambda I_T^2$ for the 232 testing days under the four different implementations of our optimal placement strategy} are shown in Table~\ref{tab:TerminalRewardb}. 
For comparison, we also computed the {corresponding values when using} 6 deterministic strategies labeled `Level 1'- `Level 6', where the strategy `Level $i$' is the one where the {MM} always posts her orders $i$-ticks deep in the order book at both {sides. 
To} complement Table~\ref{tab:TerminalRewardb}, we also present Figure \ref{fig4} below, where we display the histogram of performance criterion $W_T+S_TI_T-\lambda I_T^2$ for the 232 testing days. {Table \ref{tab:TerminalRewarda} presents the means and standard deviations of the terminal values $W_T+\bar S_TI_T$, computed using the actual average price $\bar S_T$ per share that the HFM would get when liquidating her inventory $I_T$ with a MO based on the state of the book at time $T$. We refer to $\bar S_T I_T$ as the liquidation proceeds. We do not observe significant differences with the results presented in Table \ref{tab:TerminalRewardb}, which validates our assumption of modeling the liquidation cost as $S_TI_T-\lambda I_T^2$ and the chosen penalization value of $\lambda=0.0005$}.

\begin{table}[t]
	\centering
	\setlength{\tabcolsep}{8pt}
	\resizebox{\textwidth}{!}{%
		\begin{tabular}{@{}ccccccccc@{}}
			\toprule\midrule
			& \multicolumn{2}{c}{\begin{tabular}{c}Optimal Strategy\\ with Martingale Midprice\\ Conditioning on $g_1(\pmb{e}_\tk)$\end{tabular}} & \multicolumn{2}{c}{\begin{tabular}{c}Optimal Strategy\\ with Martingale Midprice\\ Conditioning on $g_2(\pmb{e}_\tk)$\end{tabular}}&
			\multicolumn{2}{c}{\begin{tabular}{c}Optimal Strategy\\ with Martingale Midprice\\ Conditioning on $g_3(\pmb{e}_\tk)$\end{tabular}} &
			\multicolumn{2}{c}{\begin{tabular}{c}Optimal Strategy\\ with General Midprice\\ Conditioning on $g_3(\pmb{e}_\tk)$\end{tabular}}\\ \midrule
			Mean & \multicolumn{2}{c}{$8.36\times10^4$} & \multicolumn{2}{c}{$8.50\times10^4$}& \multicolumn{2}{c}{{ $8.57\times10^4$}} &\multicolumn{2}{c}{{ $8.26\times10^4$}}\\
			\midrule
			Std. & \multicolumn{2}{c}{$1.63\times10^6$} & \multicolumn{2}{c}{$1.63\times10^6$} & \multicolumn{2}{c}{{ $1.62\times10^6$}}  &\multicolumn{2}{c}{{ $1.37\times10^6$}}\\ \toprule\midrule
			& Level 1             & Level 2             & Level 3            & Level 4              & Level 5              & Level 6             & & \\ \midrule
			Mean & $-7.78\times10^{6}$        &  $-9.99\times10^{5}$       &    $	-1.14\times10^5$     &   $-3.64\times10^4$      &     $-5.16\times10^4 $   &$-3.69\times10^4  $                  & &     \\ \midrule
			Std. & $1.52\times10^{7}$        &   $4.49\times10^{6}$      &   $2.01\times10^{6}$      &   $1.07\times10^{6}$      &    $7.16\times10^{5}$     &$4.81\times10^{5}$                     & &     \\ \bottomrule
		\end{tabular}%
	}
	\vspace{1.5mm}
	\caption{{\DRed Sample Mean and Std. Dev. of the performance criterion {\DRed $W_T+S_TI_T-\lambda I_T^2$} over 232 days for the different strategies considered (optimal and fixed-placement) based on LOB MSFT Data in 2019}.}\label{tab:TerminalRewardb}\vspace{.5 cm}
\end{table}

						\begin{table}[t]
							\centering
							\setlength{\tabcolsep}{8pt}
							\resizebox{\textwidth}{!}{%
								\begin{tabular}{@{}ccccccccc@{}}
									\toprule\midrule
									& \multicolumn{2}{c}{\begin{tabular}{c}Optimal Strategy\\ with Martingale Midprice\\ Conditioning on $g_1(\pmb{e}_\tk)$\end{tabular}} & \multicolumn{2}{c}{\begin{tabular}{c}Optimal Strategy\\ with Martingale Midprice\\ Conditioning on $g_2(\pmb{e}_\tk)$\end{tabular}}&
									\multicolumn{2}{c}{\begin{tabular}{c}Optimal Strategy\\ with Martingale Midprice\\ Conditioning on $g_3(\pmb{e}_\tk)$\end{tabular}} &
									\multicolumn{2}{c}{\begin{tabular}{c}Optimal Strategy\\ with General Midprice\\ Conditioning on $g_3(\pmb{e}_\tk)$\end{tabular}}\\ \midrule
									Mean & \multicolumn{2}{c}{$8.23\times10^4$} & \multicolumn{2}{c}{$8.36\times10^4$}& \multicolumn{2}{c}{{ $8.44\times10^4$}} &\multicolumn{2}{c}{{ $8.01\times10^4$}}\\
									\midrule
									Std. & \multicolumn{2}{c}{$1.62\times10^6$} & \multicolumn{2}{c}{$1.62\times10^6$} & \multicolumn{2}{c}{{ $1.61\times10^6$}}  &\multicolumn{2}{c}{{ $1.37\times10^6$}}\\  
\toprule\midrule												& Level 1             & Level 2             & Level 3            & Level 4              & Level 5              & Level 6             & & \\ \midrule
			Mean & $-1.88\times10^{6}$        &  $-3.06\times10^{5}$       &    $	-4.41\times10^4$     &   $-3.12\times10^4$      &     $-5.26\times10^4 $   &$-3.86\times10^4  $                  & &     \\ \midrule
			Std. & $8.18\times10^{6}$        &   $3.86\times10^{6}$      &   $1.94\times10^{6}$      &   $1.05\times10^{6}$      &    $7.08\times10^{5}$     &$4.81\times10^{5}$                     & &     \\ \bottomrule
								\end{tabular}%
							}
							\vspace{1.5mm}
							\caption{{\DRed Sample Mean and Std. Dev. of the Terminal PnL, $W_T+\bar S_TI_T$ (Terminal Cash Holdings plus the Actual Liquidation Proceeds based on the LOB state at expiration), over 232 trading days based on LOB MSFT data in 2019.}}\label{tab:TerminalRewarda}\vspace{.5 cm}
						\end{table}

\begin{figure}
	\centering
	\begin{subfigure}{1\textwidth}
		\centering
		\includegraphics[width=.52\textwidth]{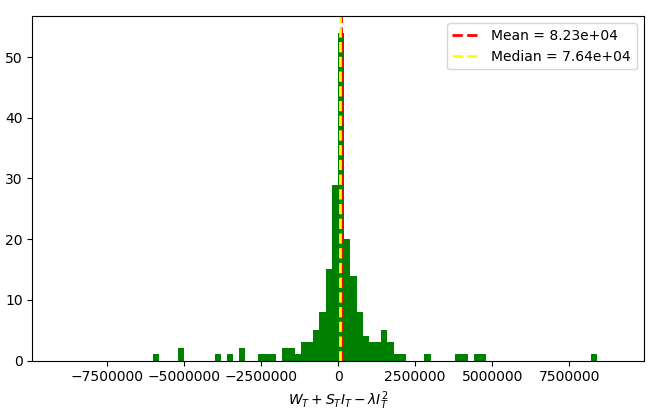}
		\caption{Optimal strategy by using the function $g=g_1$ under the assumption of a martingale price dynamics.}
	\end{subfigure}
	\begin{subfigure}{1\textwidth}
		\centering
		\includegraphics[width=.52\textwidth]{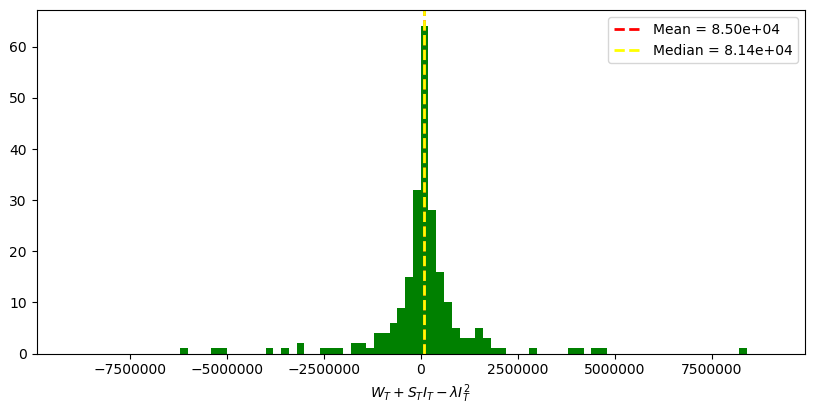}
		\caption{Optimal strategy by using the function $g=g_2$ under the assumption of a martingale price dynamics.}
	\end{subfigure}
	\begin{subfigure}{1\textwidth}
		\centering
		\includegraphics[width=.52\textwidth]{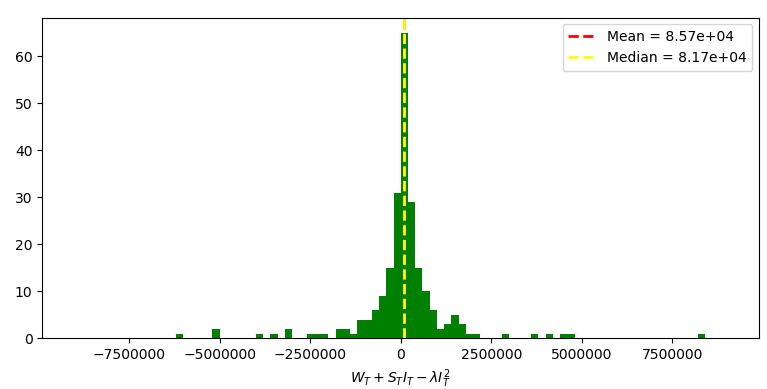}
		\caption{Optimal strategy by using the function $g=g_3$ under the assumption of a martingale price dynamics.}
	\end{subfigure}
	\begin{subfigure}{1\textwidth}
	\centering
	\includegraphics[width=.52\textwidth]{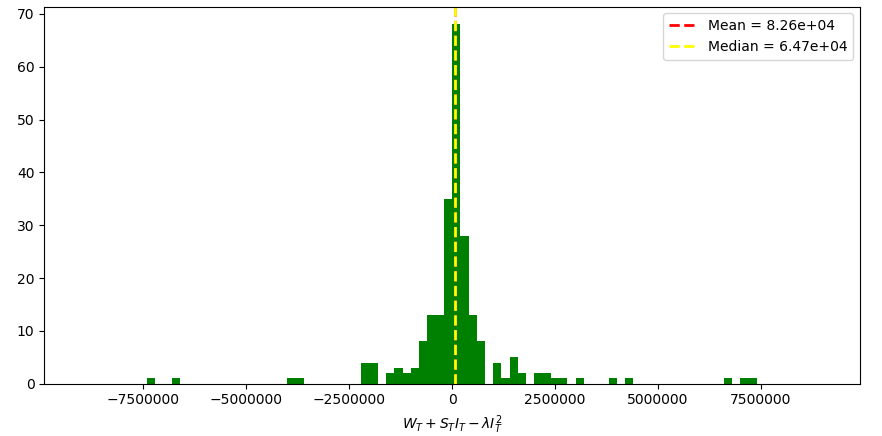}
	\caption{Optimal strategy by using the function $g=g_3$ under the assumption of a general price dynamics.}
\end{subfigure}
\medskip
	\caption{Histogram of the performance criterion obtained from the optimal strategies under different scenarios for 232 trading days of MSFT in 2019.}
	\label{fig4}
\end{figure}

\medskip

Based on the results observed in Tables \ref{tab:TerminalRewardb}\,-\,\ref{tab:TerminalRewarda} and Figure \ref{fig4}, we can conclude that our optimal strategy under any scenario significantly outperforms the fixed-level placement strategies. From {Tables \ref{tab:TerminalRewardb}\,-\,\ref{tab:TerminalRewarda},}  we can observe that the mean and standard deviation for all three policies under the martingale assumption are close to each other but the policy obtained by choosing the function $g$ as $g_3$ has the highest mean and lowest standard deviation. In contrast, when using {a one-step forecast \eqref{OSFN}} and choosing the function $g$ as $g_3$, the average of the performance of the optimal strategy is slightly lower than that under martingale assumption, but the standard deviation is also lower. Furthermore, as shown in Figure \ref{fig4}, the values of the performance criterion for the optimal strategy under all 4 scenarios concentrate around zero, but exhibit some ``outliers", which indicate heavy tails. 

{For comparison, in Tables~{\ref{tab:TerminalRewardc}}\,-\,\ref{tab:TerminalRewardLiquidd}, we report the analogous results} using the optimal placement strategies {from the preprint} \cite{zoe}, in which the probabilities $\pi^\pm_{\tkk}$ and $\pi_{\tkk}(1,1)$ are deterministic quadratic functions of time calibrated  using historical data. While the sample standard deviations of our adaptive placement strategy are slightly larger than those in \cite{zoe}, the sample means of the performance criteria are significantly better.

{\Blue We can further wonder how important adopting random demand is to achieve good PnL. In Table \ref{tab:TerminalRewardaba}-\ref{tab:TerminalRewardbb} below, we compute the {sample means and standard deviations of $W_T+S_TI_T-\lambda I_T^2$ and $W_T+\bar S_TI_T$ based on the 232 day, still assuming adaptive probabilities $\pi$, but now taking $c$ and $p$ at each test day \emph{constant} to their sample averages of the previous 20 trading days. As shown therein, though the average PnL are all positive, they are significantly smaller than those in Tables \ref{tab:TerminalRewardb}-\ref{tab:TerminalRewarda} and even those using nonadaptive $\pi$'s but stochastic demand as illustrated in Tables~{\ref{tab:TerminalRewardc}}-\ref{tab:TerminalRewardLiquidd}.} 

\bigskip
\begin{table}[ht]
\centering
\setlength{\tabcolsep}{8pt}
\resizebox{\textwidth}{!}{%
\begin{tabular}{@{}ccccccc@{}}
\toprule\midrule
 & \multicolumn{2}{c}{\begin{tabular}{c}Optimal Strategy\\ with Non-Martingale {Fundamental Price}\\ {and $\pi_\tk(1,1)\geq 0$}\end{tabular}} & \multicolumn{2}{c}{\begin{tabular}{c}Optimal Strategy\\ with Martingale {Fundamental Price}\\ {and $\pi_\tk(1,1)\geq 0$}\end{tabular}}&
 \multicolumn{2}{c}{\begin{tabular}{c}{Optimal Strategy}\\ {with Non-Martingale Fundamental Price} \\{and $\pi_\tk(1,1)\equiv 0$}\end{tabular}} \\ \midrule
Mean & \multicolumn{2}{c}{$6.13\times10^4$} & \multicolumn{2}{c}{$5.80\times10^4$}& \multicolumn{2}{c}{{ $6.11\times10^4$}}                                               \\ \midrule
Std. & \multicolumn{2}{c}{$1.22\times10^6$} & \multicolumn{2}{c}{$1.30\times10^6$} & \multicolumn{2}{c}{{ $1.22\times10^6$}}                                              \\ \toprule\midrule
\end{tabular}%
}
\vspace{.1mm}
\caption{{\DRed Sample} Mean and Std. of the Terminal Objective {Values} $W_T+S_TI_T-\lambda I_T^2$ over 232 Days. We fix $\lambda = 0.0005$. {We control cash holdings and inventory processes {\Blue assuming stochastic demand functions but deterministic intensity of MO arrivals as \cite{zoe}}.}}\label{tab:TerminalRewardc}
\vspace{.5 cm}
\end{table}

\begin{table}[h]
\centering
\setlength{\tabcolsep}{8pt}
\resizebox{\textwidth}{!}{%
\begin{tabular}{@{}ccccccc@{}}
\toprule\midrule
 & \multicolumn{2}{c}{\begin{tabular}{c}Optimal Strategy\\ with Non-Martingale {Price}\\ {and $\pi_\tk(1,1)\geq 0$}\end{tabular}} & \multicolumn{2}{c}{\begin{tabular}{c}Optimal Strategy\\ with Martingale {Price}\\ {and $\pi_\tk(1,1)\geq 0$}\end{tabular}}&
 \multicolumn{2}{c}{\begin{tabular}{c}{Optimal Strategy}\\ {with Non-Martingale Price} \\{and $\pi_\tk(1,1)\equiv 0$}\end{tabular}} \\ \midrule
Mean & \multicolumn{2}{c}{$6.00\times10^4$}& \multicolumn{2}{c}{$5.56\times10^4$}& \multicolumn{2}{c}{{$5.97\times10^4$}}                                               \\ \midrule
Std. & \multicolumn{2}{c}{$1.22\times10^6$} & \multicolumn{2}{c}{$1.30\times10^6$}& \multicolumn{2}{c}{{$1.22\times10^6$}}                                               \\ \bottomrule
\end{tabular}%
}
\vspace{.1mm}
\caption{{\DRed Sample} Mean and Std. of the Terminal {Values $W_{T}+\bar{S}_{T}I_{T}$ (Terminal Cash Holdings plus {\DRed the Actual Liquidation Proceeds based on the LOB state at expiration}) over 232 Days}. We control cash holdings and inventory processes {\Blue assuming stochastic demands functions but deterministic intensity of MO arrivals as \cite{zoe}}.}
\label{tab:TerminalRewardLiquidd}
\vspace{.5 cm}
\end{table}

\begin{table}[H]
	\centering
	\setlength{\tabcolsep}{8pt}
	\resizebox{\textwidth}{!}{%
		\begin{tabular}{@{}ccccccccc@{}}
			\toprule\midrule
			& \multicolumn{2}{c}{\begin{tabular}{c}Optimal Strategy\\ with Martingale Midprice\\ Conditioning on $g_1(\pmb{e}_\tk)$\end{tabular}} & \multicolumn{2}{c}{\begin{tabular}{c}Optimal Strategy\\ with Martingale Midprice\\ Conditioning on $g_2(\pmb{e}_\tk)$\end{tabular}}&
			\multicolumn{2}{c}{\begin{tabular}{c}Optimal Strategy\\ with Martingale Midprice\\ Conditioning on $g_3(\pmb{e}_\tk)$\end{tabular}} &
			\multicolumn{2}{c}{\begin{tabular}{c}Optimal Strategy\\ with General Midprice\\ Conditioning on $g_3(\pmb{e}_\tk)$\end{tabular}}\\ \midrule
			Mean & \multicolumn{2}{c}{$3.60\times10^4$} & \multicolumn{2}{c}{$3.64\times10^4$}& \multicolumn{2}{c}{{ $3.90\times10^4$}} &\multicolumn{2}{c}{{ $7.15\times10^4$}}\\
			\midrule
			Std. & \multicolumn{2}{c}{$2.07\times10^6$} & \multicolumn{2}{c}{$2.07\times10^6$} & \multicolumn{2}{c}{{ $2.06\times10^6$}}  &\multicolumn{2}{c}{{ $1.59\times10^6$}}\\ \toprule\bottomrule
		\end{tabular}%
	}
	\vspace{1.5mm}
	\caption{Sample Mean and Std. Dev. of the performance criterion $W_T+S_TI_T-\lambda I_T^2$ over 232 days for the different adaptive strategies, but with deterministic demand functions (i.e., $c$ and $p$ are constant to its average values in the 20 days previous to each testing day.}\label{tab:TerminalRewardaba}\vspace{.5 cm}
\end{table}

\begin{table}[H]
	\centering
	\setlength{\tabcolsep}{8pt}
	\resizebox{\textwidth}{!}{%
		\begin{tabular}{@{}ccccccccc@{}}
			\toprule\midrule
			& \multicolumn{2}{c}{\begin{tabular}{c}Optimal Strategy\\ with Martingale Midprice\\ Conditioning on $g_1(\pmb{e}_\tk)$\end{tabular}} & \multicolumn{2}{c}{\begin{tabular}{c}Optimal Strategy\\ with Martingale Midprice\\ Conditioning on $g_2(\pmb{e}_\tk)$\end{tabular}}&
			\multicolumn{2}{c}{\begin{tabular}{c}Optimal Strategy\\ with Martingale Midprice\\ Conditioning on $g_3(\pmb{e}_\tk)$\end{tabular}} &
			\multicolumn{2}{c}{\begin{tabular}{c}Optimal Strategy\\ with General Midprice\\ Conditioning on $g_3(\pmb{e}_\tk)$\end{tabular}}\\ \midrule
			Mean & \multicolumn{2}{c}{$3.26\times10^4$} & \multicolumn{2}{c}{$3.30\times10^4$}& \multicolumn{2}{c}{{ $3.56\times10^4$}} &\multicolumn{2}{c}{{ $6.67\times10^4$}}\\
			\midrule
			Std. & \multicolumn{2}{c}{$2.07\times10^6$} & \multicolumn{2}{c}{$2.07\times10^6$} & \multicolumn{2}{c}{{ $2.06\times10^6$}}  &\multicolumn{2}{c}{{ $1.60\times10^6$}}\\ \toprule\bottomrule
		\end{tabular}%
	}
	\vspace{1.5mm}
	\caption{Sample Mean and Std. Dev. of the Terminal PnL, $W_T+\bar S_TI_T$ (Terminal Cash Holdings plus the Actual Liquidation Proceeds based on the LOB state at expiration), over 232 days for the different adaptive strategies, but with deterministic demand functions (i.e., $c$ and $p$ are constant to its average values in the 20 days previous to each testing day.}\label{tab:TerminalRewardbb}\vspace{.5 cm}
\end{table}

\subsubsection{Inventory and price evolution.} To analyze whether the penalization term in the performance criterion is indeed able to push the MM to lower her inventory towards the end of the trading day, we display in Figures \ref{fig11} and \ref{fig22} two prototypical sample inventory paths throughout the trading day of August 7th when computing the optimal strategy with the choice $g_3$,  under both the martingale (Figure \ref{fig:PricePath1}) and the non-martingale (Figure \ref{fig:PricePath2}) price dynamics. To compare, we also display the inventory process under the fixed-placement strategies `Level 1'-`Level 6'.

Under the general price dynamics assumption, the inventory path can reach a higher level in the middle of the day, though it is able to {bring the inventory down} at the end of the day.
 As we can see from Figure \ref{fig:InvPath1} and Figure \ref{fig:InvPath2}, the optimal prices for both strategies swings between prices for level 1 and level 6. At the end of the trading day, the optimal ask prices under the general midprice assumption is closer to the level 1 price than under the martingale assumption, and the optimal bid prices are close to the level 6 price under both assumptions, which leads to a faster decrease in the inventory.
For {the two} optimal strategies, the end-of-day inventory is lower than the `Level 1'-`Level 6' policies. This shows the effectiveness of the liquidation penalty $-\lambda I_T^2$ in controlling inventory and avoiding large end of the day costs.

Finally, in Figures \ref{fig:PricePath1} and \ref{fig:PricePath2} we present the intraday optimal ask and bid placement prices $a_{t_{k}}$ and $b_{t_{k}}$ for August 7th when choosing the function $g$ as $g_3$ under, both, the martingale and the general midprice assumptions. We also compare them with the `Level 1'- `Level 6' benchmark policies. The price paths are for three 1-minute time intervals at the beginning of the trading day 10:00$-$10:01, in the middle of the trading day 12:45$-$ 12:46, and at the end of the trading day 15:29$-$15:30, respectively. {These graphs illustrate how the optimal placement strategies compare to the the fixed level strategies at different times of the day. In particular, the bid and ask prices behave as one should expect at the end of the day to control the inventory.}}

\section{Conclusions}
{\Blue In this manuscript, we focus on end-of-day inventory control in a market making problem. We assume the demand to be linear with random slope and intercept, which allows for greater flexibility and uncovers novel features of the resulting optimal control policy. We account for simultaneous arrivals of buy and sell MOs between consecutive market making actions, which also lead to novel patterns of the optimal policy. We allow the market maker to incorporate forecasts of the fundamental price in her placement strategy. Finally, we enable the investor to integrate the information on the arrival of market orders throughout the trading.
The performance of the proposed optimal policy is assessed using historical exchange transaction data throughout an entire year. The optimal strategy derived with the novel model specifications mentioned above yields greater flexibility and better results in our empirical study.} 

There are some key areas for future research based on our results:
\begin{itemize}
	\item It is natural to consider the possibility that the features of the demand functions also depend on the history of market orders $\pmb{e}_{t_{k}}$ rather than being assumed constants as in the current framework.
	
	\item It would be important to drop the assumption of independent between the price changes $S_{t_{k+1}}-S_{t_{k}}$ and the vector $(\onep,\onem,c^+_{t_{k+1}},p^+_{t_{k+1}},c^-_{t_{k+1}},p^-_{t_{k+}})$. 
	
	\item It is natural to consider the continuous limit of the model considered here. Such an extension could help us to consider general inventory penalties. 
	
%
\end{itemize}

\begin{figure}
	\centering
	\begin{subfigure}[b]{1\textwidth}
		\centering
		\includegraphics[width=.9\textwidth]{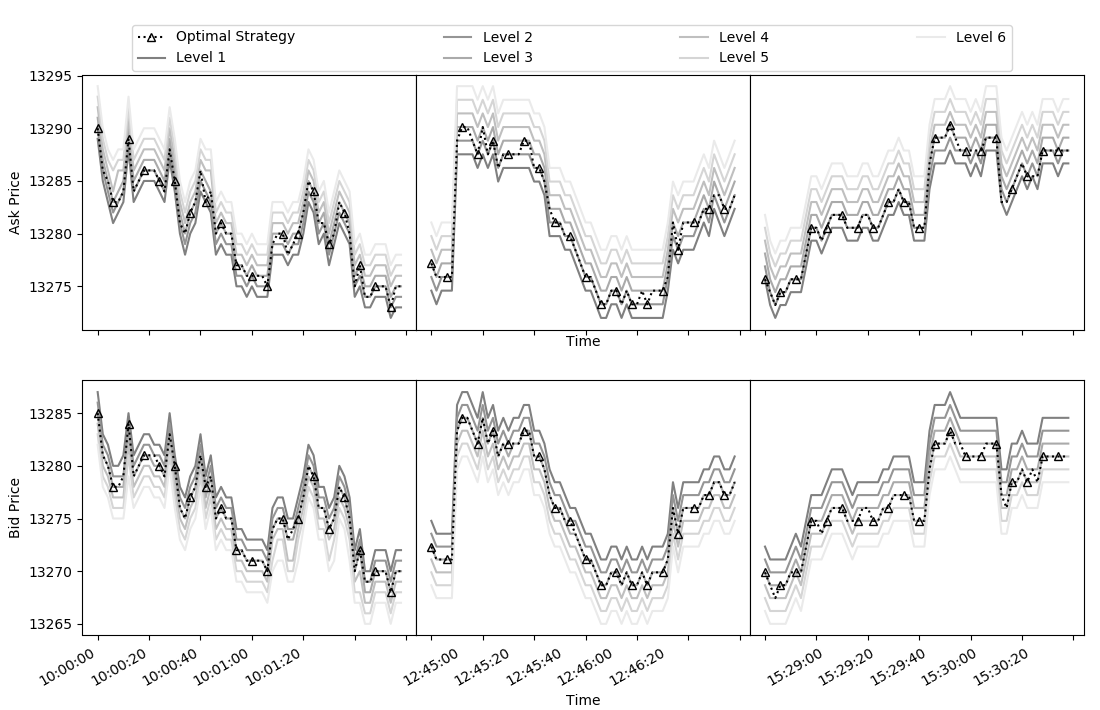}
		\caption{The Intraday Prices Paths.}
		\label{fig:PricePath1} 
	\end{subfigure}
	
	\begin{subfigure}[b]{1\textwidth}
		\centering
		\includegraphics[width=.9\textwidth]{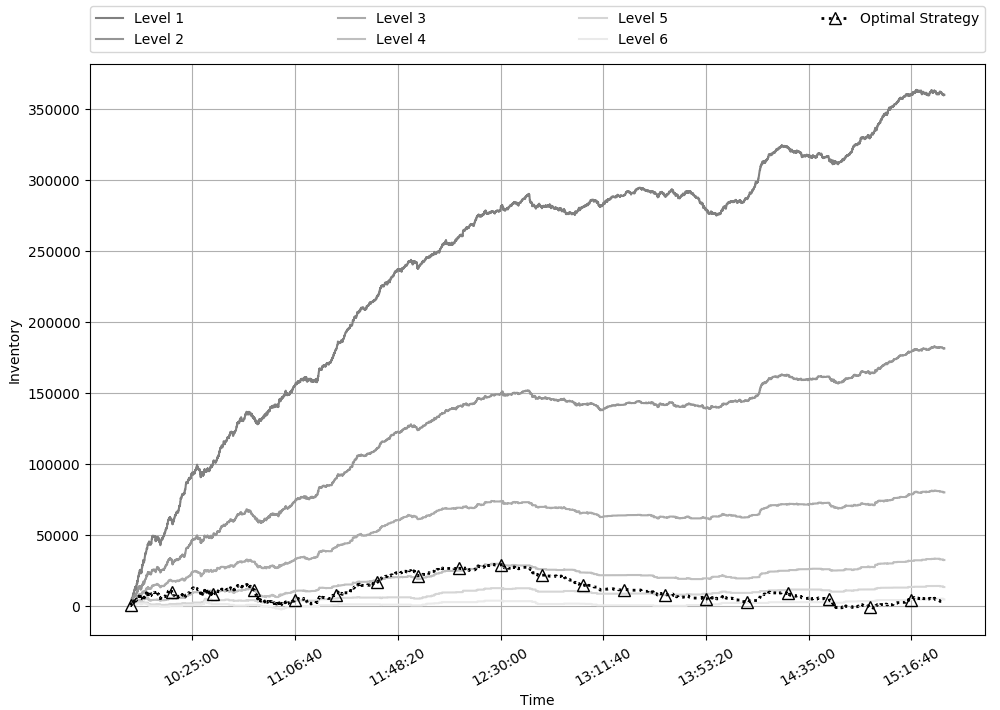}
		\caption{The Intraday Inventory Paths.}
		\label{fig:InvPath1}
	\end{subfigure}
	\vspace{0.3cm}
	\caption[Optimal Paths Martingale Price]{A comparison of the intraday price and inventory paths of the optimal strategy under martingale price dynamics assumption when choosing the function $g$ as $g_3$ and the ones of the benchmark policies on August 7th.}
	\label{fig11}
\end{figure}

\begin{figure}
	\centering
	\begin{subfigure}[b]{1\textwidth}
		\centering
		\includegraphics[width=.9\textwidth]{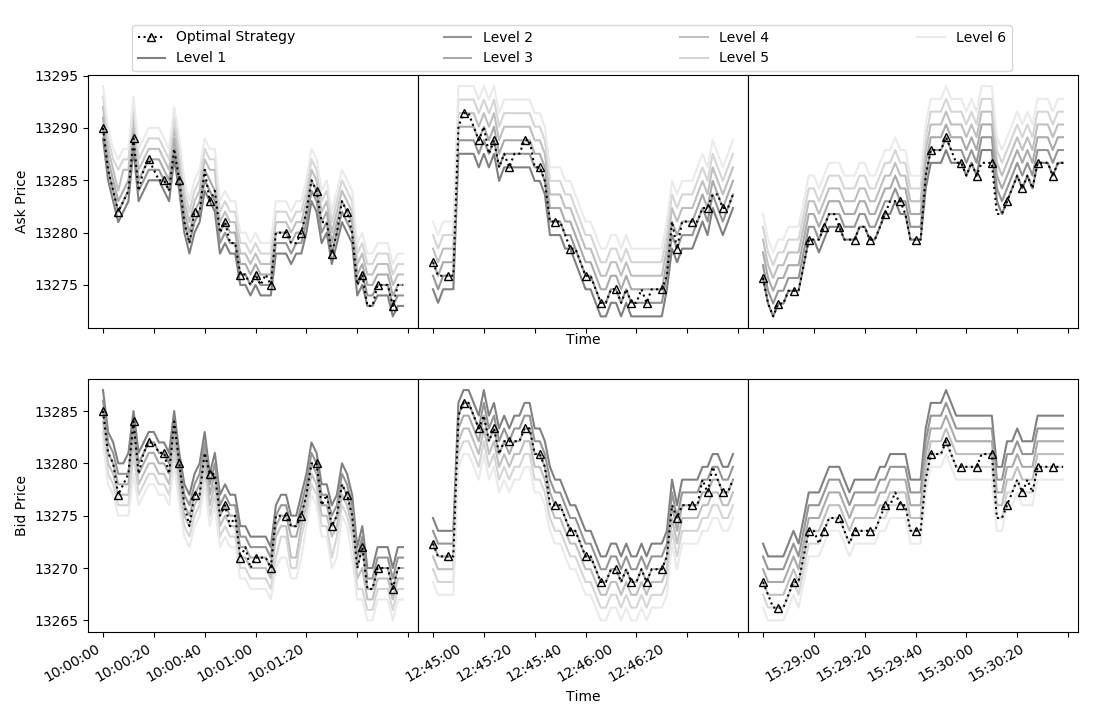}
		\caption{The Intraday Prices Paths.}
		\label{fig:PricePath2} 
	\end{subfigure}
	
	\begin{subfigure}[b]{1\textwidth}
		\centering
		\includegraphics[width=.9\textwidth]{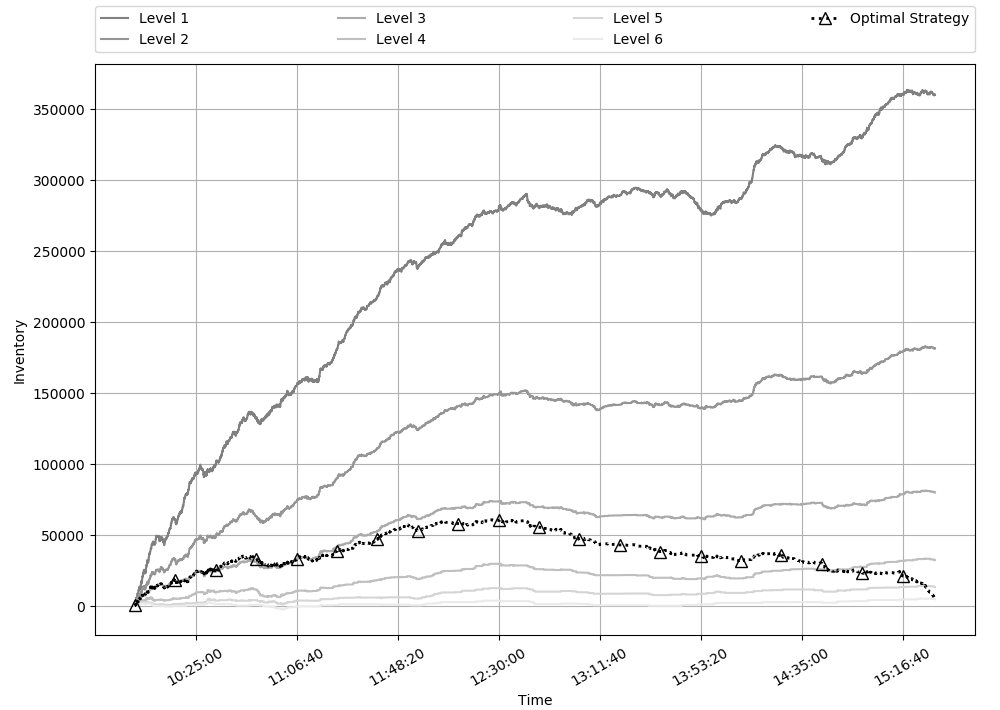}
		\caption{The Intraday Inventory Paths.}
		\label{fig:InvPath2}
	\end{subfigure}
	\vspace{0.3cm}
	\caption[Optimal Paths General Price]{A comparison of the intraday price and inventory paths of the optimal strategy under the general price dynamics assumption when choosing the function $g$ as $g_3$ and the ones of the benchmark policies on August 7th.}
	\label{fig22}
\end{figure}


\appendix

\renewcommand{\theequation}{A-\arabic{equation}}
\setcounter{equation}{0}  

\section{Proofs Of Main Results} \label{appen:section:2}

In this Appendix we provide all the proofs pertaining to Section \ref{ch2}.

\subsection{Proofs of Section \ref{sec:optimal:martingale}: Optimal Strategy for a Martingale Midprice}\label{appdx:martingale:price}

\begin{proof}[Proof of {Theorem} \ref{prop:measurability}]

{The proof is done by backwards induction. First, {note that for $k=N+1$, i.e., at the terminal time $T=t_{N+1}$, the statement {(ii)} is immediate {due to the terminal conditions} $\alpha_{t_{N+1}}=-\lambda$ and ${h}_{t_{N+1}}={g}_{t_{N+1}}=0$. So, it suffices to show the following two assertions:
\begin{enumerate}
	\item[(a)] If the statement (ii) is true for $k=j+1$, then the statements (i) and (iii) are true for $k=j$;
	\item[(b)] If the statement (ii) is true for $k=j+1$ and the statements (i) and (iii) are true for $k=j$, then the statement (ii) is true for $k=j$.
\end{enumerate} 
Let us start to prove the first assertion (a) above.}
To proceed, we consider Equation \eqref{eq113} for $k=j$. Replacing $W_\tjj$ and $I_\tjj$ on the right-hand-side {of \eqref{eq113}} by their corresponding recursive formulas \eqref{eqi11}-\eqref{eqw11}, we obtain
\begin{align}\label{BgnBlEq}
\begin{split}
V_{\tj}&=\sup_{L_{\tj}^\pm\in\mathcal{A}_{\tj}}\mathbb{E}\bigg\{{\Blue W_{t_j}}+\sum_{\delta=\pm}(\sj+\delta\ldj)\delta\onedj\cdj(\Pdj-\ldj)\\
&\qquad\qquad\qquad+\alpha_{\tjj}\big[I_\tj-\sum_{\delta=\pm}\delta\onedj\cdj(\Pdj-\ldj)\big]^2\\
&\qquad\qquad\qquad+\sjj\big[I_\tj-\sum_{\delta=\pm}\delta\onedj\cdj(\Pdj-\ldj)\big]\\
&\qquad\qquad\qquad+ {h}_{\tjj}\big[I_\tj-\sum_{\delta=\pm}\delta\onedj\cdj(\Pdj-\ldj)\big]+{g_{\tjj}}\bigg|\Fj\bigg\}.
\end{split}
\end{align}
Expanding the squares inside the expectation and rearranging terms, {we can write:}\small
\begin{align}
\begin{split}
\label{eq:expNMG-0}
V_{\tj}=\sup_{L_{\tj}^\pm\in\mathcal{A}_{\tj}}\Ex\Bigg\{&{\Blue W_{t_j}}+\sum_{\delta=\pm}\onedj\big[-\cdj(\ldj)^2+(\cdj\Pdj-\delta\cdj\sj)\ldj+\delta\cdj\Pdj\sj\big]\\
&+\alpha_{\tjj}\bigg\{I_\tj^2+\sum_{\delta=\pm}\onedj\Big[(\cdj)^2(\ldj)^2+\big(2\delta I_\tj\cdj-2(\cdj)^2\Pdj\big)\ldj\\
&\qquad\qquad\quad\qquad\qquad\qquad+(\cdj\Pdj)^2-2\delta I_\tj\cdj\Pdj\Big]\\
&\qquad\qquad+2\onepj\onemj\cpj\cmj(-\lpj\lmj+\Ppj\lmj+\Pmj\lpj-\Ppj\Pmj)\bigg\}\\
&+\sjj\bigg[I_\tj+\sum_{\delta=\pm}\onedj(-\delta\cdj\Pdj+\delta\cdj\ldj)\bigg]
\\
&+ {h}_{\tjj}I_\tj+\sum_{\delta=\pm}\onedj(-\delta {h}_{\tjj}\cdj\Pdj+\delta {h}_{\tjj}\cdj\ldj)+{g_{\tjj}}\Bigg|\Fj\Bigg\}.
\end{split}
\end{align}
\normalsize
We need to compute the conditional {expectation of each term}  above. Recall from Assumption \ref{assump:cp} that $\Gtj=\sigma(\Fj, \onepj,\onemj)$ and that, by our backwards induction hypothesis, {$\alpha_{t_{j+1}},h_{t_{j+1}}\in\mathcal{H}_{t_{j+1}}^\varpi=\sigma(\pmb{e}_{t_{j+1}})\subset\ \Gtj$}.  The idea is to apply the law of iterated expectations, $\Ex[\cdot|\Fj]=\Ex\big[\Ex[\cdot|\Gtj]\big|\Fj\big]$.  We can then pull out all the $\Gtj$-measurable factors (e.g., $\onepj$, $\onemj$, $\alpha_{t_{j+1}}$, {$h_{t_{j+1}}$},  $L_{t_j}^\pm$, $I_{t_j}$, $S_{t_j}$, etc.) from the inside expectation $\Ex[\cdot |\;\Gtj]$. We also use the fact that $(\cpj,p_{t_{j+1}}^+)$ and $(\cmj,p_{t_{j+1}}^-)$ are conditionally independent given $\Gtj$ (see Assumption \ref{assump:cp}), in addition to Eq.~\eqref{eqn:mucp}. 
As an example, we will explicitly show the computations of two terms in (\ref{eq:expNMG-0}). The remaining terms follow similar arguments. Consider $A:=\mathbb{E}\big[\alpha_\tjj\onepj\onemj\cpj\cmj\Ppj{\Pmj}\big|\Fj\big]${:}
\begin{alignat*}{3}
	A&=\mathbb{E}\bigg[\;\Ex\big[\alpha_\tjj\onepj\onemj\cpj\cmj\Ppj\Pmj\big|\Gtj\big]\bigg|\Fj\bigg]\\
	&=\mathbb{E}\bigg[\onepj\onemj\alpha_\tjj\mathbb{E}\big[\cpj\Ppj\cmj\Pmj\big|\Gtj\big]\Big|\Fj\bigg]&&\\
& =\mathbb{E}\bigg[\onepj\onemj\alpha_\tjj\mathbb{E}\big[\cpj\Ppj\big|\Gtj\big]\mathbb{E}\big[\cmj{\Pmj}\big|\Gtj\big]\Big|\Fj\bigg]&&\\
&
=\mathbb{E}\Big[\onepj\onemj\alpha_\tjj\mathbb{E}\big[\cpj\Ppj\big|\onepj\big]\mathbb{E}\big[\cmj{\Pmj}\big|\onemj\big]\Big|\Fj\Big]&&.
\end{alignat*}
Now, since $\alpha_{t_{j+1}}\in\mathcal{H}_{t_{j+1}}^\varpi$, we have 
\begin{equation}\label{ExprAlphaNd}
	\alpha_{t_{j+1}}=\Phi(\mathbbm{1}_{t_{j+1}}^+,\mathbbm{1}_{t_{j+1}}^-,\dots,\mathbbm{1}_{t_{j+2-\varpi}}^+,\mathbbm{1}_{t_{j+2-\varpi}}^-),
\end{equation} 
for some function $\Phi$. Using the well-know property 
\begin{equation}\label{eqn:prop:cond:exp}
\mathbb{E}(f(X,Y)|\mathcal{F})=\sum_x f(x,Y)\mathbb{P}(X=x|\mathcal{F})
\end{equation}
for a discrete r.v. $X$ and a $\mathcal{F}$-measurable variable $Y$, we can write:
\begin{alignat*}{3}
	A&=\mathbb{E}\Big[\onepj\onemj\alpha_\tjj\mathbb{E}\big[\cpj\Ppj\big|\onepj\big]\mathbb{E}\big[\cmj{\Pmj}\big|\onemj\big]\Big|\Fj\Big]&&\\
	&=\sum_{i,\ell\in\{0,1\}} i\ell \Phi(i,\ell,\mathbbm{1}_{t_{j}}^\pm,\dots,\mathbbm{1}_{t_{j+2-\varpi}}^\pm)\mathbb{E}\big[\cpj\Ppj\big|\onepj=i\big]\mathbb{E}\big[\cmj{\Pmj}\big|\onemj=\ell\big]\\
	&\qquad\qquad\times\mathbb{P}\big[\onepj=i,\onemj=\ell\big|\Fj\big]&\\
	&=\mu_{cp}^{+}\;\mu_{cp}^{-}\Phi(1,1,\mathbbm{1}_{t_{j}}^\pm,\dots,\mathbbm{1}_{t_{j+2-\varpi}}^\pm)\mathbb{P}\big[\onepj=1,\onemj=1\big|\Fj\big]&\\
& =\mu_{cp}^{+}\;\mu_{cp}^{-}\;\alpha_\tjj^{1,1}\pi_{t_{j+1}}(1,1),&&
\end{alignat*}
where we used the definition $\pi_{\tjj}(1,1) = \mathbbm{P}(\onepj=1,\onemj=1|\F)$ and that
\begin{align*}
\alpha_\tjj^{1,1}&=\mathbb{E}[\alpha_\tjj \,|\,\Fj,\mathbbm{1}^+_\tjj=1,\mathbbm{1}^-_\tjj=1]\\
&=\mathbb{E}[\Phi(\mathbbm{1}_{t_{j+1}}^+,\mathbbm{1}_{t_{j+1}}^-,\mathbbm{1}_{t_{j}}^\pm,\dots,\mathbbm{1}_{t_{j+2-\varpi}}^\pm)\,\big|\,\Fj,\mathbbm{1}^+_\tjj=1,\mathbbm{1}^-_\tjj=1]\\
&=\mathbb{E}[\Phi(1,1,\mathbbm{1}_{t_{j}}^\pm,\dots,\mathbbm{1}_{t_{j+2-\varpi}}^\pm)\,\big|\,\Fj,\mathbbm{1}^+_\tjj=1,\mathbbm{1}^-_\tjj=1]
=\Phi(1,1,\mathbbm{1}_{t_{j}}^\pm,\dots,\mathbbm{1}_{t_{j+2-\varpi}}^\pm).
\end{align*}

Similarly, consider $B:=\mathbb{E}\big[S_{t_{j+1}}\onedj\cdj\Pdj\big|\Fj\big]$, {for $\delta\in\{+,-\}$}. Then, by Assumption \ref{assump:price} and the martingale condition $\E[S_{t_{j+1}}-S_{t_j}|\Fj]=0$:
\begin{align*}
B& =\mathbb{E}\bigg[(S_\tjj-S_\tj)\onedj\cdj\Pdj\big|\Fj\bigg]+S_\tj \mathbb{E}\bigg[\onedj\cdj\Pdj\big|\Fj\bigg]\\
&=\mathbb{E}\bigg[S_\tjj-S_\tj\big|\Fj\bigg]\mathbb{E}\bigg[\onedj\cdj\Pdj\big|\Fj\bigg]+S_\tj \mathbb{E}\bigg[\onedj\cdj\Pdj\big|\Fj\bigg]\\
&=S_\tj \mathbb{E}\bigg[\onedj\cdj\Pdj\big|\Fj\bigg].
\end{align*}
We then proceed as before:
\begin{align*}
B&=S_\tj \mathbb{E}\bigg[\mathbb{E}\big[\onedj\cdj\Pdj\big|\Gtj\big]\big|\Fj\bigg]=S_\tj \mathbb{E}\bigg[\onedj\mathbb{E}\big[\cdj\Pdj\big|\onedj\big]\big|\Fj\bigg]\\
&=S_{t_j}\sum_{\ell=0}^1 \ell \mathbb{E}\big[\cdj\Pdj\big|\onedj=\ell\big]\mathbb{P}\big[\onedj=\ell \big|\Fj\big]=S_{t_{j}}\mu_{cp}^{\delta} \pi_\tjj^\delta.
\end{align*}
After computing the conditional expectations therein  and 
substituting $V_{t_k}$ in the LHS of \eqref{eq:expNMG-0}, we get the equation:
\begin{align}\label{eqn:Quadratic:L}
\begin{split}
&\alpha_{\tj}I_{t_j}^2+S_{t_j} I_{t_j}+{h}_{\tj}I_{t_j}+{g_{\tj}}\\
&=\sup_{L_{t_j}^\pm\in\mathcal{A}_{t_j}}\ \Big[ \sum_{\delta=\pm}{\pi^\delta_\tjj}\bigg\{(\alpha^{1\delta}_{\tjj}\mut^\delta-\muo^\delta)(\ldj)^2+{\alpha^{1\delta}_{\tjj}(\mutt^\delta-2\delta\muoo^\delta I_\tj)-\delta {h}^{1\delta}_{\tjj}\muoo^\delta}\\
&\qquad\qquad\qquad\qquad\qquad+\big[\muoo^\delta+\delta {h}^{1\delta}_{\tjj}\muo^\delta+\alpha^{1\delta}_{\tjj}(2\delta\muo^\delta I_\tj-2\muto^\delta)\big]\ldj\bigg\}\\ 
&\qquad\qquad\qquad+2\alpha^{1,1}_{\tjj}{\pi_\tjj(1,1)}\Big(\muoo^+\muo^-\lmj+\muo^+\muoo^-\lpj-\muo^+\muo^-\lpj\lmj-\muoo^+\muoo^-\Big)\\
&\qquad\qquad\qquad+\alpha^0_{\tjj}I_\tj^2+I_\tj\sj+ {h}^0_{\tjj} I_\tj+{g^{0}_{\tjj}}\Big].
\end{split}
\end{align}
The expression inside the outer brackets on the right-hand side of the above equation is a quadratic function in $(L_{t_j}^+,L_{t_j}^-)$. Simple partial differentiation  and some simplifications (see {\eqref{PDWRTL}} below) show that the unique \emph{stationary points} of the quadratic functional are given by}
	\begin{align*}
	\begin{split}
	L_{\tj}^{+,*}={}^{\scaleto{(1)}{5pt}}\!A^+_{\tj}I_\tj+{}^{\scaleto{(2)}{5pt}}\!A^+_{\tj}+{}^{\scaleto{(3)}{5pt}}\!A^+_{\tj},\qquad 
	L_{\tj}^{-,*}=-{}^{\scaleto{(1)}{5pt}}\!A^-_{\tj}I_\tj-{}^{\scaleto{(2)}{5pt}}\!A^-_{\tj}+{}^{\scaleto{(3)}{5pt}}\!A^-_{\tj},
	\end{split}
	\end{align*}
{where} ${}^{\scaleto{(1)}{5pt}}\!A^\pm_{\tj}$, ${}^{\scaleto{(2)}{5pt}}\!A^\pm_{\tj}${,} and ${}^{\scaleto{(3)}{5pt}}\!A^\pm_{\tj}$ {satisfy} the relations given in \eqref{eq:A1} {with} the auxiliary quantities {$\rho^{\pm}_{t_j}$, $\psi_{t_j}^\pm${,} and $\gamma_{t_j}$} satisfying the relations given in \eqref{eqn:gamma:eta}.
{It remains to show that $L_{\tj}^{+,*}$ and $L_{\tj}^{-,*}$ are in fact the global maxima points of the quadratic function inside the outer brackets on the right-hand side of \eqref{eqn:Quadratic:L}. This is deferred to Corollary \ref{cor:candidates}. 
Substituting the values of $\big(L_{\tj}^{+,*}, L_{\tj}^{-,*}\big)$ above into Equation {\Blue \eqref{eqn:Quadratic:L}} and equating the coefficients of $I_{t_j}^2$, $I_{t_j}$, {and the remaining terms that do not depend on $I_{\tj}$ or $S_{\tj}$,} on both sides of the Equation \ref{eqn:Quadratic:L}, 
we obtain Equations \eqref{eqn:alphak}-\eqref{eq:g}, which proves {parts (i) and (iii)} of the theorem for $k=j$.

We now prove {the assertion (b) stated at the beginning of the proof; i.e., we show that {if statement (ii) is true for $k=j+1$ and statements (i) and (iii) are} true for $k=j$, then {statement (ii)} is true for $k=j$ (that is, 
$\alpha_{t_j},h_{t_j},{g_{t_j}}\in\Hjj$)}. {We show the details for $\alpha_{t_j}\in\Hjj$ (we can similarly prove $h_{t_j},g_{t_j}\in\Hjj$).} 
First, notice that by Assumption \ref{assump:cp}-(iv),  $\pi_{t_{j+1}}^\pm,\pi_{t_{j+1}}(1,1)\in\Hjj$. Furthermore, by Assumption \ref{assump:cp}-(iv) again and the representation \eqref{ExprAlphaNd}, which follows from our backward induction assumption $\alpha_{t_{j+1}}\in\mathcal{H}_{t_{j+1}}^\varpi$, we can compute the variables $\alpha_\tjj^0$, $\alpha_\tjj^{1+}$, and $\alpha_\tjj^{1,1}$, defined as in (\ref{Dfnalh0}), as follows:\footnote{{Note that Assumption \ref{assump:cp}-(iv) actually implies that $\mathbbm{P}(\onep=i,\onem=j|\F)=\mathbbm{P}(\onep=i,\onem=j|\Hkk)$ for all $i,j$ and not only for $i=j=1$.}}
\begin{equation}\label{CndCmpAlph}
\begin{aligned}
\alpha_\tjj^0
&=\sum_{i,\ell\in\{0,1\}}\Phi(i,\ell,\mathbbm{1}_{t_{j}}^\pm,\dots,\mathbbm{1}_{t_{j+2-\varpi}}^\pm)\mathbb{P}\big[\onepj=i,\onemj=\ell\big|\Fj\big]\\
&=\sum_{i,\ell\in\{0,1\}}\Phi(i,\ell,\mathbbm{1}_{t_{j}}^\pm,\dots,\mathbbm{1}_{t_{j+2-\varpi}}^\pm)\mathbb{P}\big[\onepj=i,\onemj=\ell\big|\Hjj\big]\\
\alpha_\tjj^{1+} 
&=\sum_{\ell\in\{0,1\}}\Phi(1,\ell,\mathbbm{1}_{t_{j}}^\pm,\dots,\mathbbm{1}_{t_{j+2-\varpi}}^\pm)\mathbb{P}\big[\onemj=\ell\big|\Fj,\onepj=1\big]\\
&=\sum_{\ell\in\{0,1\}}\Phi(1,\ell,\mathbbm{1}_{t_{j}}^\pm,\dots,\mathbbm{1}_{t_{j+2-\varpi}}^\pm)\frac{\mathbb{P}\big[\onepj=1,\onemj=\ell\big|\Hjj\big]}{\mathbb{P}\big[\onepj=1\big|\Hjj\big]}\\
\alpha_\tjj^{1,1} 
&=\Phi(1,1,\mathbbm{1}_{t_{j}}^\pm,\dots,\mathbbm{1}_{t_{j+2-\varpi}}^\pm).
\end{aligned}
\end{equation}
Therefore, it is now clear that 
$\alpha_\tjj^0,\alpha_\tjj^{1\pm},\alpha_\tjj^{1,1}\in\Hjj$. 
Using an identical argument, we can conclude that $h_\tjj^0,h_\tjj^{1\pm},h_\tjj^{1,1}\in\Hjj$.
Then, since $\mu_{c^mp^n}^\pm$ are constants, we can easily see that ${}^{\scaleto{(1)}{5pt}}\!A^\pm_{\tj},{}^{\scaleto{(2)}{5pt}}\!A^\pm_{\tj}$ and ${}^{\scaleto{(3)}{5pt}}\!A^\pm_{\tj}$ are $\Hjj$-measurable. {From Eq.~\eqref{eqn:alphak}}, we conclude that $\alpha_{t_j}$ is $\Hjj$-measurable random variables. Finally, we conclude the validity of the statement {(ii)} of the theorem for $k=j$.}
\end{proof}


\begin{proof}[Proof of Lemma \ref{lemma:alpha}] 

Substituting the value of ${}^{\scaleto{(1)}{5pt}}\!{A}^\pm_{\tk}$, defined in Eq.~(\ref{eq:A1}), into the recursive Eq.~(\ref{eqn:alphak}), we get that 
\begin{equation}\label{eqn:alpha:decompos}
\alpha_\tk=\alpha_{t_{k+1}}^0+N_k/\gamma_{t_k},
\end{equation}
with 
\begin{equation}\label{eqn:Nk:Dk}
\begin{aligned}
N_k&= (\alpha_\tkk^{1+}\muo^+\pi_\tkk^+)^2 \pi_\tkk^- (\alpha_\tkk^{1-}\mut^- - \muo^-) +(\alpha_\tkk^{1-}\muo^-\pi_\tkk^-)^2 \pi_\tkk^+ (\alpha_\tkk^{1+}\mut^+ - \muo^+)\\
&\quad -2\alpha_\tkk^{1+}\alpha_\tkk^{1-}{\pi^+_{\tkk}}{\pi^-_{\tkk}}\pi_\tkk(1,1)\alpha^{1,1}_{\tkk}(\muo^+\muo^-)^2,\\
\gamma_{t_k}&=\big[{\pi_\tkk(1,1)}\alpha^{1,1}_{\tkk}\muo^+\muo^-\big]^2-{\pi^+_{\tkk}}{ \pi^-_{\tkk}}(\alpha^{1+}_{\tkk}\mut^+-\muo^+)(\alpha^{1-}_{\tkk}\mut^--\muo^-).
\end{aligned}
\end{equation}
Note that for $k=N+1$, we have that $\alpha_{t_{N+1}}=\alpha_T=-\lambda<0$ and thus $\alpha_{t_{N+1}}<0$, 
as $\lambda$ is a positive constant.
{By} backwards induction, to prove the lemma it suffices to show that $N_k/\gamma_{t_k} \in (0,-\alpha_\tkk^0)$, 
whenever $\alpha_\tkk <0$. The {\Blue remaining proof} is then divided into three smaller subparts: proving that, for $k\in\{0,1,\ldots,N\}$, (i) $\gamma_{t_k}<0$, (ii) $N_k<0$, and that (iii) $N_k/\gamma_{t_k}<-\alpha_{t_{k+1}}^0$.
\begin{itemize}
	\item[(i)] {\Blue Clearly,}
	\begin{equation} \label{eqn:alphaplus}
\begin{aligned}
\alpha_\tkk^{1+} &=\Ex[\alpha_\tkk |\F,\onep=1]\\
			&=\Ex[\alpha_\tkk |\F,\onep=1,\onem=1]\Px[\onem=1|\F,\onep=1] \\
			&\quad+ \Ex[\alpha_\tkk |\F,\onep=1,\onem=0]\Px[\onem=0|\F,\onep=1]\\
			&=\alpha_\tkk^{1,1}\frac{\pi_\tkk(1,1)}{\pi_\tkk^+}+ {\alpha_\tkk^{1,0}}\left(1-\frac{\pi_\tkk(1,1)}{\pi_\tkk^+}\right){\Blue ,}
\end{aligned}
\end{equation}
and, therefore, 
\begin{equation}\label{eqn:alpha1p}
\alpha_\tkk^{1+}\pi_\tkk^+ =\alpha_\tkk^{1,1}\pi_\tkk(1,1) + {\alpha_\tkk^{1,0}}\Big(\pi_\tkk^+-\pi_\tkk(1,1)\Big).
\end{equation}
Analogously, we have 
\begin{align}\label{eqn:alpha1m}
\alpha_\tkk^{1-}\pi_\tkk^- &=\alpha_\tkk^{1,1}\pi_\tkk(1,1) + {\alpha_\tkk^{0,1}}\Big(\pi_\tkk^--\pi_\tkk(1,1)\Big).
\end{align}
On the other hand, since by assumption $\alpha_{t_{k+1}}<0$,
\begin{align*}
&\alpha_\tkk^{1,1}=\mathbb{E}[\alpha_\tkk |\F,\onep=1,\onem=1]\leq 0,\\
&\Ex[\alpha_\tkk |\F,\onep=1,\onem=0]\leq{}0,\quad 
\mathbb{E}[\alpha_\tkk |\F,\onep=0,\onem=1]\leq 0, 
\\
&
\pi_\tkk(1,1)=\mathbbm{P}(\onep=1,\onem=1|\F)\leq
\mathbbm{P}(\mathbbm{1}_\tkk^\pm=1|\F)=\pi^\pm_{\tkk},
\end{align*}
and, thus, by Eqs.~\eqref{eqn:alpha1p}-\eqref{eqn:alpha1m}, 
\begin{equation}\label{eqn:comparison:alpha}
\alpha_\tkk^{1+} \pi_\tkk^+  \leq \alpha_\tkk^{1,1} \pi_\tkk(1,1) \leq 0,\qquad
\alpha_\tkk^{1-} \pi_\tkk^-  \leq \alpha_\tkk^{1,1} \pi_\tkk(1,1) \leq 0. 
\end{equation}
Since $\mut^\pm\geq(\muo^\pm)^2$, these equations imply that,
\begin{equation}\label{eq:Dneg}
\begin{aligned}
\gamma_{t_k}  &\leq {\pi^+_{\tkk}}{\pi^-_{\tkk}}\alpha_{\tkk}^{1+}\alpha_{\tkk}^{1-}(\muo^+\muo^-)^2 -  {\pi^+_{\tkk}}{ \pi^-_{\tkk}}(\alpha^{1+}_{\tkk}\mut^+-\muo^+)(\alpha^{1-}_{\tkk}\mut^--\muo^-) \\
&\leq {\pi^+_{\tkk}}{\pi^-_{\tkk}}\alpha_{\tkk}^{1+}\alpha_{\tkk}^{1-}\mut^+\mut^- -  {\pi^+_{\tkk}}{ \pi^-_{\tkk}}(\alpha^{1+}_{\tkk}\mut^+-\muo^+)(\alpha^{1-}_{\tkk}\mut^--\muo^-) \\
&= {\pi^+_{\tkk}}{\pi^-_{\tkk}}\big[\alpha_{\tkk}^{1+}\mut^+\muo^-+\alpha_\tkk^{1-}{\Blue \mut^-\muo^+} -\muo^+\muo^-\big]< 0,
\end{aligned}
\end{equation}	
where for the second- and last inequality above we used that $\alpha_{\tkk}^{1\pm}=\mathbb{E}[\alpha_\tkk \,|\,\F,\mathbbm{1}^\pm_\tkk=1]\leq{}0$ since, by assumption, $\alpha_\tkk<0$.

\item[(ii)] Let us first rearrange the terms in (\ref{eqn:Nk:Dk}) as follows:
\begin{equation}\label{eqappen3a}
\begin{aligned}
N_k&= (\alpha_\tkk^{1+}\muo^+\pi_\tkk^+)^2 \pi_\tkk^- (\alpha_\tkk^{1-}\mut^- - \muo^-) +(\alpha_\tkk^{1-}\muo^-\pi_\tkk^-)^2 \pi_\tkk^+ (\alpha_\tkk^{1+}\mut^+ - \muo^+)\\
&\quad\qquad -2\alpha_\tkk^{1+}\alpha_\tkk^{1-}{\pi^+_{\tkk}}{\pi^-_{\tkk}}\pi_\tkk(1,1)\alpha^{1,1}_{\tkk}(\muo^+\muo^-)^2 \\
&= \alpha_{\tkk}^{1+}\alpha_{\tkk}^{1-} \pi_\tkk^+ \pi_\tkk^- (\muo^+)^2 \big[\alpha_{\tkk}^{1+} \pi_\tkk^+ \mut^- - \alpha_\tkk^{1,1} \pi_\tkk(1,1)(\muo^-)^2       \big] \\
&\quad\qquad +
\alpha_{\tkk}^{1+}\alpha_{\tkk}^{1-} \pi_\tkk^+ \pi_\tkk^- (\muo^-)^2 \big[\alpha_{\tkk}^{1-} \pi_\tkk^- \mut^+ - \alpha_\tkk^{1,1} \pi_\tkk(1,1)(\muo^+)^2       \big] \\
&\quad\qquad -
(\alpha_\tkk^{1+}\muo^+\pi_\tkk^+)^2{\pi_\tkk^- }\muo^- - (\alpha_\tkk^{1-}\muo^-\pi_\tkk^-)^2{\pi_\tkk^+} \muo^+   .\end{aligned}
\end{equation}
Then, using once more the inequalities \eqref{eqn:comparison:alpha},
\begin{equation}\label{eqappen3}
\begin{aligned}
N_k&\leq 
\alpha_{\tkk}^{1+}\alpha_{\tkk}^{1-} \pi_\tkk^+ \pi_\tkk^- (\muo^+)^2\pi_\tkk(1,1)\alpha^{1,1}_{\tkk}[\mut^- - (\muo^-)^2] \\
&\quad\qquad +
\alpha_{\tkk}^{1+}\alpha_{\tkk}^{1-} \pi_\tkk^+ \pi_\tkk^- (\muo^+)^2\pi_\tkk(1,1)\alpha^{1,1}_{\tkk}[\mut^+ - (\muo^+)^2] \\
&\quad\qquad -
(\alpha_\tkk^{1+}\muo^+\pi_\tkk^+)^2{\pi_\tkk^-}\muo^- - (\alpha_\tkk^{1-}\muo^-\pi_\tkk^-)^2{\pi_\tkk^+}\muo^+ \\
&\leq 
-(\alpha_\tkk^{1+}\muo^+\pi_\tkk^+)^2{\pi_\tkk^-}\muo^- - (\alpha_\tkk^{1-}\muo^-\pi_\tkk^-)^2{\pi_\tkk^+}\muo^+ <0,
\end{aligned}
\end{equation}
where in the second inequality we used once more that $\max\{\alpha_\tkk^{1,1},\alpha_\tkk^{+1},\alpha_\tkk^{-1}\}\leq{}0$.

	\item[(iii)] Note that, since we already proved that $\gamma_{t_k}<0$, proving that $N_k/\gamma_{t_k} <-\alpha_\tkk^0$ a.s. whenever $\alpha_\tkk <0$ is equivalent to showing that $\alpha_\tkk^0 \gamma_{t_k} +N_k >0$ a.s. whenever $\alpha_\tkk <0$. Define
	\begin{align*}
	\varphi(\bx, \by):=&\; \alpha_{t_{k+1}}^0 \Big(\big[{\pi_\tkk(1,1)}\alpha^{1,1}_{\tkk}\muo^+\muo^-\big]^2-{\pi^+_{\tkk}}{ \pi^-_{\tkk}}(\alpha^{1+}_{\tkk}\bx-\muo^+)(\alpha^{1-}_{\tkk}\by-\muo^-)\Big)\\
	&\qquad + (\alpha_\tkk^{1+}\muo^+\pi_\tkk^+)^2 \pi_\tkk^- (\alpha_\tkk^{1-}\by - \muo^-) +(\alpha_\tkk^{1-}\muo^-\pi_\tkk^-)^2 \pi_\tkk^+ (\alpha_\tkk^{1+}\bx - \muo^+)\\
&\qquad -2\alpha_\tkk^{1+}\alpha_\tkk^{1-}{\pi^+_{\tkk}}{\pi^-_{\tkk}}\pi_\tkk(1,1)\alpha^{1,1}_{\tkk}(\muo^+\muo^-)^2
	\end{align*}
	Notice that $\varphi(\mut^+,\mut^-)=\alpha_\tkk^0 \gamma_{t_k} +N_k$. Taking partial derivatives,
\begin{equation}\label{eqn:partial:varphi:x}
\begin{split}
\frac{\partial \varphi}{\partial \bx} &= -\alpha_\tkk^0 \pi_\tkk^+ \pi_\tkk^- \Big[\alpha_{\tkk}^{1+}\alpha_{\tkk}^{1-}\by - \alpha_{\tkk}^{1+}\muo^-\Big] + \Big(\alpha_\tkk^{1-}\muo^-\pi_\tkk^-\Big)^2 \pi_\tkk^+ \alpha_\tkk^{1+}  \\
&=  \alpha_{\tkk}^{1+}\alpha_{\tkk}^{1-} \pi_\tkk^+ \pi_\tkk^- \Big[\alpha_\tkk^{1-}\pi_\tkk^-(\muo^-)^2 - \alpha_\tkk^0 \by\Big] +\alpha_\tkk^0 \pi_\tkk^+ \pi_\tkk^- \alpha_{\tkk}^{1+}\muo^-.
\end{split}
\end{equation}
\begin{equation}\label{eqn:partial:varphi:y}
\begin{split}
\frac{\partial \varphi}{\partial \by} &= -\alpha_\tkk^0 \pi_\tkk^+ \pi_\tkk^- \Big[\alpha_{\tkk}^{1+}\alpha_{\tkk}^{1-}\bx - \alpha_{\tkk}^{1-}\muo^+\Big] + \Big(\alpha_\tkk^{1+}\muo^+\pi_\tkk^+\Big)^2 \pi_\tkk^- \alpha_\tkk^{1-}  \\
&=  \alpha_{\tkk}^{1+}\alpha_{\tkk}^{1-} \pi_\tkk^+ \pi_\tkk^- \Big[\alpha_\tkk^{1+}\pi_\tkk^+(\muo^+)^2 - \alpha_\tkk^0 \bx\Big] +\alpha_\tkk^0 \pi_\tkk^+ \pi_\tkk^- \alpha_{\tkk}^{1-}\muo^+.
\end{split}
\end{equation}
Recall that by assumption $\alpha_{t_{k+1}}<0$ a.s. and thus,
\begin{align} \nonumber
\alpha_\tkk^0 &= \mathbb{E}\big[\alpha_\tkk\big|\F\big] = \mathbb{E}\bigg[\mathbb{E}\big[\alpha_\tkk|\F,\onen^\pm\big] \bigg|\F\bigg] \\ \label{ineq:alpha0}
&= \alpha_\tkk^{1\pm}\pi_\tkk^\pm + \mathbb{E}\big[\alpha_\tkk|\F,\onen^\pm=0\big](1-\pi_\tkk^\pm)\leq \alpha_\tkk^{1\pm}\pi_\tkk^\pm.
\end{align}
Using the last inequality in \eqref{eqn:partial:varphi:x}, we obtain that
\[
\frac{\partial\varphi}{\partial \bx}  \geq \alpha_{\tkk}^{1+}\alpha_{\tkk}^{1-} \pi_\tkk^+ \pi_\tkk^-  \alpha_\tkk^0 [(\muo^-)^2 -\by] + \alpha_\tkk^0 \pi_\tkk^+ \pi_\tkk^- \alpha_{\tkk}^{1+}\muo^-. \nonumber \\
\]
Therefore, using that $(\mu_c^-)^2\leq\mu^{-}_{c^2}$ and $\max\{\alpha_\tkk^0,\alpha_\tkk^{1+},\alpha_\tkk^{1-}\}\leq{}0$, we have
\begin{equation}\label{eqn:partial:varphi:bx}
\begin{split}
\frac{\partial\varphi}{\partial \bx}\big(\;\sbt\;,\mu_{c^2}^-\big) &\geq \alpha_{\tkk}^{1+}\alpha_{\tkk}^{1-} \pi_\tkk^+ \pi_\tkk^-  \alpha_\tkk^0 [(\muo^-)^2 -\mu_{c^2}^-] + \alpha_\tkk^0 \pi_\tkk^+ \pi_\tkk^- \alpha_{\tkk}^{1+}\muo^-\\
& \geq \alpha_\tkk^0 \pi_\tkk^+ \pi_\tkk^- \alpha_{\tkk}^{1+}\muo^-  \geq 0,
\end{split}
\end{equation}
which implies that $\varphi(\;\sbt\;,\mu_{c^2}^-)$ is a non-decreasing  function of the first parameter. Thus, 
\begin{equation}\label{eqn:bound:varphi:1}
\varphi\big((\mu_c^+)^2,\mu_{c^2}^-\big)\leq \varphi(\mut^+,\mut^-)=\alpha_\tkk^0 \gamma_{t_k} +N_k.
\end{equation}
Using \eqref{eqn:partial:varphi:y}, \eqref{ineq:alpha0}, and $\max\{\alpha_\tkk^0,\alpha_\tkk^{1+},\alpha_\tkk^{1-}\}\leq{}0$,
\begin{align*}
\frac{\partial \varphi}{\partial \by}\big((\mu_c^+)^2,\;\sbt\;\big) &=  \alpha_{\tkk}^{1+}\alpha_{\tkk}^{1-} \pi_\tkk^+ \pi_\tkk^- [\alpha_\tkk^{1+}\pi_\tkk^+(\muo^+)^2 - \alpha_\tkk^0 (\muo^+)^2] +\alpha_\tkk^0 \pi_\tkk^+ \pi_\tkk^- \alpha_{\tkk}^{1-}\muo^+ \\
&\geq \alpha_\tkk^0 \pi_\tkk^+ \pi_\tkk^- \alpha_{\tkk}^{1-}\muo^+\geq 0, \nonumber
\end{align*}
which shows that $\varphi\big((\mu_c^+)^2,\;\sbt\;\big)$ is a nondecreasing function in the second argument. Using that $(\mu_c^-)^2\leq  \mu_{c^2}^-$ and \eqref{eqn:bound:varphi:1}, we finally get that
\begin{equation}\label{eqn:bound:varphi:2}
\varphi\big((\mu_c^+)^2, (\mu_c^-)^2\big)\leq \varphi\big((\mu_c^+)^2,\mu_{c^2}^-\big)\leq \alpha_\tkk^0 \gamma_{t_k} +N_k.
\end{equation}
At this point, we only need to prove that $0\leq \varphi\big((\mu_c^+)^2, (\mu_c^-)^2\big)$. Evaluating, we have that
{\small
\begin{align}\label{eqn:varphi:evaluated} \nonumber
\varphi\big((\mu_c^+)^2, (\mu_c^-)^2\big) &=\muo^+\muo^-\bigg\{\alpha_\tkk^0\Big[(\pi_\tkk(1,1)\alpha_\tkk^{1,1})^2\muo^+\muo^- - {\pi^+_{\tkk}}{ \pi^-_{\tkk}}(\alpha^{1+}_{\tkk}\muo^+-1)(\alpha^{1-}_{\tkk}\muo^--1)\Big] \\ \nonumber
& \qquad \qquad 
+\muo^+(\alpha_\tkk^{1+}\pi_\tkk^+)^2\pi_\tkk^-(\alpha^{1-}_{\tkk}\muo^--1)+ \muo^-(\alpha_\tkk^{1-}\pi_\tkk^-)^2\pi_\tkk^+(\alpha^{1+}_{\tkk}\muo^+-1) \\ \nonumber
& \qquad \qquad 
-2\alpha_{\tkk}^{1+}\alpha_{\tkk}^{1-} \pi_\tkk^+ \pi_\tkk^- \pi_\tkk(1,1)\alpha_\tkk^{1,1} \muo^+\muo^-           \bigg\} \\
&=\muo^+\muo^- \ell(\muo^+), 
\end{align}}
where 
\begin{align*}
\ell(\bz)&:=\alpha_\tkk^0\Big[\bz\muo^-(\pi_\tkk(1,1)\alpha_\tkk^{1,1})^2 - {\pi^+_{\tkk}}{ \pi^-_{\tkk}}(\alpha^{1+}_{\tkk}\bz-1)(\alpha^{1-}_{\tkk}\muo^--1)\Big] \\
& \qquad \qquad 
+\bz(\alpha_\tkk^{1+}\pi_\tkk^+)^2\pi_\tkk^-(\alpha^{1-}_{\tkk}\muo^--1)+ \muo^-(\alpha_\tkk^{1-}\pi_\tkk^-)^2\pi_\tkk^+(\alpha^{1+}_{\tkk}\bz-1)   \\
& \qquad \qquad 
-2\,\bz \,\muo^-\alpha_{\tkk}^{1+}\alpha_{\tkk}^{1-} \pi_\tkk^+ \pi_\tkk^- \pi_\tkk(1,1)\alpha_\tkk^{1,1}.
\end{align*}
{Due to \eqref{ineq:alpha0}, we have
\begin{align} \nonumber
\ell(0) &= \alpha_\tkk^0\pi_\tkk^+\pi_\tkk^-(\alpha_\tkk^{1-}\muo^--1) - \muo^-(\alpha_\tkk^{1-}\pi_\tkk^-)^2\pi_\tkk^+ \\ \nonumber
&= \pi_\tkk^+\pi_\tkk^-\alpha_\tkk^{1-}\muo^-(\alpha_\tkk^0-\alpha_\tkk^{1-}\pi_\tkk^-)- \alpha_\tkk^0\pi_\tkk^+\pi_\tkk^- \\ \label{ineq:ell0}
&\geq - \alpha_\tkk^0\pi_\tkk^+\pi_\tkk^-  >0.
\end{align}
We shall prove that  $d\ell/d\bz\geq0$, which will imply that $\ell(\mu_c^+)\geq\ell(0)>0$ and, finally, that $\varphi\big((\mu_c^+)^2, (\mu_c^-)^2\big)>0$ due to Equation \eqref{eqn:varphi:evaluated}.  Finally, by using the previous equation together with \eqref{eqn:bound:varphi:2} we get that
\[
0<\varphi\big((\mu_c^+)^2, (\mu_c^-)^2\big)\leq \alpha_\tkk^0 \gamma_{t_k} +N_k,
\]
 and will conclude the proof. 

It remains to show that $d\ell/d\bz\geq0$. Taking the derivative with respect to $\bz$,
\begin{align} \label{eqn:derivative:ell}\nonumber
\frac{d \ell}{d \bz} &= \alpha_\tkk^0 \big[(\pi_\tkk(1,1)\alpha_\tkk^{1,1})^2 \muo^- -  \pi^+_{\tkk}{ \pi^-_{\tkk}}\alpha^{1+}_{\tkk}(\alpha_\tkk^{1-}\muo^--1) \big]   \\ \nonumber
&\quad
+(\alpha_\tkk^{1+}\pi_\tkk^+)^2\pi_\tkk^-(\alpha^{1-}_{\tkk}\muo^--1) + \muo^-(\alpha_\tkk^{1-}\pi_\tkk^-)^2\pi_\tkk^+\alpha^{1+}_{\tkk}  \\ \nonumber
&\quad -2\alpha_{\tkk}^{1+}\alpha_{\tkk}^{1-} \pi_\tkk^+ \pi_\tkk^- \pi_\tkk(1,1)\alpha_\tkk^{1,1} \muo^- \\
&=\pi^+_{\tkk} \pi^-_{\tkk}\alpha^{1+}_{\tkk}(\alpha_\tkk^0-\alpha_\tkk^{1+}\pi_\tkk^+) + \muo^- \Psi({\alpha_\tkk^{1,1}}\pi_\tkk(1,1)),
\end{align}
where
\begin{equation}\label{FuntPsi0}
\begin{aligned}
\Psi(\bz) &:=\alpha_\tkk^0 \bz^2 -2\bz\alpha_{\tkk}^{1+}\alpha_{\tkk}^{1-} \pi_\tkk^+ \pi_\tkk^- + (\alpha_{\tkk}^{1+}\pi_\tkk^+)^2 \alpha_{\tkk}^{1-}\pi_\tkk^- \\
&\qquad  +(\alpha_{\tkk}^{1-}\pi_\tkk^-)^2 \alpha_{\tkk}^{1+}\pi_\tkk^+ - \alpha_\tkk^0\alpha_{\tkk}^{1+}\alpha_{\tkk}^{1-} \pi_\tkk^+ \pi_\tkk^-.
\end{aligned}
\end{equation}
Due to \eqref{ineq:alpha0}, the first term of \eqref{eqn:derivative:ell} is nonegative.  It is also easy to see that the quadratic function $\Psi$ attaining its minimum value over any interval $[a,b]$ at the endpoints since $\alpha_\tkk^0\leq 0$. On the other hand, observe that $\alpha_\tkk(\onep+\onem-1) \geq \alpha_\tkk \onep\onem$. Then, by taking conditional expectation with respect to $\F$ and using \eqref{eqn:comparison:alpha}, we obtain the bounds
\begin{equation}\label{eqn:bounds:alpha11pi11}
\alpha_\tkk^{1+}\pi_\tkk^+ \vee \alpha_\tkk^{1-}\pi_\tkk^- \leq \alpha_\tkk^{1,1}\pi_\tkk(1,1) \leq (\alpha_\tkk^{1+}\pi_\tkk^+ + \alpha_\tkk^{1-}\pi_\tkk^- - \alpha_\tkk^0) \wedge 0.
\end{equation}
Therefore, to show that \eqref{eqn:derivative:ell} is nonnegative, it suffices to check that 
\[
	{\rm (1)}\;\Psi(\alpha_\tkk^{1+}\pi_\tkk^+ \vee \alpha_\tkk^{1-}\pi_\tkk^-)\geq{}0,\quad {\rm (2)}\;\Psi((\alpha_\tkk^{1+}\pi_\tkk^+ + \alpha_\tkk^{1-}\pi_\tkk^- - \alpha_\tkk^0) \wedge 0)\geq{}0.
\]
To show (1), we consider the cases (a) $\alpha_\tkk^{1+}\pi_\tkk^+ \geq \alpha_\tkk^{1-}\pi_\tkk^-$ or  (b) $\alpha_\tkk^{1+}\pi_\tkk^+ \leq \alpha_\tkk^{1-}\pi_\tkk^-$. To show (2) above, we consider the cases  (c) $(\alpha_\tkk^{1+}\pi_\tkk^+ + \alpha_\tkk^{1-}\pi_\tkk^- - \alpha_\tkk^0)\geq{}0$; or (d) $(\alpha_\tkk^{1+}\pi_\tkk^+ + \alpha_\tkk^{1-}\pi_\tkk^- - \alpha_\tkk^0)\leq{}0$. 

Suppose that (a) $\alpha_\tkk^{1+}\pi_\tkk^+ \geq \alpha_\tkk^{1-}\pi_\tkk^-$. Then, in light of \eqref{ineq:alpha0},
{\small
\begin{align*}
\Psi(\alpha_\tkk^{1+}\pi_\tkk^+) &=\alpha_\tkk^{1+}\pi_\tkk^+\Big[(\alpha_\tkk^{1-}\pi_\tkk^-)^2 - (\alpha_\tkk^{1-}\pi_\tkk^-)(\alpha_\tkk^{1+}\pi_\tkk^+)  + \alpha_\tkk^0(\alpha_\tkk^{1+}\pi_\tkk^+-\alpha_\tkk^{1-}\pi_\tkk^-) \Big] \\
&=\alpha_\tkk^{1+}\pi_\tkk^+(\alpha_\tkk^{1-}\pi_\tkk^- - \alpha_\tkk^{1+}\pi_\tkk^+)(\alpha_\tkk^{1-}\pi_\tkk^- - \alpha_\tkk^0)\geq 0.
\end{align*}}
Similarly, (b) if $\alpha_\tkk^{1-}\pi_\tkk^- \geq \alpha_\tkk^{1+}\pi_\tkk^+$, then $\Psi(\alpha_\tkk^{1-}\pi_\tkk^-)\geq 0$. 

Now, suppose that (c) $\alpha_\tkk^{1+}\pi_\tkk^+ + \alpha_\tkk^{1-}\pi_\tkk^- - \alpha_\tkk^0 \geq 0$. Then, 
\[
\Psi(0) =\alpha_\tkk^{1+}\alpha_\tkk^{1-}\pi_\tkk^+\pi_\tkk^- (\alpha_\tkk^{1+}\pi_\tkk^+ + \alpha_\tkk^{1-}\pi_\tkk^- - \alpha_\tkk^0) \geq 0.
\]
The only case left is when  (d) $\alpha_\tkk^{1+}\pi_\tkk^+ + \alpha_\tkk^{1-}\pi_\tkk^- - \alpha_\tkk^0 \leq 0$. Together with \eqref{ineq:alpha0}, we will have
\begin{equation}\label{eqn:bounds:alpha1p}
\alpha_\tkk^0 \leq \alpha_\tkk^{1+}\pi_\tkk^+ \leq \alpha_\tkk^0 - \alpha_\tkk^{1-}\pi_\tkk^-,
\end{equation}
and, by rearranging terms,
\begin{align*}
\Psi(\alpha_\tkk^{1+}\pi_\tkk^+ + \alpha_\tkk^{1-}\pi_\tkk^- - \alpha_\tkk^0) &=(\alpha_\tkk^0-\alpha_\tkk^{1-}\pi_\tkk^-)(\alpha_\tkk^{1+}\pi_\tkk^+)^2 + \alpha_\tkk^0(\alpha_\tkk^0-\alpha_\tkk^{1-}\pi_\tkk^-)^2 \\
&\quad - (2\alpha_\tkk^0-\alpha_\tkk^{1-}\pi_\tkk^-)(\alpha_\tkk^0-\alpha_\tkk^{1-}\pi_\tkk^-)(\alpha_\tkk^{1+}\pi_\tkk^+) \\
&=:\Xi(\alpha_\tkk^{1+}\pi_\tkk^+),
\end{align*}
where 
\[
\Xi(\bz)=(\alpha_\tkk^0-\alpha_\tkk^{1-}\pi_\tkk^-)\bz^2 - (2\alpha_\tkk^0-\alpha_\tkk^{1-}\pi_\tkk^-)(\alpha_\tkk^0-\alpha_\tkk^{1-}\pi_\tkk^-)\bz + \alpha_\tkk^0(\alpha_\tkk^0-\alpha_\tkk^{1-}\pi_\tkk^-)^2
\]
is a quadratic function, opening downward due to \eqref{ineq:alpha0}. Furthermore, at the lower and upper bounds of \eqref{eqn:bounds:alpha1p}, we have  $\Xi(\alpha_\tkk^0) =\Xi(\alpha_\tkk^0 - \alpha_\tkk^{1-}\pi_\tkk^-) =0$.
Therefore,  $\Xi(\bz) \geq 0$, for all $\bz \in \big[\alpha_\tkk^0,(\alpha_\tkk^0 - \alpha_\tkk^{1-}\pi_\tkk^-)\big]$. Finally, we conclude that $\Psi(\alpha_\tkk^{1+}\pi_\tkk^+ + \alpha_\tkk^{1-}\pi_\tkk^- - \alpha_\tkk^0)=\Xi(\alpha_\tkk^{1+}\pi_\tkk^+)\geq{}0$ in the case (iv) as desired.} 
\end{itemize}
\end{proof}

\begin{proof}[Proof of Corollary \ref{cor:candidates}] 

Refer to the proof of {Theorem} \ref{prop:measurability}. Denote the right hand side of Equation (\ref{eqn:Quadratic:L}) as $\sup_{L_{t_k}^\pm}{f}(\lp,\lm)$. It is easy to observe that ${f}(\lp,\lm)$ is a quadratic function of $\lp$ and $\lm$. Setting the partial derivatives with respect to $\lp$ and $\lm$ equal to $0$, we have that
\begin{align}\label{PDWRTL}
\begin{split}
0=\frac{\partial f}{\partial \lp}&=2{\pi^+_{\tkk}}(\alpha^{1+}_{\tkk}\mut^+-\muo^+)\lp
+{\pi^+_{\tkk}}\big[\muoo^++ {h}_{\tkk}^{1+}\muo^++\alpha^{1+}_{\tkk}(2\muo^+I_\tk-2\muto^+)\big]\\
&-2\alpha^{1,1}_{\tkk}{\pi_\tkk(1,1)}\muo^+\muo^-\lm+2\alpha^{1,1}_{\tkk}{\pi_\tkk(1,1)}\muo^+\muoo^-\\
0=\frac{\partial f}{\partial \lm}&=2{\pi^-_{\tkk}}(\alpha^{1-}_{\tkk}\mut^--\muo^-)\lm
+{\pi^-_{\tkk}}\big[\muoo^-- {h}_{\tkk}^{1-}\muo^-+\alpha^{1-}_{\tkk}(-2\muo^-I_\tk-2\muto^-)\big]\\
&-2\alpha^{1,1}_{\tkk}{\pi_\tkk(1,1)}\muo^+\muo^-\lp+2\alpha^{1,1}_{\tkk}{\pi_\tkk(1,1)}\muo^-\muoo^+.
\end{split}
\end{align}
Solving for $\lp$ and $\lm$, we get that {the unique stationary points ${L}_{\tk}^{+,*}$ and ${L}_{\tk}^{-,*}$ are given by
\begin{align}\label{PDWRTLSol}
	\begin{split}
	L_{\tk}^{+,*}&={}^{\scaleto{(1)}{5pt}}\!A^+_{\tk}I_\tk+{}^{\scaleto{(2)}{5pt}}\!A^+_{\tk}+{}^{\scaleto{(3)}{5pt}}\!A^+_{\tk},\\
	L_{\tk}^{-,*}&=-{}^{\scaleto{(1)}{5pt}}\!A^-_{\tk}I_\tk-{}^{\scaleto{(2)}{5pt}}\!A^-_{\tk}+{}^{\scaleto{(3)}{5pt}}\!A^-_{\tk},
	\end{split}
	\end{align}
as pointed out in Eq.~(\ref{eqn:OptimalL}). To prove that the {stationary} point $\Big({L}_{\tk}^{+,*},{L}_{\tk}^{-,*}\Big)$ is the  maximum of $f$, we apply the second derivative test. Indeed, 
\[
\left(\frac{\partial^2 f}{\partial \left(L_{t_k}^\pm\right)^2}\right)=2\pi_\tkk^\pm(\alpha^{1\pm}_{\tkk}\mut^\pm -\muo^\pm)<0,
\]
because, due to Lemma \ref{lemma:alpha}, we have $\alpha_{t_k}<0$ and, thus, $\alpha_\tk^{1\pm} =\mathbb{E}[\alpha_\tk \,|\,\mathcal{F}_{t_{k-1}},\mathbbm{1}^\pm_\tk=1]\leq{}0$.
It remains to show that
\begin{align}\nonumber
B_k& = \left(\frac{\partial^2 f}{\partial \left(\lp\right)^2}\right)\left(\frac{\partial^2 f}{\partial \left(\lm\right)^2}\right)-\left(\frac{\partial^2 f}{\partial \lp\partial \lm}\right)^2 \\
&= 4{\pi^+_{\tkk}}{\pi^-_{\tkk}}(\alpha^{1+}_{\tkk}\mut^+-\muo^+)(\alpha^{1-}_{\tkk}\mut^--\muo^-)-4\big[{\pi_\tkk(1,1)}\alpha^{1,1}_{\tkk}\muo^+\muo^-\big]^2>0.\label{ScndInqH}
\end{align}
However, $B_k$ above is just $-4\gamma_{t_k}$, with $\gamma_{t_k}$ defined as in (\ref{eqn:Nk:Dk}), and it was shown in the proof of Lemma \ref{lemma:alpha} that $\gamma_{t_k}<0$ (see (\eqref{eq:Dneg}). The proof is now complete.}
\end{proof}

\begin{proof}[Proof of Theorem~\ref{VeriThrm1}] \label{ProofVerifyH}
	
{Throughout, ${W}_{t_{i}},{I}_{t_{i}}$, for $i=k,\dots,N+1$, are  the cash holding and inventory processes resulting from {adopting an} admissible placement strategy ${L}_{t_{i}}^{\pm}$, $i=k,\dots,N$. {Similarly}, for $i = k+1, \dots, N+1$, ${W}_{t_{i}}^{*},{I}_{t_{i}}^{*}$ are the resulting cash holding and inventory processes, starting from time $t_k$ at the initial states $W_\tk,{I_\tk}$,  when setting {$L_{t_{i}}^{\pm}={L}_{t_{i}}^{\pm,{*}}$}. 
	First note that, for an arbitrary admissible placement strategy $L_{t_{i}}^{\pm}$, 
	$\{v(t_i,S_{t_i},W_{t_i},I_{t_i})\}_{i=k,\dots,N+1}$ is a supermartingale since
	\begin{align}\nonumber
		\mathbbm{E}\big[v(t_{i+1},S_{t_{i+1}},W_{t_{i+1}},I_{t_{i+1}})|\mathcal{F}_{t_{i}}\big]&\leq\sup_{\widehat{L}^\pm_{t_{i}}}\mathbbm{E}\big[v(t_{i+1},S_{t_{i+1}},\widehat{W}_{t_{i+1}},\widehat{I}_{t_{i+1}})|\mathcal{F}_{t_{i}}\big]\\
		\label{EqDmH}
		&=v(t_{i},S_{t_{i}},W_{t_{i}},I_{t_{i}}),
	\end{align}
{where above $\widehat{W}_{t_{i+1}}$ and $\widehat{I}_{t_{i+1}}$ are the time-$t_{i+1}$ cash holding and inventory for an arbitrary admissible placement strategy  $\widehat{L}_{t_{i}}^{\pm}$ when $\widehat{W}_{t_{i}}={W}_{t_{i}}$ and $\widehat{I}_{t_{i}}={I}_{t_{i}}$.} 
	The equation in \eqref{EqDmH} follows from~(\ref{eq113}) and {Corollary} \ref{cor:candidates}. That is, {$\alpha_{t_{i}},h_{t_{i}},g_{t_{i}}$
	 are picked in order for (\ref{EqDmH}) to hold true. 
	 
	From the supermartingale condition}, we then have that 
	\begin{align}
		\begin{split}
			v(t_k,\sk,W_\tk,I_\tk)&\geq \sup_{{(L^\pm_{t_i})_{k\leq i\leq N}}}\E[v(T,S_T,W_T,I_T)|\F]\\
			&=\sup_{{(L^\pm_{t_i})_{k\leq i\leq N}}}\E[W_T+S_T I_T-\lambda I_T^2|\F]\\
			&=V_\tk.\label{eq:vlarge}
		\end{split}
	\end{align}
The first equality in Eq.~(\ref{eq:vlarge}) holds because $v(T,S_T,W_T,I_T)=W_T+S_T I_T-\lambda I_T^2$ by the terminal conditions $\alpha_T=-\lambda,g_T=0,h_T=0$.
	
	Next we prove that $v(t_k,\sk,W_\tk,I_\tk)\leq V_\tk$. To this end, recall from~(\ref{eq113}) and {Corollary} \ref{cor:candidates} that {${L}_{t_{i}}^{\pm,{*}}$ are chosen so that}
%
	\[
	v(t_{i},S_{t_{i}},{{W}_{t_{i}}^{*}},{{I}_{t_{i}}^{*}})=\E[v(t_{i+1},S_{t_{i+1}},{{W}^{*}_{t_{i+1}},{I}^{*}_{t_{i+1}}})|\mathcal{F}_{t_{i}}],
	\]
	for all $i=k,\dots,{N}$. Hence, {recalling that we set ${W}^{*}_\tk=W_\tk$ and ${I}^{*}_\tk=I_\tk$, by induction,}
	\begin{align*}
		v(t_k,\sk,W_\tk,I_\tk)=v(t_k,{\sk,{W}^{*}_\tk,{I}^{*}_\tk})&=\E[v(t_{N+1},S_{t_{N+1}},{{W}^{*}_{t_{N+1}},{I}^{*}_{t_{N+1}}})|\mathcal{F}_{t_{k}}]\\
		&=\E[{{W}^{*}_T+S_T {I}^{*}_T-\lambda ({I}^{*}_T)^2}|\F].
	\end{align*}
It also trivially follows that
	\[
	\E{[{W}^{*}_T+S_T {I}^{*}_T-\lambda ({I}^{*}_T)^2}|\F]\leq{}\sup_{{(L^\pm_{t_i})_{k\leq i\leq N}}}\E[W_T+S_T I_T-\lambda I_T^2|\F]=V_{t_{k}}.
	\]
	We then conclude that} $v(t_k,\sk,W_\tk,I_\tk)\leq{}V_{t_{k}}$, which combined with (\ref{eq:vlarge}) implies that \\$v(t_k,\sk,W_\tk,I_\tk)={}V_{t_{k}}$.
\end{proof}

\subsection{Proofs of Section \ref{Sec:General:process}: Optimal Strategy for a General Midprice}\label{appdx:general:price}

\begin{proof}[Proof of Theorem \ref{thm:General:Price:controls}]

{\DRed For simplicity, we write $\rho_k^{\pm}=\rho_{t_{k}}^{\pm}$ and $\psi_k^{\pm}=\psi_{t_{k}}^{\pm}$}. {For future reference, define $\widetilde{h}_\tkk^0$, $\widetilde{h}_\tkk^{1\pm}$, and $\widetilde{h}_\tkk^{1,1}$ using the notation \eqref{Dfnalh0}.
Let us start by writing the optimization problem \eqref{eqn:optimalprob} in terms of the ansatz (\ref{eq:V_NMGb}):
\begin{equation}\label{eq1133}
\begin{aligned}
W_\tk+\alpha_{\tk}I_\tk^2+&\sk I_\tk+\widetilde{h}_{\tk}I_\tk =\sup_{L^\pm_{t_k} \in \mathcal{A}_{t_k}} \mathbb{E}[W_\tkk+\alpha_{\tkk}I_\tkk^2+\skk I_\tkk +\widetilde{h}_{\tkk}I_\tkk \,|\,\F].
\end{aligned}
\end{equation}
We will prove the result by backwards induction.}  
Consider the following statements:
\begin{itemize}
	\item[(i)] 
	For $\delta\in\{+,-\}$, {we have:}
\begin{align}\label{TrickyNdInd1}
\mathbb{E}\big[\widetilde{h}_{t_{k+1}}\mathbbm{1}_{t_{k+1}}^{\delta}c_{t_{k+1}}^{\delta}\big|\mathcal{F}_{t_{k}}\big]&=
\widetilde{h}_{t_{k+1}}^{1\delta}\pi_{t_{k+1}}^\delta\muo^\delta,\\
\label{TrickyNdInd2}
\mathbb{E}\big[\widetilde{h}_{t_{k+1}}\mathbbm{1}_{t_{k+1}}^{\delta}c_{t_{k+1}}^{\delta}p_{t_{k+1}}^{\delta}\big|\mathcal{F}_{t_{k}}\big]&=
\widetilde{h}_{t_{k+1}}^{1\delta}\pi_{t_{k+1}}^\delta\muoo^\delta.
\end{align}
	\item[(ii)] The optimal controls $\widetilde{L}_{\tk}^{\pm,*}$ that solve (\ref{eq1133}) under {dynamics \eqref{eqi11} with terminal conditions $\alpha_{t_{N+1}}=-\lambda$ and $\widetilde{h}_{t_{N+1}}=0$} are given by
\begin{equation}\label{DfnOptLbb}
\begin{aligned}
\widetilde{L}_{\tk}^{+,*}={}^{\scaleto{(1)}{5pt}}\!A^+_{\tk}I_\tk+{}^{\scaleto{(2)}{5pt}}\!\widetilde{A}^+_{\tk}
+{}^{\scaleto{(3)}{5pt}}\!\widetilde{A}^+_{\tk},\qquad
\widetilde{L}_{\tk}^{-,*}=-{}^{\scaleto{(1)}{5pt}}\!A^-_{\tk}I_\tk-^{\scaleto{(2)}{5pt}}\!\!\widetilde{A}^-_{\tk}{+}{}\,
{}^{\scaleto{(3)}{5pt}}\!\widetilde{A}^-_{\tk},
\end{aligned}
\end{equation}
where 
\begin{align}
\label{eq:A1NMGb}
&{}^{\scaleto{(2)}{5pt}}\!\widetilde{A}^\pm_{\tk}=\frac{\widetilde{h}_\tkk^{1\pm}\rho_k^\pm-\widetilde{h}_\tkk^{1\mp}\psi_k^\pm}{2\gamma_{t_{k}}},\qquad
{}^{\scaleto{(3)}{5pt}}\!\widetilde{A}^\pm_{\tk}={}^{\scaleto{(3)}{5pt}}\!A^\pm_{\tk}  \pm \frac{\rho_k^\pm-\psi_k^\pm}{2\gamma_\tk}\Delta_{t_k}^{t_k},
\end{align}
whereas $\Aone^\pm$ and ${}^{\scaleto{(3)}{5pt}}\!A^\pm_{\tk}$ are the same as in {Theorem \ref{prop:measurability}}.
	\item[(iii)]  The random variables $\widetilde{h}_{\tk}$ satisfy the following iterative equation,
\begin{align}\label{eq177}
\begin{split}
\widetilde{h}_{\tk}= \widetilde{h}^0_{\tkk}&+\sum_{\delta=\pm} \pi^\delta_\tkk\Bigg\{2\Big[\alpha^{1\delta}_{\tkk}\mut^\delta-\muo^\delta\Big] \Big[{}^{\scaleto{(1)}{5pt}}\!A^\delta_{\tk}\big({\delta\;{}^{\scaleto{(3)}{5pt}}\!\widetilde{A}^\delta_{\tk}+ {}^{\scaleto{(2)}{5pt}}\!\widetilde{A}^\delta_{\tk}}\big)\Big]+2\alpha^{1\delta}_{\tkk}\muo^\delta\big({\delta\,{}^{\scaleto{(3)}{5pt}}\!\widetilde{A}^\delta_{\tk}+ {}^{\scaleto{(2)}{5pt}}\!\widetilde{A}^\delta_{\tk}}\big)\\
&\qquad\qquad\qquad -2\delta\alpha^{1\delta}_{\tkk}\muoo^\delta
+\delta {}^{\scaleto{(1)}{5pt}}\!A^\delta_{\tk}\Big(\muoo^\delta+\delta \widetilde{h}^{1\delta}_{\tkk}\muo^\delta-2\alpha^{1\delta}_{\tkk}\muto^\delta\Big)\Bigg\}\\
&-2\alpha^{1,1}_{\tkk} \pi_{t_{k+1}}(1,1)\muo^+\muo^-\Bigg[{}^{\scaleto{(1)}{5pt}}\!A^+_{\tk}\Big({}^{\scaleto{(3)}{5pt}}\!\widetilde{A}^-_{\tk}-{}^{\scaleto{(2)}{5pt}}\!\widetilde{A}^-_{\tk}\Big)-{}^{\scaleto{(1)}{5pt}}\!A^-_{\tk}\Big({}^{\scaleto{(2)}{5pt}}\!\widetilde{A}^+_{\tk}+{}^{\scaleto{(3)}{5pt}}\!\widetilde{A}^+_{\tk}\Big)\\
&\qquad +{}^{\scaleto{(1)}{5pt}}\!A^-_{\tk}\frac{\muoo^+}{\muo^+}-{}^{\scaleto{(1)}{5pt}}\!A^+_{\tk}\frac{\muoo^-}{\muo^-}\Bigg] +\Delta_{t_k}^{t_k}\Big({}^{\scaleto{(1)}{5pt}}\!A^+_{\tk}\pi_\tkk^+\muo^+ + {}^{\scaleto{(1)}{5pt}}\!A^-_{\tk}\pi_\tkk^-\muo^-+1\Big).
\end{split}
\raisetag{4\normalbaselineskip}
\end{align}
while
\begin{align}\label{eq:gNMGb}
\begin{split}
        \widetilde{g}_{\tk}&=\widetilde{g}_{\tkk}^{\;0}+\sum_{\delta=\pm}{\pi^\delta_\tkk}\Bigg\{\Big(\alpha_{\tkk}^{1\delta}\mut^\delta-\muo^\delta\Big)\Big({{}^{\scaleto{(3)}{5pt}}\!\widetilde{A}^\delta_{\tk}+\delta\, {}^{\scaleto{(2)}{5pt}}\!\widetilde{A}^\delta_{\tk}}\Big)^2+\alpha_{\tkk}^{1\delta}\mutt^\delta-\delta\, \widetilde{h}_{\tkk}\muoo^\delta\\
        &\quad\qquad\qquad\qquad\qquad+\Big(\muoo^\delta+\delta \widetilde{h}_{\tkk}\muo^\delta-2\alpha_{\tkk}^{1\delta}\muto^\delta\Big)\Big({{}^{\scaleto{(3)}{5pt}}\!\widetilde{A}^\delta_{\tk}+{(\delta\, {}^{\scaleto{(2)}{5pt}}\!\widetilde{A}^\delta_{\tk}})}\Big)\Bigg\}\\
        &\qquad\qquad-2\alpha_{\tkk}^{1,1}{\pi_\tkk(1,1)}\muo^+\muo^-\Bigg[\Big({}^{\scaleto{(2)}{5pt}}\!\widetilde{A}^+_{\tk}+{}^{\scaleto{(3)}{5pt}}\!\widetilde{A}^+_{\tk}\Big)\Big({}^{\scaleto{(3)}{5pt}}\!\widetilde{A}^-_{\tk}-{}^{\scaleto{(2)}{5pt}}\!\widetilde{A}^-_{\tk}\Big)\\
        &\quad\qquad\qquad\qquad\qquad-\frac{\muoo^+}{\muo^+}\Big({{}^{\scaleto{(3)}{5pt}}\!\widetilde{A}^-_{\tk}-{}^{\scaleto{(2)}{5pt}}\!\widetilde{A}^-_{\tk}}\Big)-\frac{\muoo^-}{\muo^-}\Big({}^{\scaleto{(2)}{5pt}}\!\widetilde{A}^+_{\tk}+{}^{\scaleto{(3)}{5pt}}\!\widetilde{A}^+_{\tk}\Big)+\frac{\muoo^+\muoo^-}{\muo^-\muo^+}\Bigg]\\
        &\qquad\qquad+\Delta_{t_{k}}^{t_k}\bigg[\Big({}^{\scaleto{(3)}{5pt}}\!\widetilde{A}^\delta_{\tk}+ {}^{\scaleto{(2)}{5pt}}\!\widetilde{A}^\delta_{\tk}\Big)\pi_\tkk^+\muo^+-\Big({}^{\scaleto{(3)}{5pt}}\!\widetilde{A}^\delta_{\tk}- {}^{\scaleto{(2)}{5pt}}\!\widetilde{A}^\delta_{\tk}\Big)\pi_\tkk^-\muo^--\pi_\tkk^+\muoo^++\pi_\tkk^-\muoo^-\bigg]
\end{split}
\raisetag{4\normalbaselineskip}
\end{align}
	\item[(iv)] Equation \eqref{eq:LtildeNMGbb} holds true (at time $t_k$) and the random variables $\rho_k^\pm$ and $\psi_k^\pm$, as defined in Equations {\eqref{eqn:gamma:eta}}, are $\mathcal{F}_{{t_k}}-$measurable while the random variables $\xi_{k+1}$ given by Eq.~\eqref{eqn:rho:k} are  $\mathcal{F}_{{t_{k+1}}}-$ measurable.
\end{itemize}

{To start, note that the statement (i) is immediate for $k=N$ due to the terminal condition on $\widetilde{h}$. 
The strategy that we will employ to finish the proof is the following: we will show that if  statement (i) holds true for $k=j,j+1,\dots,N$, then statements (ii), (iii), and (iv) will hold true for $k=j$. In the final step, we prove that the statement (i) holds true for $k=j-1$ if (i)-(iv) holds for $k=j$.}

\medskip

Let us start with the first step described in the previous paragraph. 
Assume that statement (i) is true for $k=j, j+1,\dots,N$. By substituting the values of $W_{t_{j+1}}$ and $I_{t_{j+1}}$ given by Equation \eqref{eqi11}-\eqref{eqw11} into the optimization problem \eqref{eq1133} with $k=j$ we get,
{\small
\begin{align}
\begin{split}
\label{eqn:expanded:general:1}
V_{t_{j}}=\sup_{L_{t_j}^\pm\in\mathcal{A}_{t_j}}\Ex\Bigg\{&\sum_{\delta=\pm}\onedj\big[-\cdj(\ldj)^2+(\cdj\Pdj-\delta\cdj\sj)\ldj+\delta\cdj\Pdj\sj\big]\\
&+\alpha_{\tjj}\bigg\{I_\tj^2+\sum_{\delta=\pm}\onedj\Big[(\cdj)^2(\ldj)^2+\big(2\delta I_\tj\cdj-2(\cdj)^2\Pdj\big)\ldj\\
&\qquad\qquad\quad\qquad\qquad\qquad+(\cdj\Pdj)^2-2\delta I_\tj\cdj\Pdj\Big]\\
&\qquad\qquad+2\onepj\onemj\cpj\cmj(-\lpj\lmj+\Ppj\lmj+\Pmj\lpj-\Ppj\Pmj)\bigg\}\\
&+\sjj\bigg[I_\tj+\sum_{\delta=\pm}\onedj(-\delta\cdj\Pdj+\delta\cdj\ldj)\bigg]
\\
&+ \widetilde{h}_{\tjj}I_\tj+\sum_{\delta=\pm}\onedj(-\delta {\widetilde{h}_{\tjj}}\cdj\Pdj+\delta {\widetilde{h}_{\tjj}}\cdj\ldj)+{\Purple \widetilde{g}_{\tjj}}\Bigg|\Fj\Bigg\}.
\end{split}
\end{align}}
Computing the conditional expectations on the right-hand side using statement (i) and the techniques used in the proof of Theorem \ref{prop:measurability}, we get that
\begin{align}\label{eqn:Quadratic:L:General}
\begin{split}
&\alpha_{t_j}I_{t_j}^2+S_{t_j} I_{t_j}+\widetilde{h}_{t_j}I_{t_j}+\widetilde{g}_{t_j}\\
&=\sup_{L_{t_j}^\pm\in\mathcal{A}_{t_j}}\ \Bigg[ \sum_{\delta=\pm}{\pi^\delta_{t_{j+1}}}\bigg\{(\alpha^{1\delta}_{t_{j+1}}\mut^\delta-\muo^\delta)(L_{t_j}^\delta)^2	+\big[\muoo^\delta+\delta\muo^\delta\big( \widetilde{h}^{1\delta}_{t_{j+1}}+\Delta_{t_j}^{t_j}\big)+\alpha^{1\delta}_{t_{j+1}}(2\delta\muo^\delta I_{t_j}-2\muto^\delta)\big]L_{t_j}^\delta\\
&\qquad\qquad \qquad\qquad\quad+\alpha^{1\delta}_{t_{j+1}}(\mutt^\delta-2\delta\muoo^\delta I_{t_j})-\delta\muoo^\delta\big( \widetilde{h}^{1\delta}_{t_{j+1}}+\Delta_{t_j}^{t_j}\big)\bigg\}\\ 
&\qquad\qquad\qquad+2\alpha^{1,1}_{t_{j+1}}{\pi_{t_{j+1}}(1,1)}\Big(\muoo^+\muo^-L_{t_j}^-+\muo^+\muoo^-L_{t_j}^+-\muo^+\muo^-L_{t_j}^+L_{t_j}^- -\muoo^+\muoo^-\Big)\\
&\qquad\qquad\qquad+\alpha^0_{t_{j+1}}I_{t_j}^2+I_{t_j}S_{t_j}+ \widetilde{h}^0_{t_{j+1}} I_{t_j}+\Delta_{t_j}^{t_j}I_{t_j}+\widetilde{g}^{\;0}_{t_{j+1}}\Bigg].
\end{split}
\raisetag{3\baselineskip}
\end{align}
Then, since the function of the right-hand side is a quadratic function of the controls $L_{t_j}^\pm$, it can be shown, as in the proof Theorem \ref{prop:measurability}, that the optimal controls are the ones given in statement (ii) with $k=j$. Furthermore, after plugging in the optimal controls; equating the coefficients of $I_{t_j}$ on both sides; and the independent terms on both sides, we obtain statement (iii). 

\medskip

To obtain statement (iv), notice that by plugging the expression for {${}^{\scaleto{(2)}{5pt}}\!A^\pm_{t_j}$}, given by Equation \eqref{eq:A1}, into the Equation (\ref{eqn:htk}), we can rewrite {$h_\tj$} as:
\begin{align}\label{h1}
\nonumber
h_\tj&=\mathbb{E}(h_{t_{j+1}}\xi_{j+1}|\mathcal{F}_{t_j})+Z_j\\
&=\mathbb{E}(h_{t_{N+1}}\prod_{u=j}^{N}\xi_{u+1}|\F) +\sum_{u=j}^{N}\mathbb{E}(Z_u\prod_{i=j+1}^{u}\xi_i|\mathcal{F}_{t_j})
\end{align}
where $\{Z_k\}_{k=1}^N$ is a collection of $\Hkk-$measurable random variables. In the same fashion, by plugging the quantities defined in Equation \eqref{eq:A1NMGb} into Equation (\ref{eq177}), we can rewrite $\widetilde{h}_{\tj}$ as:
\begin{align}\label{h2}
\widetilde{h}_{\tj}&=\mathbb{E}\big[(h_{t_{j+1}}+\Delta_{t_j}^{t_j})\xi_{j+1}\big|\mathcal{F}_{t_j}\big]+Z_j=h_\tu+\sum_{u=j}^{N}\Delta_{t_u}^\tj\mathbb{E}\left[\prod_{i=j+1}^{u+1}\xi_i|\mathcal{F}_{t_j}\right].
\end{align}
Next, by substituting the value of  $\widetilde{h}_{\tj}$ in Equation (\ref{h2}) into the term  ${}^{\scaleto{(2)}{5pt}}\!\widetilde{A}^\pm_{\tj}$, as defined in Equation (\ref{eq:A1NMGb}), and then plugging this equivalent formula along with ${}^{\scaleto{(3)}{5pt}}\!\widetilde{A}^\pm_{\tu}$, as given by Equation \eqref{eq:A1NMGb} into Equation (\ref{DfnOptLbb}), we obtain Equation (\ref{eq:LtildeNMGbb}). To verify the measurability of $\rho_j^\pm$ and $\psi_j^\pm$, notice that by definition (see Equations \eqref{eqn:mucp}-\eqref{eqn:gfunc:pi}), the random variables $\pi^\pm_{t_{j+1}},\pi_{t_{j+1}}(1,1)$ and $\mu_{c^mp^n}^\pm$ are $\mathcal{F}_j$-measurable. Also, by definition, (see Equation \eqref{Dfnalh0}) the random variables $\{\alpha_{t_{j+1}}^0\},\{\alpha_{t_{j+1}}^{1\pm}\}$ and $\{\alpha_{t_{j+1}}^{1,1}\}$ are also $\mathcal{F}_{t_j}$-measurable. All of this implies that $\{\rho_j^\pm\}$ and $\{\psi_j^\pm\}$ are $\mathcal{F}_{t_j}-$measurable. However, due to the presence of the random variables $\onepj,\onemj$, the random variables $\xi_{j+1}$ cannot be $\mathcal{F}_{t_{j}}$-measurable but they are certainly $\mathcal{F}_{t_{j+1}}-$measurable.

\medskip

Finally, all what is left is to prove that if statements (i)-(iv) hold true for $k=j+1,j+2,\ldots,N$, then statement (i) holds true for $k=j$. Indeed, since Equations (\ref{TrickyNdInd1}) and (\ref{TrickyNdInd2}) hold for $k=j+1,\dots,N$, then Equation (\ref{h2}) holds for $k=j+1$. That is,
\[
\widetilde{h}_{t_{j+1}}=h_{t_{j+1}}+\sum_{i=j+1}^{N}\Delta_{t_i}^{t_{j+1}}\mathbb{E}(\prod_{u=j+2}^{i+1}\xi_u|\mathcal{F}_{t_{j+1}}).
\]
Multiplying both sides by $\onedj c_{t_{j+1}}^\delta$ and taking the conditional expectation with respect to $\mathcal{F}_{t_j}$, we have that
\begin{align}
\nonumber
\mathbb{E}\big[\widetilde{h}_{t_{j+1}}\onedj c_{t_{j+1}}^\delta\big|\mathcal{F}_{t_j}\big]&=\mathbb{E}\Bigg\{\left[h_{t_{j+1}}+\sum_{i=j+1}^{N}\Delta_{t_i}^{t_{j+1}}\mathbb{E}\left.\left(\prod_{u=j+2}^{i+1}\xi_u\right|\mathcal{F}_{t_{j+1}}\right)  \right]    \onedj c_{t_{j+1}}^\delta\Bigg|\mathcal{F}_{t_j}\Bigg\}\\
\nonumber
&=h_{t_{j+1}}^{1\delta}\pi_{t_{j+1}}^\delta\muo^\delta+\pi_{t_{j+1}}^\delta\muo^\delta \sum_{i=j+1}^{N}\Delta_{t_i}^{t_j}\mathbb{E}\Bigg[\mathbb{E}\bigg(\prod_{u=j+2}^{i+1}\xi_u\bigg|\mathcal{F}_{t_{j+1}}\bigg)\Bigg|\mathcal{F}_{t_j},\onedj=1\Bigg]\\
\nonumber
&=\widetilde{h}_{t_{j+1}}^{1\delta}\pi_{t_{j+1}}^\delta\muo^\delta.
\end{align}
where we used that $\mathbb{E}\left[\prod_{u=j+2}^{i+1}\xi_u|\mathcal{F}_{t_{j+1}}\right]$ is an $\mathcal{H}_{t_{j+1}}^\varpi$-measurable random variable. Thus, we have that Equation (\ref{TrickyNdInd1}) holds for {\DRed $k=j$}. Similarly we can prove Equation (\ref{TrickyNdInd2}) holds for $k=j$, which concludes the proof.
\end{proof}

\subsection{Proofs of Section \ref{sec:admissibility}: Admissibility of the Optimal Strategy}\label{App:proof:admissible}

\begin{proof}[Proof of Proposition \ref{prop:admissibility}] {Let us define the numerators of ${}^{\scaleto{(1)}{5pt}}\!A^+_{\tk}$ and ${}^{\scaleto{(2)}{5pt}}\!A^+_{\tk}$ in \eqref{eq:A1} as $\beta_{t_{k}}^{\pm}$ and $\eta_{t_{k}}^{\pm}$, respectively.} First, we will prove the result when the price process is a martingale. In this case, all we need to show is that 
\[
L_{\tk}^{+,*}+L_{\tk}^{-,*} = ({}^{\scaleto{(1)}{5pt}}\!A^+_{\tk}-{}^{\scaleto{(1)}{5pt}}\!A^-_{\tk})I_\tk+({}^{\scaleto{(2)}{5pt}}\!A^+_{\tk}-{}^{\scaleto{(2)}{5pt}}\!A^-_{\tk})+({}^{\scaleto{(3)}{5pt}}\!A^+_{\tk}+{}^{\scaleto{(3)}{5pt}}\!A^-_{\tk}) >0.
\]
First we prove that, under the Condition (2) (Equation \eqref{Cnd2PosSpr}) and Condition (4) of Proposition \ref{prop:admissibility}, $\alpha_\tkk^{1+}=\alpha_\tkk^{1-}$. Indeed, by Eqs.~\eqref{eqn:alpha1p} and \eqref{eqn:alpha1m},
\begin{align*}
\alpha_\tkk^{1+} &=\alpha_\tkk^{1,1}\dfrac{\pi_\tkk(1,1)}{\pi_\tkk^+} + \Ex\big[\alpha_\tkk |\F,\onep=1,\onem=0\big]\left(1-\dfrac{\pi_\tkk(1,1)}{\pi_\tkk^+}\right),\\
\alpha_\tkk^{1-} &=\alpha_\tkk^{1,1}\dfrac{\pi_\tkk(1,1)}{\pi_\tkk^-} + \mathbb{E}\big[\alpha_\tkk |\F,\onep=0,\onem=1\big]\left(1-\dfrac{\pi_\tkk(1,1)}{\pi_\tkk^-}\right).
\end{align*}
Now, recall that $\alpha_\tkk$ is $\Hkkk$-measurable (see Theorem \ref{prop:measurability}), and, by assumption, $\alpha_\tkk$ depends on $\mathbbm{1}_{t_{k+1}}^+$ and $\mathbbm{1}_{t_{k+1}}^-$ only through $\mathbbm{1}_{t_{k+1}}^++\mathbbm{1}_{t_{k+1}}^-$. This means that $\alpha_\tkk=\Phi\big(\mathbbm{1}_{t_{k+1}}^++\mathbbm{1}_{t_{k+1}}^-,\mathbbm{1}_{t_k}^\pm,\ldots,\mathbbm{1}_{t_{k+2-\varpi}}^\pm\big)$ for some function $\Phi$. It follows that
\begin{align*}
\Ex\big[\alpha_\tkk |\F,\onep=1,\onem=0\big]&=\Phi\big(1+0,\mathbbm{1}_{t_k}^\pm,\ldots,\mathbbm{1}_{t_{k+2-\varpi}}^\pm\big)
=\Ex\big[\alpha_\tkk |\F,\onep=0,\onem=1\big].
\end{align*}
Since $\pi_\tkk^+=\pi_\tkk^-$, we then conclude that $\alpha_\tkk^{1+}=\alpha_\tkk^{1-}$.
By doing the same computations as in the set of Equations \ref{eqn:alphaplus}, we can obtain the analogous relations to Equations \eqref{eqn:alpha1p} and \eqref{eqn:alpha1m} for $h_{t_{k+1}}$ and these can be used to show that $h_{t_k}^{1+}=h_{t_k}^{1-}$ in the same way as above.

Let $\alpha_\tkk^{1}:=\alpha_\tkk^{1\pm}$ and $h_\tkk^{1}:=h_\tkk^{1\pm}$.
Since $\pi_\tkk^+=\pi_\tkk^-$, it follows that 
\begin{alignat*}{3}
\beta_\tk^+-\beta_\tk^-&=
{\pi^+_{\tkk}}{\pi^-_{\tkk}}(\alpha^1_{\tkk})^2(\mu_{c}^+\mu_{c^2}^{-}-\mu_{c}^{-}\mu_{c^{2}}^{+})-\alpha^1_{\tkk}{\pi_\tkk(1,1)}\alpha^{1,1}_\tkk\mu_{c}^{-}\mu_{c}^{+}(\pi_{t_{k+1}}^{-}\mu_{c}^{-}-
\pi_{t_{k+1}}^{+}\mu_{c}^{+})&=0.\\
\eta_\tk^+-\eta_\tk^-&=
{\pi^+_{\tkk}}{\pi^-_{\tkk}}\alpha^1_{\tkk}h_\tkk^1(\mu_{c}^+\mu_{c^2}^{-}-\mu_{c}^{-}\mu_{c^{2}}^{+})-h^1_{\tkk}{\pi_\tkk(1,1)}\alpha^{1,1}_\tkk\mu_{c}^{-}\mu_{c}^{+}(\pi_{t_{k+1}}^{-}\mu_{c}^{-}-
\pi_{t_{k+1}}^{+}\mu_{c}^{+})&=0,
\end{alignat*}
implying that ${}^{\scaleto{(1)}{5pt}}\!A^+_{\tk}-{}^{\scaleto{(1)}{5pt}}\!A^-_{\tk}=0$ and ${}^{\scaleto{(2)}{5pt}}\!A^+_{\tk}-{}^{\scaleto{(2)}{5pt}}\!A^-_{\tk}=0$.
Thus, it suffice to show that ${}^{\scaleto{(3)}{5pt}}\!A^+_{\tk}+{}^{\scaleto{(3)}{5pt}}\!A^-_{\tk}>0$.  By  assumption (4) of Proposition \ref{prop:admissibility}, $\muoo^{\pm}=\muo^{\pm}\mu_p^\pm$ and $\muto^{\pm}=\mut^{\pm}\mu_p^\pm$ and, thus, we can write ${}^{\scaleto{(3)}{5pt}}\!A^+_{\tk}+{}^{\scaleto{(3)}{5pt}}\!A^-_{\tk}$ as
\begin{align*}
	{}^{\scaleto{(3)}{5pt}}\!A^+_{\tk}+{}^{\scaleto{(3)}{5pt}}\!A^-_{\tk}=\frac{1}{2\gamma_\tk}\left(\mu_p^+\Gamma^+\big(\mut^+,\mut^-\big)+\mu_p^-\Gamma^-\big(\mut^+,\mut^-\big)\right),
\end{align*}
where 
{\small
\begin{align*}
\Gamma^+(\mathbf{x},\mathbf{y})&= \pi^+_{\tkk}\pi^-_{\tkk}(\alpha^1_{\tkk}\mathbf{y}-\muo^-)(\muo^+-2\alpha^1_{\tkk}\mathbf{x})+2\big[\alpha^{1,1}_{\tkk}{ \pi_\tkk(1,1)}\muo^+\muo^-\big]^2-{ \pi^+_{\tkk}}{ \pi_\tkk(1,1)}\alpha^{1,1}_{\tkk}(\muo^+)^2\muo^-,\\
\Gamma^-(\mathbf{x},\mathbf{y})&=\pi^+_{\tkk}\pi^-_{\tkk}(\alpha^1_{\tkk}\mathbf{x}-\muo^+)(\muo^--2\alpha^1_{\tkk}\mathbf{y})+2\big[\alpha^{1,1}_{\tkk}{ \pi_\tkk(1,1)}\muo^+\muo^-\big]^2-{ \pi^-_{\tkk}}{\pi_\tkk(1,1)}\alpha^{1,1}_{\tkk}(\muo^-)^2\muo^+.
\end{align*}
}
Recall from \eqref{eq:Dneg} that $\gamma_{t_{k}}<0$ and, thus, it remains to show that the numerator $N({}^{\scaleto{(3)}{5pt}}\!A^+_{\tk}+{}^{\scaleto{(3)}{5pt}}\!A^-_{\tk})=\mu_p^+\Gamma^+\big(\mut^+,\mut^-\big)+\mu_p^-\Gamma^-\big(\mut^+,\mut^-\big)$ is also negative. Since $\Gamma^+$ is a linear function of $\mathbf{y}$, $\mut^-\geq(\muo^-)^2$ and, by Lemma~\ref{lemma:alpha}, $\alpha^1_{t_k+1}<0$, we have that
\[
\left.\frac{\partial}{\partial \mathbf{y}}\Gamma^+(\mathbf{x},\mathbf{y})\right|_{\mathbf{x}=\mut^+}=\pi^+_{\tkk}\pi^-_{\tkk}\alpha^1_\tkk(\muo^+-2\alpha^1_\tkk\mut^+)<0.
\]
From here, it follows that for every $\mathbf{x}\in\R$, $\Gamma^+\Big(\mathbf{x},\mut^-\Big)\leq\Gamma^+\Big(\mathbf{x},(\muo^-)^2\Big)$.
Similarly, 
\[
\frac{\partial}{\partial \mathbf{x}}\Gamma^+(\mathbf{x},(\muo^-)^2)=-2\alpha^1_\tkk\pi^+_{\tkk}\pi^-_{\tkk}\big[\alpha^1_\tkk(\muo^-)^2-\muo^-\big]<0,
\]
implying that $\Gamma^+\Big(\mut^+,\mut^-\Big)\leq \Gamma^+\Big((\muo^+)^2,(\muo^-)^2\Big)$. Note 
\begin{align*}
\Gamma^+\Big((\muo^+)^2,(\muo^-)^2\Big)
&=\muo^+\muo^-\bigg\{2\Big[\underbrace{(\alpha^{1,1}_\tkk\pi_\tkk(1,1))^2-(\alpha^1_\tkk)^2\pi^+_\tkk\pi^-_\tkk}_{\text{(I)}}\Big]\muo^+\muo^-+\underbrace{2\alpha^1_\tkk\pi^+_\tkk\pi^-_\tkk\muo^+}_{\text{(II)}}\\
&\qquad\qquad\quad+\pi^+_\tkk\Big[\underbrace{\alpha^1_\tkk\pi^-_\tkk\muo^--\alpha_\tkk^{1,1}\pi_\tkk(1,1)\muo^+}_{\text{(III)}}\Big]-\pi^+_\tkk\pi^-_\tkk\bigg\}.
\end{align*}
By Equation \eqref{eqn:comparison:alpha},
$(\alpha^{1,1}_\tkk\pi_\tkk(1,1))^2\leq(\alpha^1_\tkk)^2\pi^+_\tkk\pi^-_\tkk$,
implying that the term (I) above is negative. Further, by Lemma~\ref{lemma:alpha}, we know that $\alpha_\tkk<0$ and, thus, the term (II) above is also negative.
Finally, by using our assumption \eqref{Cnd1PosSpr} and Eq.~\eqref{eqn:comparison:alpha}, we can conclude that
the term (III) above also negative. It follows that
$\Gamma^+\Big(\mut^+,\mut^-\Big)\leq \Gamma^+\Big((\muo^+)^2,(\muo^-)^2\Big)<0$. Similarly, we can show that $\Gamma^-\Big(\mut^+,(\muo^-)^2\Big)<0$ and conclude that $N{({}^{\scaleto{(3)}{5pt}}\!A^+_{\tk}+{}^{\scaleto{(3)}{5pt}}\!A^-_{\tk})}<0$.

\medskip

For the general midprice dynamics, we need to recall Eqs.~\eqref{eqn:gamma:eta}. 
Then, under our conditions (1)-(4),
 we have trivially that $\rho_k^+=\rho_k^-$ and $\psi_k^+=\psi_k^-$. 
It then follows from Eq.~\eqref{eq:LtildeNMGbb} that
$\widetilde{L}_{\tk}^{+,*} + \widetilde{L}_{\tk}^{-,*}={L}_{\tk}^{+,*}+{L}_{\tk}^{-,*}>0$, as we have already shown. This concludes the general dynamics of the midprice process.
\end{proof}

\begin{proof}[Proof of Corollary \ref{cor:cond2:alt}]
The proof follows along the same lines as the one of Proposition \ref{prop:admissibility}. 
As before, since $\alpha_\tkk$ {depends} on $\mathbbm{1}_{t_{k}}^+$ and $\mathbbm{1}_{t_{k}}^-$ only through $\mathbbm{1}_{t_{k}}^++\mathbbm{1}_{t_{k}}^-$, 
\[
\alpha_\tkk^{1,0}=\Ex\big[\alpha_\tkk |\F,\onep=1,\onem=0\big]=\mathbb{E}\big[\alpha_\tkk |\F,\onep=0,\onem=1\big]=\alpha_\tkk^{0,1},
\]
and, thus, {recalling (\ref{eqn:alpha1p})-(\ref{eqn:alpha1m}),
it} follows that $\alpha_\tkk^{1+}=\alpha_\tkk^{1-}$ when $\pi_\tkk(1,1)=0$. Similarly, $h_\tkk^{1+}=h_\tkk^{1-}$ and we can denote $\alpha_\tkk^{1}:=\alpha_\tkk^{1\pm}$ and $h_\tkk^{1}:=h_\tkk^{1\pm}$.
Since $\pi_\tkk(1,1)=0$, $\mu_c^+=\mu_c^-$ and $\mu_{c^2}^+=\mu_{c^2}^-$, it also follows that $\beta_\tk^+-\beta_\tk^-=0$ and $\eta_\tk^+-\eta_\tk^-=0$, 
implying that ${}^{\scaleto{(1)}{5pt}}\!A^+_{\tk}-{}^{\scaleto{(1)}{5pt}}\!A^-_{\tk}=0$ and ${}^{\scaleto{(2)}{5pt}}\!A^+_{\tk}-{}^{\scaleto{(2)}{5pt}}\!A^-_{\tk}=0$.
We are only left to prove that ${}^{\scaleto{(3)}{5pt}}\!A^+_{\tk}+{}^{\scaleto{(3)}{5pt}}\!A^-_{\tk}>0$, but as the denominator is negative, we only need to prove that its numerator, denoted by $N({}^{\scaleto{(3)}{5pt}}\!A^+_{\tk}+{}^{\scaleto{(3)}{5pt}}\!A^-_{\tk})$, is also negative. However, under our conditions, we have {that
\begin{align*}
N({}^{\scaleto{(3)}{5pt}}\!A^+_{\tk}+{}^{\scaleto{(3)}{5pt}}\!A^-_{\tk})&=\pi^+_{\tkk}\pi^-_{\tkk}\big(\alpha^1_{\tkk}\mu_{c^2}-\mu_c\big)\bigg[\big(\mu_{cp}^+-2\alpha^1_{\tkk}\mu_{c^2p}^+\big) + \big(\mu_{cp}^--2\alpha^1_{\tkk}\mu_{c^2p}^-\big)\bigg],
\end{align*}
which is negative because $\alpha^1_{\tkk}\mu_{c^2}-\mu_c<0$ (since $\alpha^1_{t_{k+1}}\leq 0$), and each of the terms inside the brackets above are positive.} 
For the general midprice dynamics notice that our conditions and the definitions (\ref{eqn:gamma:eta}) readily imply that 
$\rho_k^+=\rho_k^-$ and $\psi_k^+=\psi_k^-=0$,
and the proof follows as in the proof of Proposition \ref{prop:admissibility}.
\end{proof}

\begin{proof}[Proof of Lemma \ref{lemma:gfunc}]
The proof will done by backwards induction. We {only give the details for $\{\alpha_{t_k}\}_{k=1}^{N+1}$  (the proof for $\{h_{t_k}\}_{k=1}^{N+1}$ is very similar). 
To start, note that at time $k=N+1$,  $\alpha_{t_{N+1}}=-\lambda$ and, hence, {it trivially satisfies the condition}. For the inductive step, assume that the lemma holds for $j=k+1$. We will now proceed to prove that $\alpha_{t_k}$ depends on $\pmb{e}_{t_k}^+,\pmb{e}_{t_k}^-$ only thorugh $\pmb{e}_{t_k}^++\pmb{e}_{t_k}^-$.
By Equation \eqref{eqn:alpha:decompos}, $\alpha_\tk=\alpha_{t_{k+1}}^0+N_k/\gamma_{t_k}$
with 
\begin{equation}\label{IJNRTN}
\begin{aligned}
N_k&= (\alpha_\tkk^{1+}\muo^+\pi_\tkk^+)^2 \pi_\tkk^- (\alpha_\tkk^{1-}\mut^- - \muo^-) +(\alpha_\tkk^{1-}\muo^-\pi_\tkk^-)^2 \pi_\tkk^+ (\alpha_\tkk^{1+}\mut^+ - \muo^+)\\
&\quad -2\alpha_\tkk^{1+}\alpha_\tkk^{1-}{\pi^+_{\tkk}}{\pi^-_{\tkk}}\pi_\tkk(1,1)\alpha^{1,1}_{\tkk}(\muo^+\muo^-)^2,\\
\gamma_{t_k}&=\big[{\pi_\tkk(1,1)}\alpha^{1,1}_{\tkk}\muo^+\muo^-\big]^2-{\pi^+_{\tkk}}{ \pi^-_{\tkk}}(\alpha^{1+}_{\tkk}\mut^+-\muo^+)(\alpha^{1-}_{\tkk}\mut^--\muo^-).
\end{aligned}
\end{equation}
By assumption, $\pi_\tkk^+$, $\pi_\tkk^-$, and $\pi_{t_{k+1}}(1,1)$ are functions of $\pmb{e}_{t_k}^++\pmb{e}_{t_k}^-$. Then, we only need to show that $\alpha_\tkk^{0}$, $\alpha_\tkk^{1+}$, $\alpha_\tkk^{1-}$, and $\alpha_\tkk^{1,1}$ have the same property.
By the induction hypothesis, we have that $\alpha_{t_{k+1}}$ is of the form:
\[
\alpha_{t_{k+1}}=\Phi(\pmb{e}_{t_{k+1}}^++\pmb{e}_{t_{k+1}}^-)=\Phi\big(\mathbbm{1}_{t_{k+1}}^++\mathbbm{1}_{t_{k+1}}^-,\ldots,\mathbbm{1}_{t_{k+2-\varpi}}^++\mathbbm{1}_{t_{k+2-\varpi}}^-\big),
\]
for a function $\Phi: \{0,1,2\}^{\varpi}\to\mathbb{R}$. Thus, using \eqref{eqn:prop:cond:exp} and \eqref{eqn:gfunc:pi},
\begin{align*}
\alpha_{t_{k+1}}^0&=\mathbb{E}[\alpha_{t_{k+1}} \,|\,\F]\\
		&= \sum\limits_{u=0}^1\sum\limits_{v=0}^1 \Phi\left(u+v,\mathbbm{1}_{t_k}^++\mathbbm{1}_{t_k}^-\ldots,\mathbbm{1}_{t_{k+2-\varpi}}^++\mathbbm{1}_{t_{k+2-\varpi}}^-\right)\Px\left[\mathbbm{1}_{t_{k+1}}^+=u,\mathbbm{1}_{t_{k+1}}^-=v\,\Big|\,\Hkk\right],
\end{align*}
implying that $\alpha_{t_{k+1}}^0$ is a function of $\pmb{e}_{t_k}^++\pmb{e}_{t_k}^-$ in light of \eqref{eq1234}.
We can similarly deal with $\alpha_\tkk^{1+}$, $\alpha_\tkk^{1-}$, and $\alpha_\tkk^{11}$ using the expressions in (\ref{CndCmpAlph}). For instance,
\begin{align*}
\alpha_\tkk^{1+} 
&=\mathbb{E}[\alpha_{t_{k+1}} \,|\,\F,\mathbbm{1}_{t_{k+1}}^+=1]\\
&=\sum_{\ell\in\{0,1\}}\Phi(1+\ell,\mathbbm{1}_{t_{k}}^++\mathbbm{1}_{t_{k}}^-,\dots,\mathbbm{1}_{t_{k+2-\varpi}}^++\mathbbm{1}_{t_{k+2-\varpi}}^-)\mathbb{P}\big[\mathbbm{1}_{t_{k+1}}^-=\ell\big|\Fj,\mathbbm{1}_{t_{k+1}}^+=1\big]\\
&=\sum_{\ell\in\{0,1\}}\Phi(1+\ell,\mathbbm{1}_{t_{k}}^++\mathbbm{1}_{t_{k}}^-,\dots,\mathbbm{1}_{t_{k+2-\varpi}}^++\mathbbm{1}_{t_{k+2-\varpi}}^-)\frac{\mathbb{P}\big[\mathbbm{1}_{t_{k+1}}^+=1,\mathbbm{1}_{t_{k+1}}^-=\ell\big|\Hjj\big]}{\mathbb{P}\big[\mathbbm{1}_{t_{k+1}}^+=1\big|\Hkk\big]},
\end{align*}
which now is evident that is a function of $\pmb{e}_{t_k}^++\pmb{e}_{t_k}^-$.}
\end{proof}

\subsection{Proofs of Section \ref{sec:invent:analysis}: Inventory Analysis of the optimal strategy}\label{App:proof:inventory}

\begin{proof}[Proof of Proposition \ref{ExpFutInv}]
{For simplicity, we omit `*' when referring to the optimal strategies $L_{\tk}^{\pm,*}$. 
By Eq.~\eqref{eqi11} and Corollary \ref{cor:candidates}, we have that 
\[
I_{t_{k+1}}=I_{t_{k}}-\onep\cp(\Pp-\lp) +\onem\cm(\Pm-\lm)
\]
with $L_{\tk}^{+}={}^{\scaleto{(1)}{5pt}}\!A^+_{\tk}I_\tk+{}^{\scaleto{(2)}{5pt}}\!A^+_{\tk}+{}^{\scaleto{(3)}{5pt}}\!A^+_{\tk}$ and $L_{\tk}^{-}=-\,{}^{\scaleto{(1)}{5pt}}\!A^-_{\tk}I_\tk-{}^{\scaleto{(2)}{5pt}}\!A^-_{\tk}+{}^{\scaleto{(3)}{5pt}}\!A^-_{\tk}$. Under the condition $\pi_{t_{k+1}}(1,1)=0$,
the coefficients simplify to 
\[
	{}^{\scaleto{(1)}{5pt}}\!A^\pm_{\tk}=\frac{\muo^\pm\alpha^{1\pm}_{\tkk}}{\muo^\pm-\alpha^{1\pm}_{\tkk}\mut^\pm},\quad{}^{\scaleto{(2)}{5pt}}\!A^\pm_{\tk}=\frac{h^{1\pm}_{\tkk}}{2}\frac{\muo^\pm}{\muo^\pm-\alpha^{1\pm}_{\tkk}\mut^\pm},\quad
	{}^{\scaleto{(3)}{5pt}}\!A^\pm_{\tk}=\frac{\alpha^{1\pm}_{\tkk}}{2}\frac{\mu_{cp}^\pm-2\alpha^{1\pm}_{\tkk}\mu_{c^2p}^\pm}{\muo^\pm-\alpha^{1\pm}_{\tkk}\mut^\pm}.
\]
As explained in the proof of Corollary \ref{cor:cond2:alt}, ${}^{\scaleto{(1)}{5pt}}\!A^+_{\tk}={}^{\scaleto{(1)}{5pt}}\!A^-_{\tk}$, ${}^{\scaleto{(2)}{5pt}}\!A^+_{\tk}={}^{\scaleto{(2)}{5pt}}\!A^-_{\tk}$, $\alpha_\tkk^{1+}=\alpha_\tkk^{1-}$, and $h_\tkk^{1+}=h_\tkk^{1-}$. In particular, we also have ${}^{\scaleto{(3)}{5pt}}\!A^+_{\tk}={}^{\scaleto{(3)}{5pt}}\!A^-_{\tk}$. Hereafter, we omit $\pm$ in the superscripts of ${}^{\scaleto{(\ell)}{5pt}}\!A^\pm_{\tk}$, $\alpha_\tkk^{1\pm}$, $h_\tkk^{1\pm}$,  etc. Furthermore, since $\pi_{t_{k+1}}=\pi^+_{t_{k+1}}=\pi^{-}_{t_{k+1}}$, we can see that \eqref{eqn:htk} simplifies to 
	\begin{align*}\nonumber
	h_{\tk}&= h^0_{\tkk}+2\pi_\tkk \Bigg\{2\Big(\alpha^{1}_{\tkk}\mut-\muo\Big)\;{}^{\scaleto{(1)}{5pt}}\!A_{\tk}{}^{\scaleto{(2)}{5pt}}\!A_{\tk}+2\alpha^{1}_{\tkk}\muo{}^{\scaleto{(2)}{5pt}}\!A^\delta_{\tk}+{}^{\scaleto{(1)}{5pt}}\!A_{\tk} h^{1}_{\tkk}\mu_c\Bigg\}\\
	&= h^0_{\tkk}+2\pi_\tkk \frac{\alpha^{1}_{\tkk}\muo^2}{\muo-\alpha^{1}_{\tkk}\mut}h^{1}_{\tkk},
	\end{align*}  
which, combined with the condition $h_{t_{N+1}}=0$, implies that $h_{\tk}\equiv 0$ and, thus, ${}^{\scaleto{(2)}{5pt}}\!A^+_{\tk}\equiv 0$, for all $k$.

We can then write
\[
	I_{t_{k+1}}=\xi_{t_{k+1}}+\Big(1+\eta_{t_{k+1}}{}^{\scaleto{(1)}{5pt}}\!A_{\tk}\Big)I_{t_{k}},
\]
where 
\begin{align*}
	\eta_{t_{k+1}}&:=\onep\cp+\onem\cm,\\
	\xi_{t_{k+1}}&:=(\onep\cp-\onem\cm){}^{\scaleto{(3)}{5pt}}\!A_{\tk}-(\onep\cp\Pp-\onem\cm \Pm).
\end{align*}
Since $\mathbb{E}[\xi_{t_{k+1}}|\mathcal{F}_{t_{k}}]=0$ and $\mathbb{E}[\eta_{t_{k+1}}|\mathcal{F}_{t_{k}}]=2\pi_{t_{k+1}}\mu_c$, we can then conclude \eqref{Krel1a}-(i).

For \eqref{Krel1a}-(ii), first note that, by simple induction, we have that, when $I_0=0$, 
\begin{align}\label{ReprsnInv}
	I_{t_{k+1}}=\sum_{j=1}^{k+1}\xi_{t_{j}}\prod_{\ell=j}^{k}(1+
	\eta_{t_{\ell+1}}{}^{\scaleto{(1)}{5pt}}\!A_{t_{\ell}}),
\end{align}
under the usual convention that $\prod_{\ell=k+1}^{k}=1$. Next, we claim that, for any $i$, there is a function $\Phi_i:\{0,1,2\}^{\varpi}\to\mathbb{R}$ such that 
\begin{align}\label{ICEWi}
	\mathbb{E}\left[\left.\prod_{\ell=i}^{k}(1+
	\eta_{t_{\ell+1}}{}^{\scaleto{(1)}{5pt}}\!A_{t_{\ell}})\right|\mathcal{F}_{t_{i}}\right]=\Phi_i\left(\pmb{e}_{t_i}^++\pmb{e}_{t_i}^-\right).
\end{align}
Indeed, as shown in the proof of Lemma \ref{lemma:gfunc}, $\alpha_\tkk^{1\pm}$ are functions of $\pmb{e}_{t_k}^++\pmb{e}_{t_k}^-$ and, thus, ${}^{\scaleto{(1)}{5pt}}\!A_{t_{k}}=\Psi_k(\pmb{e}_{t_k}^++\pmb{e}_{t_k}^-)$, for some function $\Psi_k$. Then, by conditioning on $\mathbbm{1}_{t_{k+1}}^+$ and $\mathbbm{1}_{t_{k+1}}^-$,
\[
	\mathbb{E}\left[\left.(1+
	\eta_{t_{k+1}}{}^{\scaleto{(1)}{5pt}}\!A_{t_{k}})\right|\mathcal{F}_{t_{k}}\right]=
	1+2\mu_c\pi_{t_{k+1}}\Psi_k(\pmb{e}_{t_k}^++\pmb{e}_{t_k}^-)=:\Phi_{k}\left(\pmb{e}_{t_k}^++\pmb{e}_{t_k}^-\right). 
\]
Suppose \eqref{ICEWi} holds for $i=j+1$. Then,
\begin{align*}
	\mathbb{E}\left[\left.\prod_{\ell=j}^{k}(1+
	\eta_{t_{\ell+1}}{}^{\scaleto{(1)}{5pt}}\!A_{t_{\ell}})\right|\mathcal{F}_{t_{j}}\right]&=\mathbb{E}\left[\left.(1+
	\eta_{t_{j+1}}{}^{\scaleto{(1)}{5pt}}\!A_{t_{j}})\mathbb{E}\left[\left.\prod_{\ell=j+1}^{k}(1+
	\eta_{t_{\ell+1}}{}^{\scaleto{(1)}{5pt}}\!A_{t_{\ell}})\right|\mathcal{F}_{t_{j+1}}\right]\right|\mathcal{F}_{t_{j}}\right]\\
	&=\mathbb{E}\left[\left.(1+
	\eta_{t_{j+1}}{}^{\scaleto{(1)}{5pt}}\!A_{t_{j}})\Phi_{j+1}\left(\pmb{e}_{t_{j+1}}^++\pmb{e}_{t_{j+1}}^-\right)\right|\mathcal{F}_{t_{j}}\right]\\
	&=2(1+2\mu_c\Psi_j(\pmb{e}_{t_j}^++\pmb{e}_{t_j}^-))
	\Phi_{j+1}(1,\mathbbm{1}_{t_{j}}^++\mathbbm{1}_{t_{j}}^-,\dots,\mathbbm{1}_{t_{j+2-\varpi}}^++\mathbbm{1}_{t_{j+2-\varpi}}^-)\pi_{t_{j+1}}\\
	&\quad +
	\Phi_{j+1}(0,\mathbbm{1}_{t_{j}}^++\mathbbm{1}_{t_{j}}^-,\dots,\mathbbm{1}_{t_{j+2-\varpi}}^++\mathbbm{1}_{t_{j+2-\varpi}}^-)(1-2\pi_{t_{j+1}}),
\end{align*}
which is clearly of the form $\Phi_j\left(\pmb{e}_{t_j}^++\pmb{e}_{t_j}^-\right)$. This shows the validity of \eqref{ICEWi}. 

We are now ready to show the result:
\begin{align*}
	&\mathbb{E}\left[\left.\xi_{t_{j}}\prod_{\ell=j}^{k}(1+
	\eta_{t_{\ell+1}}{}^{\scaleto{(1)}{5pt}}\!A_{t_{\ell}})\right|\mathcal{F}_{t_{j-1}}\right]\\
	&=\mathbb{E}\left[\left.\xi_{t_{j}}\mathbb{E}\left[\left.\prod_{\ell=j}^{k}(1+
	\eta_{t_{\ell+1}}{}^{\scaleto{(1)}{5pt}}\!A_{t_{\ell}})\right|\mathcal{F}_{t_{j}}\right]\right|\mathcal{F}_{t_{j-1}}\right]\\
	&=\mathbb{E}\left[\left.\xi_{t_{j}}\Phi_{j}\left(\pmb{e}_{t_{j}}^++\pmb{e}_{t_{j}}^-\right)\right|\mathcal{F}_{t_{j-1}}\right]\\
	&=\left(\mathbb{E}\left[\left.\xi_{t_{j}}\right|\mathcal{F}_{t_{j-1}},\mathbbm{1}_{t_{j}}^+=1,\mathbbm{1}_{t_{j}}^-=0\right]
	+\mathbb{E}\left[\left.\xi_{t_{j}}\right|\mathcal{F}_{t_{j-1}},\mathbbm{1}_{t_{j}}^+=0,\mathbbm{1}_{t_{j}}^-=1\right]\right)\\
	&\quad \quad \times
	\Phi_{j}(1,\mathbbm{1}_{t_{j-1}}^++\mathbbm{1}_{t_{j-1}}^-,\dots,\mathbbm{1}_{t_{j+1-\varpi}}^++\mathbbm{1}_{t_{j+1-\varpi}}^-)\pi_{t_{j}}\\
	&\quad +\mathbb{E}\left[\left.\xi_{t_{j}}\right|\mathcal{F}_{t_{j-1}},\mathbbm{1}_{t_{j}}^+=0,\mathbbm{1}_{t_{j}}^-=0\right]
	\Phi_{j}(0,\mathbbm{1}_{t_{j-1}}^++\mathbbm{1}_{t_{j-1}}^-,\dots,\mathbbm{1}_{t_{j+1-\varpi}}^++\mathbbm{1}_{t_{j+1-\varpi}}^-)(1-2\pi_{t_{j+1}}),
\end{align*}
which, recalling that $\xi_{t_{j}}:=(\mathbbm{1}_{t_{j}}^+c_{t_{j}}^+-\mathbbm{1}_{t_{j}}^-c_{t_{j}}^-){}^{\scaleto{(3)}{5pt}}\!A_{t_{j-1}}-(\mathbbm{1}_{t_{j}}^+c_{t_{j}}^+p_{t_{j}}^+-\mathbbm{1}_{t_{j}}^-c_{t_{j}}^-p_{t_{j}}^-)$, equals to $0$ because $\mathbb{E}\left[\left.\xi_{t_{j}}\right|\mathcal{F}_{t_{j-1}},\mathbbm{1}_{t_{j}}^+=0,\mathbbm{1}_{t_{j}}^-=0\right]=0$ and 
\begin{align*}
\mathbb{E}\left[\left.\xi_{t_{j}}\right|\mathcal{F}_{t_{j-1}},\mathbbm{1}_{t_{j}}^+=1,\mathbbm{1}_{t_{j}}^-=0\right]
&=\mathbb{E}\left[\left.c^+_{t_j}\right|\mathcal{F}_{t_{j-1}},\mathbbm{1}_{t_{j}}^+=1\right]{}^{\scaleto{(3)}{5pt}}\!A_{t_{j-1}}-\mathbb{E}\left[\left.c^+_{t_j}p^+_{t_j}\right|\mathcal{F}_{t_{j-1}},\mathbbm{1}_{t_{j}}^+=1\right]\\
&=\mu_{c}{}^{\scaleto{(3)}{5pt}}\!A_{t_{j-1}}-\mu_{cp}\\
\mathbb{E}\left[\left.\xi_{t_{j}}\right|\mathcal{F}_{t_{j-1}},\mathbbm{1}_{t_{j}}^+=0,\mathbbm{1}_{t_{j}}^-=1\right]
&=-\mathbb{E}\left[\left.c^-_{t_j}\right|\mathcal{F}_{t_{j-1}},\mathbbm{1}_{t_{j}}^-=1\right]{}^{\scaleto{(3)}{5pt}}\!A_{t_{j-1}}+\mathbb{E}\left[\left.c^-_{t_j}p^-_{t_j}\right|\mathcal{F}_{t_{j-1}},\mathbbm{1}_{t_{j}}^-=1\right]\\
&=-\mu_{c}{}^{\scaleto{(3)}{5pt}}\!A_{t_{j-1}}+\mu_{cp}.
\end{align*}
We then conclude that $\mathbb{E}\left[\xi_{t_{j}}\prod_{\ell=j}^{k}(1+
	\eta_{t_{\ell+1}}{}^{\scaleto{(1)}{5pt}}\!A_{t_{\ell}})\right]=0$ and, thus, in light of the representation \eqref{ReprsnInv}, we conclude that $\mathbb{E}\left[I_{t_{k+1}}\right]=0$, for any $k$.} 
\end{proof}

\subsection{Proofs of Section \ref{sec:invent:runningpen}: Running inventory penalization}\label{App:run:inv}
\begin{proof}[Proof of Theorem \ref{lemma:runn:penalty}]
Plugging the ansatz \eqref{eqn:ansatzPhi} into \refeq{eqn:new:DPP}, we arrive to the equation 
\begin{equation}\label{eq:new:recursion}
\begin{aligned}
&W_\tk+{}^{\scaleto{(\phi)}{5pt}}\!\alpha_{\tk}I_\tk^2+\sk I_\tk+{}^{\scaleto{(\phi)}{5pt}}\!h_{\tk}I_\tk +{}^{\scaleto{(\phi)}{5pt}}\!g_{\tk}\\
&\quad =\sup_{L^\pm_{t_k} \in \mathcal{A}_{t_k}} \mathbb{E}\Big[W_\tkk+{}^{\scaleto{(\phi)}{5pt}}\!\alpha_{\tkk}I_\tkk^2+\skk I_\tkk +{\Blue {}^{\scaleto{(\phi)}{5pt}}\!h_{\tkk}I_\tkk} +{}^{\scaleto{(\phi)}{5pt}}\!g_{\tkk}-\phi I_{t_{k+1}}^2\,\Big|\,\F\Big].\\
&\quad =\sup_{L^\pm_{t_k} \in \mathcal{A}_{t_k}} \mathbb{E}\Big[W_\tkk+({}^{\scaleto{(\phi)}{5pt}}\!\alpha_{\tkk}-\phi)I_\tkk^2+\skk I_\tkk +{\Blue {}^{\scaleto{(\phi)}{5pt}}\!h_{\tkk}I_\tkk} +{}^{\scaleto{(\phi)}{5pt}}\!g_{\tkk}\,\Big|\,\F\Big].
\end{aligned}
\end{equation}
The above expression has the same structure as \eqref{eq113}, but with $\alpha_{\tkk}$ replaced by ${}^{\scaleto{(\phi)}{5pt}}\!\alpha_{\tkk}-\phi$. This is the reason why the results of Theorem \ref{prop:measurability} and Corollary \ref{cor:candidates} follow along the same steps as in Theorem \ref{prop:measurability} and Corollary \ref{cor:candidates} with the term $\alpha_{\tkk}$ (and its variants) {\Blue replaced} by ${}^{\scaleto{(\phi)}{5pt}}\!\alpha_{\tkk}-\phi$. The proof of Lemma \ref{lemma:alpha} also follows though it needs a bit more of care. 
Indeed, this proof is heavily based on the inequalities \eqref{eqn:comparison:alpha}, \eqref{ineq:alpha0}, \eqref{eqn:bounds:alpha11pi11}, which still hold for ${}^{\scaleto{(\phi)}{5pt}}\!\alpha_{\tkk}^{1\pm}$ and ${}^{\scaleto{(\phi)}{5pt}}\!\alpha_{\tkk}^{1,1}$, but we need them to hold for ${}^{\scaleto{(\phi)}{5pt}}\!\alpha_{\tkk}^{1\pm}-\phi$ and ${}^{\scaleto{(\phi)}{5pt}}\!\alpha_{\tkk}^{1,1}-\phi$. Indeed, for \eqref{eqn:comparison:alpha}, since $ \pi_\tkk(1,1)\leq \pi_\tkk^\pm$, then $ -\phi\pi_\tkk^\pm\leq -\phi\pi_\tkk(1,1)$, which together with ${}^{\scaleto{(\phi)}{5pt}}\!\alpha_\tkk^{1\pm} \pi_\tkk^\pm  \leq {}^{\scaleto{(\phi)}{5pt}}\!\alpha_\tkk^{1,1} \pi_\tkk(1,1)$ implies $({}^{\scaleto{(\phi)}{5pt}}\!\alpha_\tkk^{1\pm}-\phi) \pi_\tkk^+  \leq ({}^{\scaleto{(\phi)}{5pt}}\!\alpha_\tkk^{1,1}-\phi) \pi_\tkk(1,1)$. The proof of \eqref{ineq:alpha0} is similar. For the second inequality in \eqref{eqn:bounds:alpha11pi11}, since $\pi_\tkk^+ + \pi_\tkk^- - \pi_\tkk(1,1)\leq 1$, we have $-\phi\pi_\tkk(1,1)\leq{}-\phi\pi_\tkk^+ -\phi \pi_\tkk^- +\phi$, which, together with ${}^{\scaleto{(\phi)}{5pt}}\!\alpha_\tkk^{1,1}\pi_\tkk(1,1) \leq {}^{\scaleto{(\phi)}{5pt}}\!\alpha_\tkk^{1+}\pi_\tkk^+ + {}^{\scaleto{(\phi)}{5pt}}\!\alpha_\tkk^{1-}\pi_\tkk^- - \alpha_\tkk^0$, implies 
\begin{equation*}
({}^{\scaleto{(\phi)}{5pt}}\!\alpha_\tkk^{1,1}-\phi)\pi_\tkk(1,1) \leq ({}^{\scaleto{(\phi)}{5pt}}\!\alpha_\tkk^{1+}-\phi)\pi_\tkk^+ + ({}^{\scaleto{(\phi)}{5pt}}\!\alpha_\tkk^{1-}-\phi)\pi_\tkk^- - (\alpha_\tkk^0-\phi).
\end{equation*}
The rest of the proof follows the same arguments as in the proof of Lemma \ref{lemma:alpha}.
\end{proof}

\bibliographystyle{jtbnew}



\end{document}